\documentclass[twoside,11pt]{article}

\usepackage{graphicx, subfigure}
\usepackage[top=1in,bottom=1in,left=1in,right=1in]{geometry}
\usepackage[numbers,sort&compress]{natbib} 
\usepackage{amsfonts, amsmath, amssymb, amsthm, bbm}
\usepackage[normalem]{ulem}
\usepackage{comment}
\usepackage{eufrak}

\usepackage{color}
\definecolor{darkred}{RGB}{100,0,0}
\definecolor{darkgreen}{RGB}{0,100,0}
\definecolor{darkblue}{RGB}{0,0,150}

\usepackage{hyperref}
\hypersetup{colorlinks=true, linkcolor=darkred, citecolor=darkgreen, urlcolor=darkblue}
\usepackage{url}

\newtheorem{thm}{Theorem}[section]
\newtheorem{prp}[thm]{Proposition}
\newtheorem{lem}[thm]{Lemma}
\newtheorem{cor}[thm]{Corollary}
\newtheorem{rem}{Remark}[section]

\def\beq{\begin{equation}}
\def\eeq{\end{equation}}
\def\beqn{\begin{eqnarray*}}
\def\eeqn{\end{eqnarray*}}
\def\bitem{\begin{itemize}}
\def\eitem{\end{itemize}}
\def\benum{\begin{enumerate}}
\def\eenum{\end{enumerate}}
\def\bmult{\begin{multline*}}
\def\emult{\end{multline*}}
\def\bcenter{\begin{center}}
\def\ecenter{\end{center}}

\DeclareMathOperator*{\argmin}{arg\, min}

\def\cA{\mathcal{A}}
\def\cB{\mathcal{B}}
\def\cC{\mathcal{C}}

\def\cJ{\mathcal{J}}
\def\bbK{\bbK}

\def\cN{\mathcal{N}}

\def\cU{\mathcal{U}}
\def\cV{\mathcal{V}}

\def\bA{\mathbf{A}}
\def\bB{\mathbf{B}}

\def\bD{\mathbf{D}}

\def\bI{\mathbf{I}}

\def\bO{\mathbf{O}}

\def\bU{\mathbf{U}}
\def\bV{\mathbf{V}}
\def\bW{\mathbf{W}}
\def\bX{\mathbf{X}}
\def\bY{\mathbf{Y}}
\def\bZ{\mathbf{Z}}

\def\bi{\mathbf{i}}

\def\bx{\mathbf{x}}

\def\1{{\mathbf 1}}

\newcommand\bSigma{{\boldsymbol\Sigma}}
\newcommand\bOmega{{\boldsymbol\Omega}}
\newcommand\bGamma{{\boldsymbol\Gamma}}

\def\bbB{\mathbb{B}}

\def\bbE{\mathbb{E}}

\def\bbK{\mathbb{K}}

\def\bbR{\mathbb{R}}

\newcommand{\E}{\operatorname{\mathbb{E}}}
\renewcommand{\P}{\operatorname{\mathbb{P}}}

\newcommand{\var}[1]{\operatorname{Var}\left(#1\right)}

\title{Adaptive estimation of High-Dimensional Signal-to-Noise Ratios}
\author{Nicolas Verzelen\footnote{(corresponding author) INRA, UMR 729 MISTEA, F-34060 Montpellier, FRANCE}~\ and Elisabeth Gassiat\footnote{Laboratoire de Math\'ematiques d'Orsay, Univ. Paris-Sud, CNRS, Universit\'e Paris-Saclay, 91405 Orsay, FRANCE.}}
\begin{document}

\date{}
\maketitle

\begin{abstract}
We consider the equivalent problems of estimating the residual variance, the proportion of explained variance $\eta$ and the signal strength in a high-dimensional linear regression model with Gaussian random design. Our aim is to understand the impact  of  not knowing the sparsity of the regression parameter and not knowing the distribution of the design on minimax estimation rates of $\eta$. Depending on the sparsity $k$ of the regression parameter, optimal estimators of $\eta$ either rely on estimating the regression parameter or are based on $U$-type statistics, and have minimax rates depending on $k$. In the important situation where $k$ is unknown, we build an adaptive procedure whose convergence rate simultaneously achieves the minimax risk over all $k$ up to a logarithmic loss which we prove to be non avoidable. Finally, the knowledge of the design distribution is shown to play a critical role. When the  distribution of the design is unknown, consistent estimation of explained variance is indeed possible in much narrower regimes than for known design distribution.
\end{abstract}

\section{Introduction}

\subsection{Motivations}
In this paper, we investigate the estimation of the proportion of explained variation in high-dimensional linear models with random design, that is the ratio of the variance of the signal to the total amount of variance of the observation. Although this question is of great importance in many applications where the aim is to quantify to what extent covariates explain the variation of the response variable, our analysis  is mainly motivated by problems of heritability estimation. In such studies, the response variable is a phenotype measured on $n$ individuals and the predictors are genetic markers on each of these individuals. Then, heritability corresponds to the proportion of phenotypic variance which can be explained by genetic factors. Usually, the   number of predictors $p$ greatly exceeds the number $n$ of individuals. 
When the phenotype under investigation can be explained by a small number of genetic factors, the corresponding regression parameter is sparse, and methods exploiting sparsity are of utmost interest. 
It appeared recently in biological studies that, for some  complex human traits, there was a huge gap (which has been called  the ``dark matter'' of the genome) between the genetic variance explained by populations studies and the one obtained by genome wide associations studies (GWAS), see \cite {maher:2008}, \cite{stein:2012} or \cite{goldstein:2009}. To explain this gap, it has been hypothesized that some traits might be ``highly polygenic'',  meaning that genetic factors explaining the phenotype could be so numerous that the corresponding regression parameter  may no anymore considered to be sparse. This may be the case for instance when psychiatric disorders are associated to neuroanatomical changes as in \cite{amaral:2008} or \cite{steen:2006}, see also \cite{toro:2014}.
 As a consequence,  sparsity-based methods would be questionable in this situation.
When the researcher faces the data, she does not know in general the proportion of relevant predictors, that is the level of sparsity of the parameter.  In this work, our first aim is to understand the impact  of the ignorance of the sparsity level on heritability estimation. 
Another important feature of the model when estimating proportion of explained variation is the covariance matrix of the predictors. There is a long standing gap between estimation procedures that assume the knowledge of this covariance (e.g.\  \cite{bonnet,janson2015eigenprism})
(which mathematically is the same as assuming that the covariance is the identity matrix) 
and practical situations where it is generally unknown.  Our second aim is to evaluate the impact of the ignorance of the covariance matrix on heritability estimation.\\

 \noindent
 To be more specific, 
 consider the random design high-dimensional linear model 
\beq\label{eq:linear_model}
y_i = \bx_i \beta^* + \epsilon_i\ , \quad i=1,\ldots, n
\eeq
where  $y_i,\epsilon_i \in \bbR$, $\beta^*\in \bbR^p$, $i=1,\ldots,n$, and $\bX=\left(\begin{array}{l}\bx_{1}\\ \vdots \\ \bx_{n}\end{array} \right)\in \bbR^{n\times p}$. We assume that the noise $\epsilon=(\epsilon_{1},\ldots,\epsilon_{n})^{T}$ and the  the lines $\bx_i$, $i=1,\ldots,n$, of $\bX$ are independent random variables. We also assume that the $\epsilon_{i}$, $i=1,\ldots,n$, are independent and identically distributed (i.i.d.) with distribution  $\cN(0,\sigma^2)$, and that  the lines $\bx_i$, $i=1,\ldots,n$, of $\bX$ are also i.i.d. with distribution $\cN(0, \bSigma)$.
%
Throughout the paper, the covariance matrix $\bSigma$ is assumed to be invertible and the noise level $\sigma$ is unknown (the case of known noise level is evoked in the discussion section).
 Our general objective is the  optimal estimation of the signal-to-noise ratio
\beq\label{eq:snr}
\theta := \frac{\E\left[\|\bx_1^T\beta^*\|_2^2\right]}{\sigma^2}= \frac{\|\bSigma^{1/2}\beta^*\|_2^2}{\sigma^2}\ , 
\eeq
or equivalently the proportion of explained variation
\beq\label{eq:R^2}
\eta=\eta(\beta^*,\sigma) := \frac{\E[\|\bx_1^T\beta^*\|_2^2]}{\var{y_1}}= \frac{\theta}{1+\theta}\ 
\eeq
when the vector $\beta^*$ is unknown and possibly sparse. In the sequel,  $\beta^*$ is said to be $k$-sparse, when at most $k$ coordinates of $\beta^*$ are non-zero.
\\

Note that estimating $\eta$ amounts to decipher the signal strength from the noise level in $\var{y_1}= \sigma^2+ \|\bSigma^{1/2}\beta^*\|_2^2$. Since $\|Y\|_2^2/\var{y_1}$ follows a $\chi^2$ distribution with $n$ degrees of freedom, it follows that  $\|Y\|_2^2/n = \var{y_1}[1+ O_P(n^{-1/2})]$ and it is therefore almost equivalent (up to a parametric $n^{-1/2}$ loss) to estimate the proportion of explained variation $\eta$, the quadratic function $\beta^{*T}\bSigma\beta^*$ or the noise level $\sigma^2$. For the sake of presentation, we mostly express our results in terms of the estimation of $\eta$, but they can be easily extended to the signal strength or to the noise estimation problems.

\subsection{Main results}

There are two main lines of research for estimating $\sigma$ or $\eta$ in a high-dimensional setting. Under the assumption that $\beta^*$ is $k$-sparse with some small $k$, it has been established that $\beta^*$ can be estimated at a fast rate (roughly $\sqrt{k\log p /n}$) using for instance Lasso-type procedures, so that using an adequate plug-in method one could hope to estimate $\eta$ well. Following this general approach, some authors have obtained $k\log(p)/n$-consistent~\cite{2012_Sun} and $\sqrt{1/n}$-consistent~\cite{bayati2013estimating,fan_variance} estimators of $\sigma$  in some specific regimes. When $\beta^*$ is  dense (that is when many coordinates of $\beta^*$ are nonzero), such approaches fail. In this regime, a $U$-type estimator~\cite{dicker_variance}  has been proved to achieve consistency at the rate $\sqrt{p}/n$. However, its optimality  has never been assessed.

\bigskip 

Our first main contribution is the proof that the adaptation to unknown sparsity is indeed possible when $\bSigma$ is known, but at the price of a $\sqrt{\log(p)}$ loss factor in the convergence rate when $\beta^*$ is dense.  The idea is the following. Let  $\widehat{\eta}^{D}(\bSigma^{-1})$ be a  $U$-type estimator which is  $\sqrt{p}/n$-consistent, the true parameter $\beta^*$ being sparse or not. We shall denote it  the dense estimator.
Let also  $\widehat{\eta}^{SL}$ be a  $k\log(p)/n$-consistent estimator when $\beta^*$ is $k$-sparse for some small $k$.  Then, if the real  $\beta^*$ is sparse, both estimators should be fairly accurate and should give similar answers, and if the real $\beta^*$ is dense, or not sparse enough, then $\widehat{\eta}^{SL}$ will be quite wrong and will give an answer slightly different from the dense estimator. 
Therefore, the idea is to choose the sparse estimator  $\widehat{\eta}^{SL}$ when both estimators are close enough, so that the quickly convergence rate is obtained when the unknown sparsity $k$ is small, and to choose the dense estimator when both estimators are not close, in which case the slower rate is attained which is appropriate in the dense regime. Such a procedure should adapt well to unknown sparsity.
Now, to be able to give a precise definition of the estimator, that is to set what ``close enough'' quantitatively means, one needs a precise understanding of the behavior of the dense and of the sparse estimators. 
Thus as  a first and preliminary step, we obtain a deviation inequalities for the dense estimator, see Theorem \ref{prp:risk_dense}. We also establish the minimax estimation risk of $\eta$ as a function of ($k$, $n$, $p$) when the parameter $\beta^*$ is  $k$-sparse (see Table \ref{tab:minimax} below) and when $\bSigma$ is known, thereby assessing that Dicker's procedure~\cite{dicker_variance} is optimal in the dense regime ($k\geq \sqrt{p}$) and an estimator based on the square-root Lasso~\cite{2012_Sun} is near optimal in the  sparse regime ($k\leq \sqrt{p}$).   
Again for known $\bSigma$, we finally construct a data-driven combination of $\widehat{\eta}^{D}(\bSigma^{-1})$ (the dense estimator) and $\widehat{\eta}^{SL}$ (the sparse estimator) following the idea explained before. We prove that such a procedure is indeed adaptive to unknown sparsity, see Theorem \ref{prp:adaptation}, and that it achieves the minimax adaptive rate with a $\sqrt{\log(p)}$ loss factor compared to the non adaptive minimax rate. This logarithmic term is proved to be unavoidable, see Proposition \ref{prp:lower_adaptation}.

\bigskip 
Our second main contribution is an analysis of the proportion of explained variance estimation problem  under unknown $\bSigma$.
The construction of dense estimators such as $\widehat{\eta}^{D}(\bSigma^{-1})$ 
requires the knowledge of the covariance matrix $\bSigma$. 
But in many practical situations, the covariance structure of the covariates is unknown. For unknown $\bSigma$, there are basically two main situations:
\begin{itemize}
\item Under sufficiently strong structural assumptions on $\bSigma$ so that $\bSigma^{-1}$ can be estimated at the rate $\sqrt{p}/n$ in operator norm, a simple plug-in method allows to build a minimax and an adaptive minimax procedure 
with the same rates as when $\bSigma$ is known, see Corollary \ref{cor:plug_in}. 
\item Our main result is that, for a general covariance matrix $\bSigma$, it is basically impossible to build a consistent estimator of $\eta$ when $k$ is much larger than $n$; see Theorem \ref{thrm:lower_minimax_unknown_variance} and its comments for a precise statement. This is in sharp contrast with the situation where $\bSigma$ is known, for which the problem of estimating $\eta$ can be handled in regimes where $\beta^*$ is impossible to estimate (e.g.\  $k=p$ and $p= n^{1+\kappa}$ with $\kappa\in (0,1)$ as depicted in Table \ref{tab:minimax}). For unknown and arbitrary $\bSigma$, the range of $(k,n,p)$ for which $\eta$ can be consistently estimated seems to be roughly the same as for estimating $\beta^*$, suggesting that signal estimation ($\beta^*$) is nearly as difficult as signal strength estimation ($\beta^{*T}\bSigma\beta^*$). This impossibility result unveils that, in the high-dimensional dense case, the knowledge of the covariance matrix is fundamental and one cannot extend known procedures such as~\cite{dicker_variance,dicker2016maximum} or  $\widehat{\eta}^{D}(\bSigma^{-1})$ to this unknown variance setting.
\end{itemize}

\begin{table}[h]
\caption{Optimal estimation risk $\E[(\widehat{\eta}-\eta)^2]$ when $\beta^*$ is $k$-sparse and $\bSigma$ is known. 
Here,  $a\in(0,1/2)$ is any arbitrarily small constant and  it is assumed below that $n\leq  p\leq n^2$. 
The results remain valid for $p\geq n^2$ if we replace the quantities $\tfrac{k^2\log^2(p)}{n^2}$ and $\tfrac{p}{n^2}$ by   $ \tfrac{k^2\log^2(p)}{n^2}\wedge 1$ and $\tfrac{p}{n^2}\wedge 1$, respectively.
}
\bigskip
\label{tab:minimax}
\centering
\begin{tabular}{|c||c|c|}
\hline
{\sc Sparsity regimes} & {\sc Minimax risk} & {\sc Near-optimal procedure} \\ \hline \hline 
$\displaystyle k\leq \tfrac{\sqrt{n}}{\log(p)}$ &   $\displaystyle  \tfrac{1}{n}$ &  {\sc square-root Lasso estimator $\widehat{\eta}^{SL}$ \eqref{eq:definition_eta_Lasso}} \\ \hline 
$\displaystyle \tfrac{\sqrt{n}}{\log(p)} \leq k\leq p^{1/2-a}$ &   $\displaystyle \tfrac{k^2\log^2(p)}{n^2}$ &  {\sc square-root Lasso estimator $\widehat{\eta}^{SL}$ \eqref{eq:definition_eta_Lasso}} \\ \hline 
$\displaystyle k\geq \sqrt{p}$ &$\displaystyle \tfrac{p}{n^2}$&  {\sc Dense estimator $\widehat{\eta}^{D}(\bSigma^{-1})$~\eqref{eq:etadense}} (see also ~\cite{dicker_variance}) \\ \hline
\end{tabular}
\end{table}

\subsection{Related work}


The literature on minimax estimation of quadratic functionals initiated in~\cite{donoho_nussbaum} is rather extensive~(see e.g.\  \cite{MR2253108,Laurent00}). 
In the Gaussian sequence model, that is $n=p$ and $\bX=\bI_p$, Collier et al~\cite{collier2015minimax} have derived the minimax estimation rate of the functional $\|\beta^*\|_2^2$ for $k$-sparse vector $\beta^*$ when the noise level $\sigma$ is known. However, we are not aware of any minimax result in the high-dimensional linear model even under known noise level.\\

Another problem related to the estimation of the quadratic functional $\beta^{*T}\bSigma \beta^*$ is signal detection, which aims at testing the null hypothesis $H_0$:``$\beta^*=0$'' versus $H_{1,k}[r]$: ``$\|\bSigma^{1/2}\beta^*\|_2^2\geq r\text{ and } |\beta^*|_{0}\leq k$'' (where $|\beta^*|_{0}$ denotes the number of non nul coordinates of $\beta^*$). The minimax separation distance is then the smallest $r$ such that a test of $H_0$ vs $H_{1,r}$ is able to achieve small type I and type II error probabilities. This minimax separation distance is somewhat analogous to a local minimax estimation risk of $\|\bSigma^{1/2}\beta^*\|_2^2$ around $\beta^*=0$. In the Gaussian sequence model, minimax separation distances haven been studied in \cite{baraud_minimax,ingster_suslina}. These results have  been extended to the high-dimensional linear model under both known~\cite{2010_EJS_Ingster,2011_AS_Arias-Castro} and unknown~\cite{2010_EJS_Ingster,2010_AS_Verzelen} noise level. Our first minimax lower bound (Proposition \ref{prp:lower_minimax}) is largely inspired from these earlier contributions, but the minimax lower bounds for adaptation problems require more elaborate argument. In particular, the proof of Theorem \ref{thrm:lower_minimax_unknown_variance} is largely based on new ideas. \\

Recent works have been devoted to the adaptive estimation of sparse parameters $\beta^*$ in \eqref{eq:linear_model} under unknown variance. As a byproduct, one can then obtain estimators of the variance   \cite{2012_Sun,bayati2013estimating}.  See also~\cite{fan_variance} for more direct approaches to variance estimation. 
In Section \ref{sec:known}, we rely on the square-root Lasso estimator to construct the  estimator $\widehat{\eta}^{SL}$ which turns out to be minimax in the sparse regime. \\

In the dense regime, we already mentioned the contribution of Dicker~\cite{dicker_variance} that propose method of moments and maximum likelihood based procedures to estimate $\eta$ when $\bSigma$ is known. It is  shown that the square risk of these estimators goes to $0$ at rate $\sqrt{p}/n$. When $p/n$ converges to a finite non-negative constant, these estimator are asymptotically normally distributed. Dicker also considers the case of 
unknown $\bSigma$ when $\bSigma$ is highly structured (allowing $\bSigma$ to be estimable in operator norm at the parametric rate $n^{-1/2}$).
Janson et al.~\cite{janson2015eigenprism} introduce the procedure EigenPrism for computing confidence intervals of $\eta$ and study its asymptotic behavior when $\bSigma$ is known and $p/n$ converges to a constant $c\in (0,\infty)$. Under similar assumptions, Dicker et al.~\cite{dicker2016maximum} have considered a maximum likelihood based estimator.  Bonnet et al.~\cite{bonnet} consider a mixed effect model, which is equivalent to assuming that the parameter $\beta^*$ follows a prior distribution. In the asymptotic where $p/n\to c$, they also propose a $n^{-1/2}$-rate consistent estimator of $\eta$. To summarize, none of the aforementionned contributions has studied minimax convergence rates, the problem of adaptation to sparsity or the estimation problem for unknown $\bSigma$ (to the exception of~\cite{dicker_variance}). 
\\

Finally, there has been a recent interest in  the adaptive estimation of other functionals in the linear model~\eqref{eq:linear_model}, such as  the coordinates $\beta^*_i$ of $\beta^{*}$ or the sum of coordinates $\sum_{i=1}^n \beta^*_i$~\cite{2014_van_de_geer_confidence,2014javanmard,MR3153940,cai2015confidence,javanmard2015biasing}. However, both the statistical methods and the regimes are qualitatively different for these functionals.

\subsection{Notations and Organization}

The set of integers $\{1,\ldots,p\}$ is denoted $[p]$. For any subset $J$ of $ [p]$,  $\bX_J$ is the $n\times |J|$ corresponding submatrix of $\bX$. Given a symmetric matrix $\bA$, $\lambda_{\max}(\bA)$ and $\lambda_{\min}(\bA)$ respectively stand for the largest and the smallest eigenvalue of $\bA$, $|A|$ denotes the determinant of $A$. 
For a vector $u$, $\|u\|_{p}$ denotes its $l_p$ norm and $|u|_{0}$ stands for its $l_0$ norm (ie number of non-zero components). For any matrix $\bA$, $\|\bA\|_{p}$ denotes the $l_p$ norm of the vectorialized version of $\bA$, that is $(\sum |\bA_{i,j}|^p)^{1/p}$. The Frobenius norm is also denoted $\|\bA\|_F$. Finally, the $l_2$ operator norm of a matrix $\bA$ writes $\|\bA\|_{op}$. In what follows, $C$, $C'$,\ldots denote universal constants whose value may vary from line to line whereas $C_{1}, C_{2}$ and $C_{3}$ denote numerical constants that will be used in several places of our work.

In Section \ref{sec:known_sparsity}, we introduce the two main procedures and characterize the minimax estimation risk of $\eta$ when both the covariance matrix $\bSigma$ and the sparsity are known. Section \ref{sec:adapt} is devoted to the problem of adaptation to the unknown sparsity, whereas the case of unknown covariance $\bSigma$  is studied in Section \ref{sec:unknown_covariance}. Extensions to fixed design regression and other related problems  are discussed in Section \ref{sec:discussion}. All the proofs 
are postponed to the end of the paper.

\section{Minimax rates for known sparsity}\label{sec:known_sparsity}
\label{sec:known}

In this section, we consider two estimators. In the spirit of~\cite{dicker_variance}, the first estimator $\widehat{\eta}^D(\bSigma^{-1})$ is designed for the dense regime ($|\beta^*|_0\geq p^{1/2}$) and it is proved to be consistent with rate $\sqrt{p}/n$ irrespectively of the parameter sparsity. When $\beta^*$ is in fact highly sparse, the estimator  $\widehat{\eta}^{SL}$  based on the square-root Lasso better exploits the structure of $\beta^*$ and achieves the estimation rate $\tfrac{|\beta^*|_0\log(p)}{n}+ n^{-1/2}$. It turns out that these  two  procedures (almost) achieve the minimax estimation rate when $|\beta^*|_0$ is known.\\

\subsection{Dense regime}\label{sec:dense_regime}


In this subsection, we introduce an estimator of $\eta$ which will turn out to be mostly interesting for dense parameters $\beta^*$. Its definition is close to that in~\cite{dicker_variance}. We provide a detailed analysis of this estimator, and our bounds in Theorem \ref{prp:risk_dense} below will turn out to be useful both for the adaptation problem and for the case of unknown $\bSigma$. \\

Since $\var{y_1}$ is easily estimated by $\|Y\|_{2}^{2}/n$, the main challenge is to estimate $\|\bSigma^{1/2}\beta^*\|^{2}$. Thus, the question is how to separate in $Y$  the randomness coming from $\bX \beta^*$ from that coming from the $\epsilon_i$'s, $i=1,\ldots,n$. The  idea is to use the fact that the noise $\epsilon$ is isotropic whereas, conditionally on $\bX$, $\bX\beta^*$ is not isotropic. Respectively denote $(\lambda_i,u_i)$, $i=1,\ldots, n$ the eigenvalues and eigenvectors of $({\bf X}{\bf X}^T)/p$. 
We will prove, that in a high-dimensional setting where $p>n$, $\bX\beta^*$ is slightly more aligned with left eigenvectors 
of $\bX$ associated to large eigenvalues than with those associated to small eigenvalues. This subtle phenomenon suggests that the distribution of the random variable $T$
\[T :=  \frac{1}{n^2}\sum_{i=1}^n (\lambda_i-\bar{\lambda}) (Y^T u_i)^2\ ,\quad\quad \text{ where }\quad \bar{\lambda}:= \sum_{i=1}^n \lambda_i/n\ ,\]
(almost) does not depend on the noise level $\sigma$ and, at the same time, captures some functional of the signal $\beta^*$. This functional turns out to be $\beta^{*T}\bSigma^2\beta^*$. One can rewrite the random variable as a quadratic form of $\bY$
\beq\label{eq:defintion_T}
T= \frac{Y^T\left(\bX \bX^T -tr(\bX \bX^T)\bI_n /n\right)Y}{n^2}.
\eeq
Working with a normalized estimator $ \widehat{V}:= \tfrac{Tn^2}{\|Y\|_2^2(n+1)}$, we state in the following theorem that $\widehat{V}$ concentrates exponentially fast around $\beta^{*T} \bSigma^2 \beta^*/\var{y_1}$.

\begin{thm}\label{prp:risk_dense}
Assume that $p\geq n$. \\
There exist numerical constants $C_1$ and $C_2$  such that for  all $t\leq n^{1/3}$,
\beq\label{eq:concentration_eta_dense}
\P\left[\Big|\widehat{V} - \frac{\beta^{*T} \bSigma^2 \beta^*}{\var{y_1}}\Big|\leq C_1\|\bSigma\|_{op}\frac{\sqrt{pt}}{n}\right]\geq 1- C_2 e^{-t}.
\eeq
There exists a numerical constant $C$ such that 
 \beq\label{eq:risk_dense}
 \E\left[\Big(\widehat{V} - \frac{\beta^{*T} \bSigma^2 \beta^*}{\var{y_1}}\Big)^2\right]\leq C\|\bSigma\|_{op}^2 \frac{p}{n^2}\ . 
\eeq
 \end{thm}

 \begin{rem}\emph{ The proof  relies on recent exponential concentration inequalities for  Gaussian chaos~\cite{2015_adamszak} and a new concentration inequality of the spectrum of $\bX\bX^T/n$ around $tr(\bSigma)/n$ (Lemma \ref{lem:concentration_vp_wishart}). 
 The concentration inequality \eqref{eq:concentration_eta_dense} will be the key tool in the construction of adaptive estimators in the next section.}
 \end{rem}

 \begin{rem}\emph{ When $\bSigma$ is the identity matrix, the above theorem enforces that $\widehat{V}$ estimates the proportion of explained variation $\eta$ at the rate $\sqrt{p}/n$, uniformly over all $\beta^*$ and $\sigma>0$. Note that 
 $\widehat{V}$ is only consistent in the regime where $n^2$ is large compared to $p$.}
 \end{rem}


%


\noindent For arbitrary $\bSigma$ (with bounded eigenvalues), the above theorem only implies that $\widehat{V}$ is  {\it of the same order as} $\eta$, that is, there exists positive constant $c$ and $C$ such that $c \lambda_{\min}(\bSigma)\leq \widehat{V}/\eta\leq  C\lambda_{\max}(\bSigma)$.\\
Nevertheless, when the covariance $\bSigma$ is known, it is possible to get a consistent estimator of $\eta$.  Replace the design matrix $\bX$ in the linear regression model by $\tilde{\bX}:= \bX \bSigma^{-1/2}$ in such a way that its rows $\tilde{\bx}_i$ follow i.i.d. standard normal distributions and 
\beq
Y=\tilde{\bX} \bSigma^{1/2}\beta^*+ \epsilon\ .
\eeq
Then, we define the estimator $\widehat{\eta}^D$ as $\widehat{V}$ where $\bX$ is replaced by $\tilde{\bX}$, so that  $\widehat{\eta}^D$ is a quadratic form of $\bY$ with a matrix involving the precision matrix, that is the inverse covariance matrix $\bSigma^{-1}$. Let us  denote $\bOmega:= \bSigma^{-1}$, and define
\beq
\label{eq:etadense}
\widehat{\eta}^D (\bOmega):= \frac{Y^T\left(\bX\bOmega \bX^T -tr(\bX \bOmega \bX^T)\bI_n /n\right)Y}{(n+1)\|Y\|^2}
\eeq
(we could replace $tr(\bX \bOmega \bX^T)$ by $p$ in the above definition without changing the rate in the corollary below).
We straightforwardly derive from Theorem \ref{prp:risk_dense} that $\widehat{\eta}^D (\bOmega)$ estimates $\eta$ at the rate $\sqrt{p}/n$.

\begin{cor}\label{cor:risk_dense}
Assume that $p\geq n$.
There exists a  numerical constant $C$ such that the estimator $\widehat{\eta}^D (\bOmega)$ satisfies 
 \beq\label{eq:risk_dense_1}
 \E\left[\big(\widehat{\eta}^D  (\bOmega)- \eta\big)^2\right]\leq C\frac{p}{n^2}\ . 
\eeq
\end{cor}

\begin{rem}\emph{ It turns out that $\widehat{\eta}^D(\bOmega)$ is consistent for $p$  small compared to $n^2$ even though  consistent 
 estimation of $\beta^*$ is impossible in this regime. Although developed independently, the estimator $\widehat{\eta}^D(\bOmega)$ shares some similarities with the method of moment based estimator of Dicker~\cite{dicker_variance}, which also achieves the  $\sqrt{p}/n$ convergence rate. }
\end{rem}

 \subsection{Sparse regime: square-root Lasso estimator}\label{sec:sparse}

When $\beta^*$ is highly sparse, the signal to noise ratio estimator is based on a Lasso-type estimator of $\beta^*$  proposed in~\cite{squarerootlasso,2012_Sun}. As customary for  Lasso-type methods, we shall work with a standardized  version  $\bW$  of the matrix $\bX$, whose columns $\bW_{\bullet j}$ satisfy $\|\bW_{\bullet j}\|_2=1$. Since the noise-level $\sigma$ is unknown, we cannot readily use the classical Lasso estimator whose optimal value of the tuning parameter depends on $\sigma$. Instead, we rely on  the square-root Lasso~\cite{squarerootlasso} defined by 
\beq\label{eq:definition_square_root_Lassos}
\widetilde{\beta}_{SL} := \argmin_{\beta \in \mathbb{R}^p}\sqrt{\|Y-\bW \beta\|^2_2} + \frac{\lambda_0}{\sqrt{n}} \|\beta\|_1\ , \quad \quad (\widehat{\beta}_{SL})_j:= (\widetilde{\beta}_{SL})_{j}/\|\bx_j\|_2\ .
\eeq
In the sequel, the tuning parameter $\lambda_0$ is set to $\lambda_0:= 13\sqrt{\log(p)}$ (there is nothing specific with this particular choice).
In the proof, we will also use an equivalent definition of the square-root estimator introduced in \cite{2012_Sun}
\beq\label{eq:definition_scaled_lass}
(\widetilde{\beta}_{SL} , \widetilde{\sigma}_{SL})=  \argmin_{\beta \in \mathbb{R}^p,\ \sigma'>0}\left[\frac{n\sigma'}{2}+ \frac{\|Y-\bW\beta\|_2^2}{2\sigma'}\right]+ \lambda_0\|\beta\|_1\ .
\eeq 
(To prove the equivalence between the two definitions, minimize \eqref{eq:definition_scaled_lass} with respect to $\sigma'$.) Notice that $\widetilde{\sigma}_{SL}= \|Y-\bW\widehat{\beta}_{SL}\|_2/\sqrt{n}$. 
Then, we define the estimator 
\beq \label{eq:definition_eta_Lasso}
\widehat{\eta}^{SL}:= 1- \frac{n\widetilde{\sigma}^2_{SL}}{\|Y\|_2^2}= 1 - \frac{\|Y-\bW \widehat{\beta}_{SL}\|_2^2}{\|Y\|_2^2} \ , 
\eeq
The following proposition is a consequence of Theorem 2 in \cite{2012_Sun}.

\begin{prp}\label{thm:eta_gauss_Lasso}
 There exist two numerical constants $C$ and $C'$ such that the following holds. Assume that $\beta^*$ is $k$-sparse, that $p\geq n$ and 
 \beq\label{eq:hypo_max_dimension}
 k\log(p) \frac{\lambda_{\max}(\bSigma)}{\lambda_{\min}(\bSigma)} \leq  C n\ . 
 \eeq
 Then the square-root Lasso based estimator $\widehat{\eta}^{SL}$ satisfies
 \beq\label{eq:risk_lasso_eta_gl}
\E\left[( \widehat{\eta}^{SL} - \eta)^{2} \right] \leq C' \left[\frac{1}{n}+ \frac{k^2\log^2(p)}{n^2}\frac{\lambda^2_{\max}(\bSigma)}{\lambda^2_{\min}(\bSigma)} \right]\ .
 \eeq
 
\end{prp}

\begin{rem}
\emph{Condition \eqref{eq:hypo_max_dimension} is unavoidable, as the minimax risk of proportion of explained variation estimation is bounded away from zero when $k\log(p)$ is large compared to $n$ (see Proposition \ref{prp:lower_minimax} later). To ease the presentation, we have expressed Condition \eqref{eq:hypo_max_dimension} in terms of largest and smallest eigenvalues of $\bSigma$. One could in fact replace these quantities by local ones such as  compatibility constants (see the proof for more details). }
\end{rem}

\subsection{Minimax lower bound}\label{sec:minimax_lower}

We shall prove in the sequel that a combination of the estimators $\widehat{\eta}^D(\bOmega)$ and $\widehat{\eta}^{SL}$ essentially achieves the minimax estimation risk.  In the following minimax lower bound we assume that the covariance $\bSigma$ is the identity matrix $\bI_p$.

Define $\mathbb{B}_0[k]$ the collection of $k$-sparse vectors of size $p$. 
Given any estimator $\widehat{\eta}$, define the maximal risk $R(\widehat{\eta},k)$ over $k$-sparse parameters  by 
\[R(\widehat{\eta},k):= \sup_{\beta\in \bbB_0[k],\ \sigma>0 }\bbE_{\beta,\sigma}\left[\{\hat{\eta}-\eta(\beta,\sigma)\}^2\right]\ ,\]
where $\bbE_{\beta,\sigma}[.]$ is the expectation with respect to $(Y,\bX)$ where $Y=\bX\beta+\epsilon$, with $\epsilon\sim \cN(0,\sigma^2\bI_n)$ and the covariance matrix of the rows of $\bX$ is $\bI_p$. Then,  the minimax risk is denoted  $R^*(k):=\inf_{\hat{\eta}}R(\widehat{\eta},k)$. 

\begin{prp}[Minimax lower bound]\label{prp:lower_minimax}
There exists a numerical constant $C>0$ such that for any $1\leq k\leq p$, 
\beq\label{eq:lower_bound_minimax}
R^*(k)\geq C
 \left(\left\{\left[\frac{k}{n}\log\left(1+ \frac{p}{k^2}\vee \sqrt{\frac{p}{k^2}}\right)\right]^{2}\wedge 1 \right\}+ \frac{1}{n}\right)\ .
\eeq
\end{prp}


The proof of this proposition follows the lines developed to derive  minimax lower bounds for the signal detection problem (see e.g.\ Theorem 4.3 in \cite{2010_AS_Verzelen}). Nevertheless, as this proposition is a first step towards more complex settings, we provide a self-contained proof in Section \ref{sec:lower}.\\

In \eqref{eq:lower_bound_minimax}, we recognize three regimes: 
\begin{itemize}
 \item If $k\geq p^{1/2}$, the minimax rate is larger than $(\sqrt{p}/n)\wedge 1$. This  optimal  risk is achieved by the dense estimator $\widehat{\eta}^D(\bOmega)$ up to a constant number. 
 \item If $k \leq p^{1/2-\gamma}$ for some arbitrary small $\gamma>0$, the minimax rate is of order 
 \[\frac{1}{\sqrt{n}}+\left(\frac{k\log(p)}{n}\right) \wedge 1 \ .\]
 More precisely for $k\leq [\sqrt{n}/\log(p)]
 $, it is of order $n^{-1/2}$, whereas for larger $k$ it is of order $k\log(p)/n\wedge 1$. This bound is achieved by the square-root Lasso estimator $\widehat{\eta}^{SL}$, which does not require the knowledge of $\bSigma$ and $k$.
 \item For $k$ close to $p^{1/2}$ (e.g.\  $k=(p/\log(p))^{1/2}$), the minimax lower bound \eqref{eq:lower_bound_minimax} and the upper bound \eqref{eq:risk_lasso_eta_gl} only match up to some $\log(p)$ factors. Such a logarithmic mismatch has also been obtained in the related work~\cite{baraud_minimax} on minimax detection rates for testing the null hypothesis $\beta^*=0$ when the design matrix is fixed and orthonormal, that is $p=n$ and $\bX=\bI_p$. In this orthonormal setting, Collier et al.~\cite{collier2015minimax} have very recently closed this gap. Transposed in our setting, their results would suggest that the optimal risk is of order $k\log(p/k^2)/n$, suggesting that  Proposition \ref{prp:lower_minimax} is sharp. In the specific case where $\bSigma=\bI_p$, it seems possible to extend the estimator of $\|\beta^*\|_2^2$ introduced by~\cite{collier2015minimax} to our setting by considering the pairwise correlations $Y^T\bW_{\bullet j}$ for $j=1,\ldots,p$. Such estimator would then presumably be  $k\log(p/k^2)/n$ consistent. As this approach does not seem extend easily 
 to arbitrary $\bSigma$, we did not go further in this direction.
\end{itemize}

\section{Adaptation to unknown sparsity}\label{sec:adapt}

In practice, the number $|\beta^*|_0$  of non-zero components of $\beta^*$ is unknown. In this section, our purpose is to build an estimator $\widehat{\eta}$ that adapts to the unknown sparsity $|\beta^*|_0$. Although the computation of the estimators $\widehat{\eta}^{D}(\bOmega)$ and $\widehat{\eta}^{SL}$ does not require the knowledge of $|\beta^*|_0$, the choice of one estimator over the other depends on this quantity. Observe that, when $p\geq n^2$, 
the dense estimator $\widehat{\eta}^{D}(\bOmega)$ is not consistent. Therefore, only the estimator $\widehat{\eta}^{SL}$ is useful and $\widehat{\eta}^{SL}$ alone is minimax adaptive to the sparsity $k$ (up to a possible $\log$ factor when $k$ is of the order of $p^{1/2}$). This is why we focus on the regime where $p$ is large compared to $n$ and where $p\log p \leq n^{2}$.


It turns out  that no estimator  $\widehat{\eta}$ can  simultaneously achieve the minimax  risk $R^*(k)$ over all $k=1,\ldots, p$, and that there is an unavoidable loss for adaptation. This may be seen
in the following proposition.

\begin{prp}\label{prp:lower_adaptation} 
Assume that $p\log p \leq n^{2}$, and that for some  $a\in ]0,1/2[$,  $p^{1-a}(\log p)^{2} \geq 16n$.
Then for any estimator $\widehat{\eta}$, for all $k$ such that $\sqrt{p\log p} \leq k \leq p$, one has
$$
\frac{R(\widehat{\eta},1)}{\frac{1}{n}\sqrt{\frac{p}{n}}} + \frac{R(\widehat{\eta},k)}{\frac{p\log p}{n^{2}}}  \geq \frac{a^{2}}{4^{5}}.
$$
\end{prp}

Recall that $R^*(1)$ is of order $1/n$ and $R^*(k)$ is of order $p/n^2$. Proposition \ref{prp:lower_adaptation} implies that any estimator $\widehat{\eta}$ whose maximal risk over $\bB_0[k]$ is smaller than $p\log(p)/n^2$ exhibits 
a huge maximal risk over $\bB_0[1]$. As a consequence, any estimator admitting a reasonable risk bound over $\bB_0[1]$ should have a maximal risk at least of order $p\log(p)/n^2$ for all $k\in [\sqrt{p\log(p)},p]$. Next, we define an estimator $\widehat{\eta}^{A}$ simultaneously achieving the risk $R^*(k)$ for $k$ small compared to $\sqrt{p}$ and achieving the risk $R^* (k)\log p$ in the dense regime where $k \geq \sqrt{p\log p}$.


Define the numerical constant $c_0$ as two times the constant $C_{1}$ arising in the deviation bound \eqref{eq:concentration_eta_dense}  of Theorem \ref{prp:risk_dense}. 
We build an adaptive estimator by combining the estimator $\widehat{\eta}^{SL}$ and  $\widehat{\eta}^{D}$ as follows 
\beq\label{eq:definition_adaptive}
\widehat{\eta}^A := \left\{\begin{array}{ccc}
                            \widehat{\eta}^{SL} & \text{ if } & |\widehat{\eta}_T^D(\bOmega) - \widehat{\eta}^{SL}|\leq c_0\sqrt{p\log(p)}/n \\
                            \widehat{\eta}_T^{D}(\bOmega)  & \text{ else }
                           \end{array}
\right.
\eeq
where, for technical reasons, we consider
$\widehat{\eta}_T^D(\bOmega):= \min(1,\max (0, \widehat{\eta}^D(\bOmega)))$  a truncated version of $\widehat{\eta}^D(\bOmega)$ which  lies in $[0,1]$.

The rationale behind $\widehat{\eta}^A$ is the following. Suppose that $\beta^*$ is $k$-sparse, with $k\leq \sqrt{p}$, in which case, $\widehat{\eta}^{SL}$ achieves the optimal rate. With large probability,  $|\widehat{\eta}_T^D(\bOmega) -\eta|$ is smaller than $c_0\sqrt{p\log(p)}/(2n)$ (this is true for arbitrary $\beta^*$) and $|\widehat{\eta}^{SL} -\eta|$ is smaller than $(1/\sqrt{n}+k\log(p)/n)$ which is smaller than $ c_0\sqrt{p\log(p)}/(2n)$. Hence, $\widehat{\eta}^A$ equals $\widehat{\eta}^{SL}$ with large probability. Now assume that $k\geq \sqrt{p}$, in which case the optimal rate is of order $\sqrt{p}/n$ and is achieved by $\widehat{\eta}^{D}_T(\bOmega)$. Observe that $\widehat{\eta}^A=\widehat{\eta}_T^D(\bOmega)$ except if $\widehat{\eta}^{SL}$ is at distance less than $c_0\sqrt{p\log(p)}/n$ from $\widehat{\eta}_T^D(\bOmega)$. Consequently, $|\widehat{\eta}^A-\eta|\leq c_0\sqrt{p\log(p)}/n+|\widehat{\eta}_T^D(\bOmega)-\eta|$. Formalizing the above argument, we arrive at the following.

\begin{thm}\label{prp:adaptation}
There exists a numerical constant $C$ such that the following holds. Assume that $p\geq n$.
For any integer $k\in [p]$, any $k$-sparse vector $\beta^*$ and any $\sigma>0$, the estimator  $\widehat{\eta}^A$ satisfies
\[
 \E\left[\left(\widehat{\eta}^A - \eta\right)^2\right]\leq C \left[\frac{1}{n} + \left(\frac{k^2\log^2(p)}{n^2}\frac{\lambda^2_{\max}(\bSigma)}{\lambda^2_{\min}(\bSigma)}\right)   \wedge \left(\frac{p\log(p)}{n^2}\right)\right]\ .
\]
\end{thm}

As a consequence of Propositions \ref{prp:lower_minimax}, \ref{prp:lower_adaptation} and \ref{prp:adaptation}, and, in the asymptotic regime where $p\log p \leq n^{2}$ and $p^{1-a}$ is large compared to $n$ for some positive $a$, $\widehat{\eta}^A$ is achieves the optimal adaptive risk for all 
 $k\in\{1,\ldots, p^{1/2-\gamma}\}\cup\{(p\log(p))^{1/2},\ldots, p\}$ where $\gamma>0$ is arbitrary small. For $k$ close to $\sqrt{p}$, there is still a logarithmic gap between the upper and lower bounds as in the non-adaptive section. 

\begin{rem}\emph{
Theorem \ref{prp:risk_dense} is the basic stone for the construction of   $\widehat{\eta}^A$ by the use of the deviation inequality.
The constant $c_{0}$ may be quite large, as the constant $C_{1}$ in the deviation inequality, making the estimator difficult to use in practice if $n$ and $p$ are not large enough. 
Theorem \ref{prp:adaptation} however allows to understand how adaptation to sparsity is possible. 
}
\end{rem}

\section{Minimax estimation when $\bSigma$ is unknown}\label{sec:unknown_covariance}
\label{sec:Sunknown}

In this section, we investigate the case where  the covariance matrix $\bSigma$ is unknown. As the computation of the sparse estimator $\widehat{\eta}_{SL}$ does not require the knowledge of $\bSigma$, the optimal estimation rate is therefore unchanged when $|\beta^*|_0$ is much smaller than $\sqrt{p}$. In what follows we therefore focus on the regime where $|\beta^*|_0\geq \sqrt{p}$.


\subsection{Positive results under restrictions on $\bSigma$}\label{sec:plug-in}

Here, we prove that a simple plug-in method allows to achieve the minimax rate as long as one can estimate the inverse covariance matrix $\bOmega$ 
sufficiently well. This approach has already been considered by Dicker~\cite{dicker_variance} who has proved a result analogous  to  Proposition~\ref{prp:upper_risk_unknown_Sigma}. For the sake of completeness, we provide detailed arguments and also consider the problem of adaptation to the sparsity. Without loss of generality, we may assume that we have at our disposal an independent copy of $\bX$, denoted $\bX^{(2)}$ (if it is not the case, simply divide the data set into two subsamples of the same size).

Given an estimator $\widehat{\bOmega}$ of $\bOmega:=\bSigma^{-1}$ based on the matrix $\bX^{(2)}$, the proportion of explained variation  $\eta$ is estimated 
as in Section \ref{sec:dense_regime}, 
using (\ref{eq:etadense}), 
except that the true inverse covariance matrix is replaced by its estimator:
\beq \label{eq:definition_widehateta}
 \widehat{\eta}^D(\widehat{\bOmega})  := \frac{Y^T\left(\bX\widehat{\bOmega} \bX^T -tr(\bX \widehat{\bOmega}\bX^T)\bI_n /n\right)Y}{(n+1)\|Y\|^2}.
\eeq

\begin{prp}\label{prp:upper_risk_unknown_Sigma}
Assume that $p\geq n$. 
For any non-singular estimator $\widehat{\bOmega}$ based on the sample $\bX^{(2)}$, 
\beq\label{eq:upper_risk_unknown_Sigma}
  \P\left[\big|\widehat{\eta}^D(\widehat{\bOmega})- \eta \big|\geq C_1 \|\bSigma\|_{op}\|\widehat{\bOmega}\|_{op}\frac{\sqrt{pt}}{n}+  \|\bSigma\|_{op}\|\widehat{\bOmega}-\bOmega\|_{op}\Big| \bX^{(2)}  \right]\leq C_2e^{-t} \ ,
  \eeq
  for all $t<n^{1/3}$. Here, $C_{1}$ and $C_{2}$ are the numerical constants that appear in Theorem \ref{prp:risk_dense}.
\end{prp}

Thus, if one is able to estimate $\bOmega$ at the rate $\sqrt{p}/n$, then $\widehat{\eta}^D(\widehat{\bOmega})$ achieves the same estimation rate as if $\bSigma$ was known. To illustrate this qualitative situation, we describe an example of a class $\cU$ of precision matrices and an estimator $\widehat{\bOmega}$ satisfying this property.

\bigskip

For any square matrix $\bA$, define its matrix $l_1$ operator norm by $\|\bA\|_{1\to 1}=\max_{1\leq j\leq p}\sum_{1\leq i \leq p}|\bA_{i,j}|$. Given any $M>0$ and $M_1>0$, consider the following collection $\cU$ of sparse inverse covariance matrices
\beq\label{definition_U}
\cU:=\cU(M,M_1):=\left\{
\bOmega:\ \bOmega \succ 0:\quad \begin{array}{l}
\frac{1}{M_1}\leq \lambda_{\min}(\bOmega)\leq \lambda_{\max}(\bOmega)\leq M_1,\quad \|\bOmega\|_{1\to 1}\leq M, \\                               
 \max_{1\leq j\leq p}\sum_{i=1}^p \mathbf{1}_{\bOmega_{i,j}\neq 0}\leq \sqrt{\frac{p}{n\log(p)}}
\end{array}
\right\}\ .
\eeq
Cai et al~\cite{clime} introduced the CLIME estimator to estimate sparse precision matrices. Let $\lambda_n>0$ and $\rho>0$ be two tuning parameters, whose value will be fixed in Lemma \ref{lem:estimation_omega} below. Denote $\widehat{\bSigma}^{(2)}:= \bX^{(2)T}\bX^{(2)}/n$ the empirical covariance matrix based on the observations $\bX^{(2)}$.

\medskip 

Let $\widehat{\bOmega}_1$ be the solution of the following optimization problem
\beq \label{eq:definition_clime}
\min \|\bOmega'\|_1 \ ,\quad  \text{ subject to }\quad \|\widehat{\bSigma}^{(2)} \bOmega'-\bI_p\|_{\infty}\leq \lambda_n, \ \bOmega'\in \mathbb{R}^{p\times p}\ .
\eeq
Then, the CLIME estimator $\widehat{\bOmega}_{CL}$ is obtained by symmetrizing $\widehat{\bOmega}_{1}$: for all $i,j$,
we take $(\widehat{\bOmega}_{CL})_{i,j}=(\widehat{\bOmega}_{1})_{i,j}$ if $|(\widehat{\bOmega}_{1})_{i,j}|\leq |(\widehat{\bOmega}_{1})_{j,i}|$ and $(\widehat{\bOmega}_{CL})_{i,j}=(\widehat{\bOmega}_{1})_{j,i}$ in the opposite case. We may now apply Theorem 1.a in~\cite{clime}  to our setting with $\eta= 1/5\wedge 1/\sqrt{M_1}$, $K=e^{1/2}$ and $\tau=1$.
This way we obtain the following.

\begin{lem}\label{lem:estimation_omega} There exists a numerical constant $C_{3}>0$ such that the following holds. Fix $\lambda_n= 2[25 \vee M_1](3+e^{3}(5\vee \sqrt{M_1} )^2M\sqrt{\log(p)/n}$. Assume that $\log(p)\leq n/8$ and that $\bOmega$ belongs to $\cU$. Then, the CLIME estimator satisfies 
\beq\label{eq:control_tilde_omega}
\|\widehat{\bOmega}_{CL}- \bOmega\|_{op}\leq C_{3} M^2 M_1^{2} \frac{\sqrt{p}}{n} \ ,
\eeq 
with probability larger than $1-4/p$. 
\end{lem}

\noindent
Let us modify  the estimator of $\eta$ so that it effectively lies in $[0,1]$. Let
$\widehat{\eta}_T^D(\widehat{\bOmega}):= \min(1,\max (0, \widehat{\eta}^D(\widehat{\bOmega})))$.
\begin{cor}\label{cor:dense-plug_in}
Assume that $p\geq n$ and that $\bOmega$ belongs to the collection  $\cU$ defined above.
Then, tre exists a universal constant $C>0$   such that the following holds.  For any $\beta^*$ and $\sigma>0$,
\[
\E\left[\left\{\widehat{\eta}_T^D(\widehat{\bOmega}_{CL})- \eta\right\}^2\right]\leq \Big[C M^4 M_1^{ 6} \frac{p}{n^2}\Big] \wedge 1.
\]
\end{cor}

 We shall now define  an adaptive estimator  $\widehat{\eta}^A_{CL}$ in the same spirit as $\widehat{\eta}^A$ in the previous subsection. Define $c_0(M,M_1)$ by
 $$
 c_0(M,M_1):=4 C_{1}M_{1}^{2}+2 C_{3} M^{2}M^{3}_{1}.
 $$
 Here, $C_{1}$ is the numerical constant that appears in Theorem \ref{prp:risk_dense} and $C_{3}$ the numerical constant that appears in Lemma \ref{lem:estimation_omega}.
 Define the estimator as:
\beq\label{eq:definition_adaptive_CL}
\widehat{\eta}^A_{CL} := \left\{\begin{array}{ccc}
                            \widehat{\eta}^{SL} & \text{ if } & |\widehat{\eta}_T^D(\widehat{\bOmega}_{CL}) - \widehat{\eta}^{SL}|\leq c_0(M,M_1)\sqrt{p\log(p)}/n \\
                            \widehat{\eta}^{D}_T(\widehat{\bOmega}_{CL})  & \text{ else } \ .
                           \end{array}
\right.
\eeq
We then obtain that $\widehat{\eta}^A_{CL}$ is asymptotically minimax adaptive to $\bOmega$ (if it is known that $\bOmega\in \cU$) and to sparsity, in the same regimes as those in which $\widehat{\eta}^A$  is asymptotically minimax adaptive to sparsity.

\begin{cor}\label{cor:plug_in}
Assume that $\bOmega$ belongs to the collection  $\cU$ defined above.
Then, there exists a  constant  $C(M,M_1)>0$ only depending on $M$ and $M_1$ such that the following holds.  For any integer $k\in [p]$, any $k$-sparse vector $\beta^*$ and any $\sigma>0$,
\beq \label{eq:risk_A_CL}
 \E\left[\left(\widehat{\eta}^A_{CL} - \eta\right)^2\right]\leq C(M,M_1) \left[\frac{1}{n} + \left(\frac{k^2\log^2(p)}{n^2}\right)   \wedge \left(\frac{p\log(p)}{n^2}\right)\right]\ . 
 \eeq 
\end{cor}

\begin{rem}\emph{
 When $\bOmega$ belongs to $\cU$, the estimator $\widehat{\eta}_T^D(\widehat{\bOmega}_{CL})$ achieves a similar risk bound to that of  $\widehat{\eta}_T^D(\bOmega)$. Also, $\widehat{\eta}^A_{CL}$ performs as well as  estimator $\widehat{\eta}^A$ which requires the knowledge of $\bOmega$. As a consequence, there does not seem to be a price to pay for the adaptation to $\bOmega$ under the restriction $\bOmega\in \cU$.}
\end{rem}

\begin{rem}\emph{
If the quantity $\sqrt{p/(n\log(p))}$ in  the sparsity condition  $ \max_{1\leq j\leq p}\sum_{i=1}^p \mathbf{1}_{\bOmega_{i,j}\neq 0}\leq \sqrt{p/(n\log(p))}$ in the definition \eqref{definition_U} of $\cU$ is replaced by some $s\geq  \sqrt{p/(n\log(p))}$, the CLIME-based estimator  $\widetilde{\eta}_T^D(\widehat{\bOmega}_{CL})$ will only  be consistent at the rate $s\sqrt{\log(p)/n}$ which is slower than the desired $\sqrt{p}/n$. This is not completely unexpected as we prove in the next subsection that a reliable estimation of $\eta$ becomes almost impossible when the collection of precision matrices is too large.
}
\end{rem}

\subsection{Impossibility results}
\label{subsec:impossible}

We now turn to the general problem where $\bSigma$ is only assumed to have  bounded eigenvalues. As explained in the beginning of Section \ref{sec:adapt}, the estimator $\widehat{\eta}^{SL}$, which does not require the knowledge of $\bSigma$, is minimax adaptive to $\mathbb{B}_0[k]$ when $p\geq n^2$. Hence, we focus in the remainder of this section on the regime $n\leq p \leq n^{2}$.

\medskip

In this subsection and the corresponding proofs, we denote $\P_{\beta,\sigma,\bSigma}$ the distribution of $(Y,\bX)$, in order to emphasize the dependency of the data distributions with respect to the covariance matrix of $\bX$. For any $M>1$, let us introduce $\Xi[M]$ the set of positive symmetric matrices of size $p$
whose eigenvalues lie in the compact $[1/M, M]$. The purpose of these bounded eigenvalues in $(1/M,M)$ is to prove that the difficulty in the estimation problem does not simply arise because of poorly invertible covariance matrices.

Denote $\overline{R}^*[p,M]$ the minimax estimation risk of the the proportion of explained variation $\eta$ when the covariance matrix is unknown
\beq\label{eq:definition_minimax_risk}
\overline{R}^*[p,M]:= \inf_{\widehat{\eta}}\sup_{\beta \in \mathbb{B}_0[p],\ \sigma>0}\sup_{\bSigma\in \Xi[M]}\E_{\beta,\sigma ,\bSigma}\Big[\left(\widehat{\eta}-\eta(\beta,\sigma)\right)^2\Big]\ . 
\eeq

When the covariance matrix $\bSigma$ is known, the minimax rate has been shown to be of order $\sqrt{p}/n$ and therefore goes to $0$ as soon as $p$ is small compared to $n^2$. The following proposition shows that, for unknown $\bSigma$, there is no consistent estimators of $\eta$ when $p$ is large compared to $n$.

\begin{thm}\label{thrm:lower_minimax_unknown_variance}
  Consider an asymptotic setting where both $n$ and $p$ go to infinity. Then, there exists a positive numerical constant $C$ and a function $M: x\mapsto M(x)$ mapping $(0,\infty)$ to $(1,\infty)$ such that the following holds.
 If for some  $\varsigma>0$,
 \beq \label{eq:assumption_minimax_lower_unknown_variance}
 \frac{n^{1+\varsigma}}{p}\to 0
 \ , 
 \eeq
 then  the minimax risk $\overline{R}^*[p,M(\varsigma)]$ is bounded away from zero, that is $\underline{\lim}\overline{R}^*[p,M(\varsigma)]\geq C$.
\end{thm}

\begin{rem}
\emph{
Theorem \ref{thrm:lower_minimax_unknown_variance} tells us that it is impossible to consistently estimate  the proportion of explained variation in a high-dimensional setting where $p$ is much larger than $n$.  
This lower bound straightforwardly extends to $\overline{R}^*[k,M(\varsigma)]$ when $k$ is much larger than $n$ in the sense $n^{1+\varsigma}/k\to 0$ for some $\varsigma>0$.
}

\end{rem}

\begin{rem}[Dependency of constants with $\varsigma$]\emph{
In the proof of Theorem \ref{thrm:lower_minimax_unknown_variance}, the bound $M[\varsigma]$ is increasing when $\varsigma$ gets closer to zero, thereby allowing the spectrum of $\bSigma$ to be broader. If we want to consider the minimax risk $\overline{R}^*[p,M]$ with a fixed $M>0$ (independent of $\xi$), then one can prove that, whenever $n^{1+\varsigma}/p\to 0$, then $\underline{\lim}\overline{R}^*[p,M]\geq C(\varsigma)$ for some $C(\varsigma)>0$ only depending on $\varsigma$. Some details are provided in the proof.
}
\end{rem}

Let us get a glimpse of the  proof by trying to build an estimator of $\eta(\beta^*,\sigma)$ in the high-dimensional regime $p\geq n$. As $\bOmega$ is unknown and cannot be consistently estimated in this regime, a natural candidate would be to consider $\widehat{\eta}^D(\bI_p)= \widehat{V}$ as defined below \eqref{eq:defintion_T}. By Theorem \ref{prp:risk_dense}, one has
\[\widehat{\eta}^D(\bI_p)= \frac{\beta^{*T}\bSigma^2\beta^*}{\var{y_i}} + O_P(\frac{\sqrt{p}}{n})\ .\]
Although the signal strength $\beta^{*T}\bSigma \beta^*$ cannot be consistently estimated for unknown $\bSigma$ (Theorem \ref{thrm:lower_minimax_unknown_variance}), it is interesting to note that some regularized version of the signal strength $\beta^{*T}\bSigma^2\beta^*$ is estimable at the rate $\sqrt{p}/n$ (this phenomenon was already observed in~\cite{dicker_variance}).

Going one step further, one can consistently estimate  $\beta^{*T}\bSigma^3\beta^*$ for $p\leq n^{3/2}$ by considering a quadratic form of $Y$ as in $T$ \eqref{eq:defintion_T} but with higher-order polynomials of $\bX$. For $p$ of order $n^{1+\varsigma}$ for some small $\varsigma>0$, it will be possible to consistently estimate all $a_q:= \beta^{*T}\bSigma^q\beta^*$ for $q=2,3,\ldots, r(\varsigma)$ where $r(\varsigma)$ is a positive integer only depending on $\varsigma$.

Then, one may wonder whether it is possible to reconstruct $a_1= \beta^{*T}\bSigma \beta^*$ from $(a_q)$, $q=2,\ldots, r(\varsigma)$.  Observe that $a_q$ is the $q$-th moment of a positive discrete measure $\mu$ supported by the spectrum of $\bSigma$ and whose corresponding weights are the square norms of the projections of $\beta^*$ on the eigenvectors of $\bSigma$. As a consequence, estimating $\beta^{*T}\bSigma \beta^*$ from $(a_q)$, $q=2,\ldots, r(\varsigma)$ is a partial moment problem where one aims at recovering the first moment of the measure $\mu$ given its higher order moments up to $r(\varsigma)$.
Following these informal arguments, we build, in the proof of Theorem \ref{thrm:lower_minimax_unknown_variance}, two discrete measures $\mu_1$ and $\mu_2$ supported on  $(1/M(\varsigma),M(\varsigma))$ whose $q$-th moments coincide for $q=2,\ldots, r(\varsigma)$ and whose first moments are far from each other. Define $\cB_1$ (resp. $\cB_2$) the collection of parameter $(\beta^*,\bSigma)$ whose corresponding measure is $\mu_1$ (resp. $\mu_2$). Then, we show that no test can consistently  distinguish the hypothesis $H_0:(\beta^*,\bSigma)\in \cB_1$ from $H_1:(\beta^*,\bSigma)\in \cB_2$. 
As the signal strengths $\beta^{*T}\bSigma \beta^*$  of parameters in $\cB_1$ are far from those in $\cB_2$, this implies that consistent estimation is impossible in this setting.

\bigskip

\begin{rem}\emph{
Let us summarize our findings on the minimax estimation risk when $\bSigma$ is unknown and $n\leq p\leq n^{2}$:
\begin{itemize}
 \item if $k$ is small compared to $\sqrt{p}$, the minimax risk is of order $[k\log(p)/n\wedge 1] + n^{-1/2}$ and is achieved by the square-root Lasso estimator $\widehat{\eta}^{SL}$.
 \item if $k$ is large compared to $n$ (in the sense  $n^{1+\varsigma}/k\to 0$ for some $\varsigma>0$), then consistent estimation is impossible.
 \item if $k$ lies between $\sqrt{p}$ and $n/\log(p)$, the square-root Lasso estimator $\widehat{\eta}^{SL}$ is consistent at the rate $k\log(p)/n$. We conjecture that this rate is optimal.
 \item if $k$ lies between $n/\log(p)$ and $n$, we are not aware of any consistent estimator $\eta$ and we conjecture that consistent estimation is impossible.
\end{itemize}
} 

\end{rem}

\section{Discussion and extensions}\label{sec:discussion}

We focused in this work on the estimation risk of $\eta$ in high-dimensional linear models under two major assumptions: the  design is random (with possibly unknown covariance matrix) and the  level of noise $\sigma$ is unknown. 
We first discuss  how the difficulty of the problem is modified when the two assumptions are not satisfied: when the design is not random, then consistent estimation of $\eta$ is impossible in the dense regime, and when the level of noise is known, then the estimation of $\eta$ becomes much easier in the dense regime. Finally, we mention the problem of constructing optimal confidence intervals.

\subsection{Fixed design}

If the regression design $\bX$ is considered as fixed, then the counterpart of the proportion of explained variation would be 
\[\eta[\beta^*,\sigma, \bX]:= \frac{\|\bX\beta^*\|_2^2/n}{\|\bX \beta^*\|_2^2/n+\sigma^2}\ .\]
In this new setting, the square-root Lasso estimator still estimates $\eta[\beta^*,\sigma, \bX]$ at the rate 
$n^{-1/2}+ k\log(p)/n$ up to multiplicative constants only depending on the sparse eigenvalues and compatibility constants of $\bX$. In contrast, the construction of $\widehat{V}$  relies on the fact that $\bX$ is random and is independent of the isotropic noise $\epsilon$. When $\bX$ is considered as fixed,  $\widehat{V}$ does not consistently estimate $\eta[\beta^*,\sigma, \bX]$ for $p$ small compared to $n^2$. As a simple example, take $\sigma=1$ and define $\beta^*$ by $\beta^{*T}v_i= \lambda^{-1/2}_{i}$
 for $i=1,\ldots,n$ where 
$(v_i)_i$ denote the right eigenvectors of $\bX$ and $(\lambda^{1/2}_i)_i$ its singular values. Then, the random variables $T$ and $\hat{V}$ (defined in Section \ref{sec:dense_regime}) are concentrated around $0$, whereas $\eta[\beta^*,\sigma, \bX]$ equals $1/2$.

More generally, the next proposition states that it is impossible to consistently estimate $\eta[\beta^*,\sigma, \bX]$ in a high-dimensional setting $p\geq n+1$. The randomness of $\bX$ therefore plays a fundamental role in the problem.

\begin{prp}\label{prp:lower_fixed_design}
 Assume that $p> n$ and consider any fixed design $\bX$ such that $\mathrm{Rank}(\bX)=n$. Given $\beta^*$ and $\sigma$, denote $\underline{\P}_{\beta^*,\sigma}$ and $\underline{\E}_{\beta^*,\sigma}$ the probability and expectation with respect to the distribution $Y=\bX\beta^*+\epsilon$ with $\epsilon\sim \cN(0,\sigma^2\bI_n)$. Then, the minimax estimation risk satisfies
 \beq\label{eq:lower_fixed_design}
 \inf_{\widehat{\eta}}\sup_{\beta^* \in \mathbb{R}^p,\, \sigma\geq 0 }\underline{\E}_{\beta^*,\sigma}[(\widehat{\eta}-\eta[\beta^*,\sigma, \bX])^2]\geq \frac{1}{4}\ .
 \eeq
\end{prp}

\subsection{Knowledge of the noise level}

Throughout this manuscript, we assumed that the noise level $\sigma$ was unknown. As explained in the introduction, the situation is qualitatively different when  $\sigma$ is known. Let us briefly sketch the optimal convergence rates in this setting, still restricting ourselves to $p\geq n$. For any $k=1,\ldots, p$ define the maximal risk and the minimax risks
\[R(\widehat{\eta},k,\sigma):= \sup_{\beta\in \bbB_0[k]}\bbE_{\beta,\sigma}\left[\{\hat{\eta}-\eta(\beta,\sigma)\}^2\right]\ , \quad \quad R^*(k,\sigma):=\inf_{\hat{\eta}}R(\widehat{\eta},k,\sigma)\ ,\]
It follows from the minimax lower bounds for signal detection~\cite{2011_AS_Arias-Castro,2010_EJS_Ingster}, that for some $C>0$ (lower bounds in~\cite{2011_AS_Arias-Castro,2010_EJS_Ingster} are asymptotic but it is not difficult to adapt the arguments to obtain non-asymptotic bounds to the price of worse multiplicative constants),
\beq\label{eq:lower_known_sigma}
R^*(k,\sigma)\geq C \left(\frac{k}{n}\log(1+ \frac{p}{k^2})\right)^2\wedge \frac{1}{n}\ , 
\eeq
which is of order $[k\log(p)/n]^2\wedge n^{-1}$ except in the regime where $n$ is of order $p$ and where $k$ is of order $p^{1/2}$ in which case the logarithmic factors do not match. As for the upper bounds, since $\|Y\|_2^2/[\sigma^2+ \beta^{*T}\bSigma\beta^*]$ follows a $\chi^2$ distribution with $n$ degrees of freedom,  the estimator $\widehat{\eta}^{D,\sigma}:= 1- \frac{n\sigma^2}{\|Y\|_2^2}$ admits a quadratic risk (up to constants)  smaller  than $1/n$. This implies that the proportion of explained variation $\eta$ can be efficiently estimated for arbitrarily large $p$. 
For small $k$, one can use the Gauss-Lasso estimator based on  $\tilde{\beta}^{SL}$. 
Let $\hat{J}$ be the set of integers $j$ such that $\tilde{\beta}^{SL}\neq 0$ and define:
\[\widehat{\eta}^{GL,\sigma}:=  \frac{\|\boldsymbol{\Pi}_{\hat{J}}Y\|_2^2/n}{\sigma^2+ \|\boldsymbol{\Pi}_{\hat{J}}Y\|_2^2/n}\  \] 
where $\boldsymbol{\Pi}_{\hat{J}}=\boldsymbol{X}_{\hat{J}}(\boldsymbol{X}_{\hat{J}}^{T}\boldsymbol{X}_{\hat{J}})^{-1}\boldsymbol{X}_{\hat{J}}^{T}$
 is the orthogonal projector of $\mathbb{R}^{n}$ onto the space spanned by the columns of $\boldsymbol{X}_{\hat{J}}$. The Gauss-Lasso estimator was introduced to get an estimator of heritability in the sparse situation in a first version of this work
\cite{vg2016}.
Following the proof of Theorem 2.3 in  \cite{vg2016} we may obtain that, under Assumption \eqref{eq:hypo_max_dimension} and when $|\beta^*|_0=k$, 
\[
\E\left[( \widehat{\eta}^{GL,\sigma}- \eta)^{2} \right] \leq C' \frac{k^2\log^2(p)}{n^2}\frac{\lambda^2_{\max}(\bSigma)}{\lambda^2_{\min}(\bSigma)} \ .\]
In conclusion, the rate $[k\log(p)/n] \wedge n^{-1/2}$ is (up to a possible logarithmic multiplicative term)  optimal. These results contrast with the case of unknown $\sigma$ in two ways: (i) The optimal rate is order-wise faster when  $\sigma$ is known especially when $k$ is small ($n^{-1/2}$ versus $k\log(p)/n$) and when $k,p$ are larger ($p^{1/2}/n$ versus $n^{-1/2}$).  (ii) Since $\widehat{\eta}^{D,\sigma}$ and $\widehat{\eta}^{GL,\sigma}$ do not use the knowledge of $\bSigma$, adaptation to unknown covariance of the covariates is possible.

\subsection{Minimax confidence intervals}

In practice, one may not only be interested in the estimation of $\eta(\beta^*,\sigma)$, but also on building confidence intervals~\cite{janson2015eigenprism}. In the proof of Theorem \ref{prp:risk_dense} and in Proposition \ref{thm:eta_gauss_Lasso}, we obtain exponential concentration inequalities of $\widehat{\eta}^D(\bOmega)$ and $\widehat{\eta}^{SL}$ around $\beta^*$. 
This allows to get, for any $\alpha>0$ and any $k=1,\ldots, p$,  confidence intervals
\beqn
IC^{D}_{\alpha}&:=&\Big[\widehat{\eta}^D(\bOmega) \pm C(\alpha)\frac{\sqrt{p}}{n}\Big]\ ,\\ IC^{SL}_{\alpha,k}&:=&\Big[\widehat{\eta}^{SL} \pm C'(\alpha)\Big(\frac{1}{n^{1/2}}+ \frac{k\log(p)}{n}\frac{\lambda^2_{\max}(\bSigma)}{\lambda^2_{\min}(\bSigma)}\Big)\Big]\ ,
\eeqn
where $C(\alpha)$ and $C'(\alpha)$ are universal constants only depending on $\alpha$. When $p\geq n$, $IC^{D}_{\alpha}$ is honest over $\mathbb{R}^p$ in the sense that
\[
\inf_{\beta\in \mathbb{B}_0[p],\  \sigma>0}
\P_{\beta,\sigma}\left[ \eta \in \mathrm{IC}^D_{\alpha} \right]\geq 1-\alpha\ .
\]
For $p\geq n$ and if Assumption \eqref{eq:hypo_max_dimension} is satisfied, then the confidence interval $IC^{SL}_{\alpha,k}$ is honest  over $\mathbb{B}_0[k]$ in the sense that
\[
\inf_{\beta\in \mathbb{B}_0[k],\  \sigma>0}
\P_{\beta,\sigma}\left[ \eta \in \mathrm{IC}^{SL}_{\alpha,k} \right]\geq 1-\alpha\ .
\]

\medskip 

In high-dimensional linear regressions, there have been recent advances towards the construction of optimal confidence regions both for the unknown vector $\beta^*$~\cite{nickl_vandegeer} or low-dimensional functional of the parameters such as components $\beta^*_{i}$~\cite{2014_van_de_geer_confidence,2014javanmard,MR3153940,cai2015confidence} or $\sum_{i}\beta^*_{i}$~\cite{cai2015confidence}. Building on this line of work, it seems at hand to prove the minimax optimality of $\mathrm{IC}^D_{\alpha}$ and $\mathrm{IC}^{SL}_{\alpha,k}$, proving the existence of such honest confidence intervals. Of course, as already noticed when constructing our adaptive estimator, the choice of the constants $C(\alpha)$ and $C'(\alpha)$ are probably  far to be optimal in applications. 

 A further step would be to study  the problem of the construction (if possible) of adaptive confidence intervals. We leave those important questions for future research.

\section{Proof of the upper bounds}\label{sec:proof_minimax_upper}

\subsection{Proof of Theorem \ref{prp:risk_dense}}

\subsubsection{Some preliminary notation and deviation bounds}

Consider the spectral decomposition $\bSigma= \bO \bGamma\bO^T$ where $\bGamma$ is a diagonal matrix and $\bO$ is an orthogonal matrix. Define the matrix $\bZ=\bX \bO \bGamma^{-1/2}$ whose entries are  independent  standard normal variables. We denote  
\beq\label{eq:definition_rho0}
\rho:= \frac{\bGamma^{1/2} \bO^T \beta^*}{\|\bSigma^{1/2}\beta^*\|_2^2}\ .
\eeq

\bigskip 
In the following, we need to control the eigenvalues of $\bX\bX^T$. Define $\bA:= \bX\bX^T-tr[\bSigma]\bI_n$ and  note that  
\[
\bA=
 \bZ\bGamma\bZ^T-tr(\bGamma)\bI_n
= \sum_{j=1}^p \bGamma_{jj} \left(\bZ_{\bullet j}\bZ_{\bullet j}^T-\bI_n\right)\ , 
\]
where $\bZ_{\bullet j}$ stands for the $j$-th column of $\bZ$, so that $\bA$
is a weighted sum of centered Wishart matrices with parameters $(1,n)$. Extending the deviation inequalities of Davidson and Szarek~\cite{Davidson2001} for Wishart matrices to weighted sums of Wishart matrices, we obtain the following, which is proved in Appendix \ref{appendixB}.

\begin{lem}\label{lemma_concentration_spectre_A}
For any $t>0$,
\beq\label{eq:concentration_spectre_A}
\P\left[\|\bA\|_{op}\leq   2\sqrt{tr(\bSigma)} \|\bSigma\|_{op}^{1/2}\left(\sqrt{n}  + 10+\sqrt{2t}\right)+ 3 \|\bSigma\|_{op}\left[n + 100+ 2t\right]\right]\geq 1 - 2e^{-t}\ .
\eeq
As a consequence, for all $n\geq 20$, we get that 
\beq\label{eq:concentration_spectre_A2}
\P\left[\|\bA\|_{op}\leq 
 25 \|\bSigma\|_{op}(\sqrt{np}+n) \right]\geq 1 -2e^{-n} \ .
\eeq
\end{lem}
To control $\|\bA\|_{op}$, we could have applied  non-commutative Bernstein inequalities (Theorem 6.1.1 in~\cite{2015arXiv150101571T}). However, this approach would have produced additional logarithmic terms.

\subsubsection{\bf Analysis of $T$}

We decompose $T$ into four terms, whose deviations will be controlled independently.
\beqn
T&:=& T_a + T_b+ T_c+T_d\\
T_a& :=&  \frac{\beta^{*T}\bX^T\left(\bX \bX^T -tr(\bSigma)\bI_n \right)\bX \beta^*}{n^2}\\
T_b & := & \frac{Y^T\left[\{tr(\bSigma) -tr(\bX \bX^T)/n\}\bI_n \right]Y}{n^2}\\
T_c& :=&  \frac{\epsilon^T \bA\epsilon }{n^2} \ , \quad \quad  T_d :=  2\frac{ \epsilon^T\bA\bX \beta^* }{n^2} \ .
\eeqn 

\paragraph{Control of $T_a$.} The main term in the above decomposition is $T_a$. Since its control is quite technical, we only state a deviation bound for the time being. Subsections \ref{sec:dev_gaussian} and \ref{sec:concentration_T1} below are devoted to the proof of this lemma.

\begin{lem}\label{lem:concentration_T1}
For all $t\leq n^{1/3}$, we have
\beq\label{eq:concentration_T1}
\P\left[|T_{a}- (1+n^{-1})\|\bSigma\beta^*\|_2^2|\geq C\|\bSigma^{1/2}\beta^*\|_2^2  \|\bSigma\|_{op}\frac{\sqrt{pt}}{n}\right]\leq 2e^{-t}\ .
\eeq 
\end{lem}

\paragraph{Control of $T_b$, $T_c$, and $T_d$.} Since $tr(\bX \bX^T)\bI_n /n$ is a Gaussian quadratic form, we have by Lemma \ref{lem:chi_2} that 
\beq\label{eq:concentration_trace_A}
\P\left[\left|tr(\bSigma) -tr(\bX \bX^T) /n\right|\geq 8\|\bSigma\|_{op}\sqrt{\frac{p t}{n}} \right]\leq 2e^{-t}\  ,\quad \quad \forall t<np,
\eeq
where we used that $\|\bSigma\|_{F} \leq \sqrt{p} \|\bSigma\|_{op}$.
Also, $\|\bY\|^2_2/\var{y_{1}}$ follows a $\chi^2$ distribution with $n$ degrees of freedom, which implies 
$\P[\|\bY\|^2_2\geq \var{y_{1}}(n + 4\sqrt{nt})]\leq e^{-t}$ for all $t<n$. We conclude that 
for all $t<n$, 
\beq \label{eq:concentration_T2}
\P\left[|T_b|\geq 
40 \var{y_{1}}\|\bSigma\|_{op} \frac{\sqrt{p t}}{n^{3/2}} \right]\leq 3e^{-t}\ .
\eeq
The term $T_c$ is a Gaussian chaos of order $4$. We could apply the general deviation bounds from \cite{2015_adamszak}, but it is easier to work here conditionally to $\bX$. Conditionally to $\bX$, $T_c$ is a quadratic form with respect to $\epsilon$. By Lemma \ref{lem:chi_2}, 
\beqn
\P\left[\frac{n^2|T_c|}{\sigma^2} \geq  |tr(\bA)| + 8\|\bA\|_{op}\sqrt{nt}\right]\leq 2 e^{-t},\quad \quad \forall t<n\ ,
\eeqn
where we used $\|\bA\|_F\leq \sqrt{n}\|\bA\|_{op}$. Gathering this bound with the deviation inequality \eqref{eq:concentration_trace_A} for $tr(\bA)$,  the deviation inequality \eqref{eq:concentration_spectre_A2} for $\|\bA\|_{op}$, and using the fact that $p\geq n$, we conclude that, if $n\geq 20$, 
for all $t<n$, 
\beq \label{eq:concentration_T3}
\P\left[|T_c|\geq 208 \sigma^2 \|\bSigma\|_{op} \frac{\sqrt{p t}}{n^{3/2}} \right]\leq 6e^{-t}\ .
\eeq
Conditionally to $\bX$, $n^2T_d/(2\sigma)$ follows a centered normal distribution with variance $\|\bA\bX\beta^*\|_2^2\leq \|\bA\|_{op}^2\|\bX\beta^*\|_2^2$. Hence, 
\[
\P\left[\frac{n^2|T_d|}{2\sigma}\geq  \|\bA\|_{op}\|\bX\beta^*\|_2\sqrt{2t} \right]\leq 2e^{-t}\ ,\quad \quad \forall t>0\ .
\]
Then, $\|\bA\|_{op}$ is  controlled by \eqref{eq:concentration_spectre_A2} and $\|\bX\beta^*\|_2^2/\|\bSigma^{1/2}\beta^*\|_2^2$  follows a $\chi^2(n)$ distribution so that it can be controlled using Lemma \ref{lem:chi_2}. If $n\geq 20$, for all $t<n$, we arrive at 
\beq \label{eq:concentration_T4}
\P\left[T_d\geq 125 \sigma \|\bSigma^{1/2}\beta^*\|_2 \|\bSigma\|_{op} \frac{\sqrt{pt}}{n}\right]\leq 5e^{-t}.
\eeq

 \bigskip
 \noindent
 Gathering all the deviation inequalities (\ref{eq:concentration_T1}--\ref{eq:concentration_T4}), we obtain that for some constants  $C,C'>0$, if $n\geq 20$,
 \beq\label{eq:concentration_T}
 \P\left[\Big|T- \beta^{*T}\bSigma^2\beta^*\left(1+n^{-1}\right)\Big|\geq C \var{y}\|\bSigma\|_{op}\frac{\sqrt{pt}}{n}\right]\leq C'e^{-t}
 \eeq
 for all $t \leq n^{1/3}$.

\subsubsection{Analysis of $\widehat{V}$}

Since $\|Y\|_2^2/\var{y_1}$ follows a $\chi^2$ distribution with $n$ degrees of freedom, we obtain by Lemma \ref{lem:chi_2} that  $\P\left[|\|Y\|_2^2/n-\var{y_{1}}|\geq 4\var{y_{1}}\sqrt{t/n}\right]\leq 2e^{-t}$ for all $t<n$, so that using (\ref{eq:concentration_T}) and the fact that $p\geq n$, we conclude that, for all $t\leq n^{1/3}$, with probability larger than $1-(2+C')e^{-t}$, for some constant $C>0$,
\[
\Big|\widehat{V} - \frac{\beta^{*T} \bSigma^2 \beta^*}{\var{y}}\Big|\leq C\|\bSigma\|_{op}\frac{\sqrt{pt}}{n}\ , 
\]
and the first part of Theorem   \ref{prp:risk_dense} is proved.\\
Let us now turn to the second moment of $U:= \widehat{V} - \frac{\beta^{*T} \bSigma^2 \beta^*}{\var{y_i}}$. Define $\cA$ the event such that $|U|\leq C\|\bSigma\|_{op}\frac{\sqrt{pn^{1/3}}}{n}$, where $C$ is the same constant as in the above bound. The probability of $\cA$ is larger than $1-C'e^{-n^{1/3}}$ for some $C'>0$.
Then, the square risk decomposes as
\beqn 
\E\left[U^2\right]&\leq &\E\left[\mathbf{1}_{\cA}U^2\right] + 2 \E\Big[\mathbf{1}_{\cA^c}\Big(\frac{\beta^{*T} \bSigma^2 \beta^*}{\var{y}}\Big)^2\Big]+2 \E\left[\mathbf{1}_{\cA^c}(\widehat{V})^2\right]\\
& \leq & C\|\bSigma\|_{op}^2 \frac{p}{n^2}+ 2\P(\cA^c)\|\bSigma\|_{op}^2 + 2\left[\P(\cA^c)\right]^{1/2}\left[\E\left\{(\widehat{V})^4\right\}\right]^{1/2}\ ,
\eeqn 
where we have integrated the above deviation inequality in the last line. It remains to control the fourth moment of $\hat{V}$. We have
\[
 \widehat{V}\leq \frac{Y^T\left(\bX \bX^T -tr(\bX \bX^T)\bI_n /n\right)Y}{n\|Y\|_2^2}\leq \frac{\|\bA\|_{op}}{n}+ \frac{|tr(\bSigma)-tr(\bX\bX^T)/n|}{n}.
\]
Gathering the deviation inequalities \eqref{eq:concentration_spectre_A} and \eqref{eq:concentration_trace_A}, we derive that for some constants $\tilde{C}$ and $\tilde{C}'$, for any $t>0$, 
\[
 \P\left[\widehat{V}\geq \tilde{C} \|\bSigma\|_{op}\left(\sqrt{\frac{p}{n}}+ \frac{\sqrt{pt}}{n}+t \right)\right]\leq \tilde{C}'e^{-t}\ . 
\]
Integrating this deviation inequality, we obtain that $\left[\E\left\{(\widehat{V})^4\right\}\right]^{1/2}$ is upper bounded by a constant times $\|\bSigma\|_{op}^2p/n$. In conclusion, for some numerical constant numbers $C$ and $C'$,
\[
 \E\left[\Big|\widehat{V} - \frac{\beta^{*T} \bSigma^2 \beta^*}{\var{y_i}}\Big|^2\right]\leq C\|\bSigma\|_{op}^2 \frac{p}{n^2}\left[1+ ne^{-\frac{n^{1/3}}{2}} \right]\leq C'\|\bSigma\|_{op}^2 \frac{p}{n^2}\ . 
\]
and the second part of Theorem   \ref{prp:risk_dense} is proved.

\subsubsection{Deviation inequalities for Gaussian chaos}\label{sec:dev_gaussian}

We shall use  deviation inequalities for (non necessarily homogeneous) Gaussian chaos. Let us recall a recent result from Adamczak and Wolff~\cite{2015_adamszak}. In order to state this result, we need to introduce some new notation. 
 \\
Let $d$ and $q$ denote positive integers. Consider a $d$-indexed matrix $\bB=(b_{i_1,\ldots,i_d})^q_{i_1,\ldots,i_d=1}$. For $\bi=(i_1,\ldots,i_d)\in[q]^d$ and $I \subset [d]$ we write $\bi_I= (i_k)_{k\in I}$. Let $P_d$ be the set of partitions of $[d]$ into non empty disjoint subsets. Given a partition $\cJ=\{J_1,\ldots,J_k\}$, define the norm
\beq\label{eq:definition_norm_J}
\|\bB\|_{\cJ}= \sup\left\{\sum_{\bi\in[q]^d}b_{\bi}\prod_{l=1}^kx^{(l)}_{\bi_{J_l}}:\quad \|x^{(l)}\|_F\leq 1,\quad 1\leq l\leq k\right\} \ , 
\eeq
where $x^{(l)}$ is a $|J_l|$-indexed matrix and $\|x^{(l)}\|_F$ is its Frobenius norm. 
\\
Note that taking union of subsets in the partition increases the norm: given  $\cJ=\{J_1,\ldots,J_k\}$,  the partition $\cJ'=\{\{J_1\cup J_2\}, J_3 ,\ldots,J_{k}\}$ satisfies $\|\bB\|_{\cJ}\leq \|\bB\|_{\cJ'}$. Indeed, the 
$|J_1|+|J_2|$-dimensional matrix  $x^{(1)}\otimes x^{(2)}$ in the definition \eqref{eq:definition_norm_J} of $\|\bB\|_{\cJ}$ satisfies $\|x^{(1)}\otimes x^{(2)}\|_F\leq 1$.

\begin{prp}[Theorem 2 in \cite{2015_adamszak}]\label{prop:deviation_chaos}
Let $f:\mathbb{R}^q\mapsto \bbR$ be a polynomial of $q$ variables of total degree smaller or equal to $D$. For any integer  $d\geq 1$,  let $\bGamma^df$ denote the $d$-th derivative of $f$. Let $Z=(Z_1,\ldots,Z_q)$ denote a $q$-dimensional standard Gaussian vector. Then, for any $t>0$, 
\beq\label{eq:deviation_chaos}
\P\left[|f(Z) -\E\left(f(Z)\right)|\geq t \right]\leq 2\exp\left[- \min_{1\leq d\leq D}\min_{\cJ\in \P_d}\left(\frac{C^d t}{\|\E[\bGamma^d f(Z)]\|_{\cJ}}\right)^{\frac{2}{\#\cJ}} \right]\ ,
\eeq
where $C$ is a numerical constant.
\end{prp}

\subsubsection{Proof of Lemma \ref{lem:concentration_T1}}\label{sec:concentration_T1}

Define the variable
\beqn 
V:= \frac{n^2 T_a}{\|\bSigma^{1/2}\beta^*\|_2^2} &=& \rho^T \bZ^T\Big[(\bZ \bGamma \bZ^T) - tr(\bSigma)\bI_n\Big] \bZ \rho \ ,
\eeqn 
where we recall that $\rho$ is introduced in \eqref{eq:definition_rho0}.
First, we compute the expectation of $V$:
\beq\label{eq:exprectation_V}
\E[V]= (n^2+n) \rho^T \bGamma\rho = \|\bSigma\beta^*\|_2^2/ \|\bSigma^{1/2}\beta^*\|_2^2\ .
\eeq
$V$ is a polynom $f(\bZ)$ of degree 4 of the  $q=np$ independent standard Gaussian variables $\bZ=(Z_{i,j})_{1\leq i \leq n, 1\leq j \leq p})$ so that we can apply Proposition \ref{prop:deviation_chaos}. Since $V$ is the sum of an homogeneous polynom of degree $4$ and an homogeneous polynom of degree $2$, we only have to consider the derivatives of order $2$ and of order $4$ in \eqref{eq:deviation_chaos}, all the other terms (of order $1$ and $3$) being null. Write $V$ as 
\beq\label{eq:definition_V}
V = \sum_{i,j=1}^n\sum_{k,l,m=1}^{p} \bZ_{ik} \bZ_{jl}(\bZ_{im}\bZ_{jm}-\delta_{i,j}) \rho_{k}\rho_{l}\bGamma_{mm}\ ,
\eeq
where $\delta_{i,j}=1$ is the indicator function of $i=j$. We may express $V$ using
the four-indexed matrix $\bB$:
\[\bB_{(j_1,k_1),(j_2,k_2),(j_3,k_3),(j_4,k_4)}= \rho_{k_1}\rho_{k_4}\delta_{j_1,j_2}\delta_{j_3,j_4}\delta_{k_2,k_3} \bGamma_{k_2k_2}\ \]
as follows:
$$
V=f(\bZ):=\sum_{j_{1},\ldots,j_{4}=1}^{n}\sum_{k_{1},\ldots,k_{4}=1}^{p}\bB_{(j_1,k_1),(j_2,k_2),(j_3,k_3),(j_4,k_4)}\bZ_{j_{1}k_{1}}\left(\bZ_{j_{2}k_{2}}\bZ_{j_{3}k_{3}}-\delta_{j_{2},j_{3}}\right)\bZ_{j_{4}k_{4}}
$$
so that the expectation $\E[\bGamma^4f(\bZ)]$ of the  fourth derivative of $f(\bZ)$ is obtained by a symmetrization of $\bB$. More precisely, for any index $(i_1,\ldots, i_4)$ in $([n]\times [p])^{4}$, $\E[\bGamma^4f(\bZ)]_{i_1,i_2,i_3,i_4}= \sum_{\sigma }\bB_{i_{\sigma(1)},\ldots, i_{\sigma(4)}}$ where the sum runs over all permutations of $\{1,\ldots, 4\}$. Using the triangular inequality,  we shall obtain a bound on $\|\E[\bGamma^4f(\bZ)]\|_{\cJ}$ from a bound on $\|\bB\|_{\cJ}$.
Thus it suffices to bound $\|\bB\|^2_{\cJ}$ for all partitions $\cJ$. We start with $\cJ=\{1,2,3,4\}$. 
\[
\|\bB\|^2_{\{1,2,3,4\}} = \|\bB\|_F^2 \leq n^2 \|\rho\|_2^4 tr(\bGamma^2)= n^2 tr(\bSigma^2) 
\]
Let us now consider any partition $\cJ=\{J_1,J_2\}$ of size $2$. Without loss of generality, there exists $t \in \{1,2,3\}$ such that $t\in J_1$ and $t+1\in J_2$. Since each entry of $\bB$ contains a Dirac $\delta_{j_{t},j_{t+1}}$ or $\delta_{k_{t},k_{t+1}}$, there is a $n$ or $p$ factor  less in $\|\bB\|^2_{\cJ}$ in comparison to $\|\bB\|^2_{\{1,2,3,4\}}$, and we get 
$\|\bB\|^2_{J_1,J_2}\leq n\lambda^2_{\max}(\bSigma) (p\vee n)$. Let us illustrate this with $J_1=\{1\}$, $J_2=\{2,3,4\}$. By symmetry, $\|\bB\|_{\cJ}$ is achieved for $x^{(1)}_{j,k}= \rho_j n^{-1/2}$, and by Cauchy-Schwarz inequality, we obtain  $\|\bB\|^2_{\cJ}= n tr(\bSigma^2)$. \\
If now the partition $\cJ=\{J_1,\ldots, J_r\}$ has cardinality larger than $2$,  it was observed in the previous subsection that $\|\bB\|^2_{\cJ}\leq \|\bB\|^2_{J_1,\cup_{s>2}J_s}$. We have thus proved that, for all $t>0$,
\beq\label{eq:control_derivative_4}
\min_{\cJ\in P_4}\left(\frac{ t}{\|\E[\bGamma^4 f(\bZ)]\|_{\cJ}}\right)^{\frac{2}{\#\cJ}}\geq C \frac{t^2}{n^2 tr(\bSigma^2) }\wedge\left[ \bigwedge_{k=2,3,4}\left(\frac{t^2}{n\lambda^2_{\max}(\bSigma) (p\vee n)}\right)^{1/k}\right]\  .
\eeq

Let us now turn to the second derivative of $f(\bZ)$. Denote $\bB'= \E[\bGamma^2f(\bZ)]$. Coming back to the definition \eqref{eq:definition_V} of $V$, 
observe that $\bB'_{(j_1,k_1),(j_2,k_2)}$ is zero when $(j_1\neq j_2)$ because any term involving  $j_1$ and $j_2$ in $V$ contains exactly two terms with index $j_1$ and two terms with index $j_2$. Now, if $j_1=j_2$, the entries of $\bB'$ are bounded in absolute values by 
\[|\bB'_{(j,k_1),(j,k_2)}|\leq C   n  |\rho_{k_1}||\rho_{k_2}|(\bGamma_{k_1k_1}+ \bGamma_{k_2k_2})+  \bGamma_{k_1k_1} \ .  \]
As a consequence, 
\beqn 
\|\bB'\|^2_{\{1,2\}}&= &\|\bB'\|_F^2 \leq C n \sum_{k_1,k_2}n^2|\rho|_{k_1}^2|\rho|_{k_2}^2 \lambda^2_{\max}(\bSigma) + n\sum_{k_1} \bGamma^2_{k_1k_1}\\
&\leq &\left[n^3 \lambda^2_{\max}(\bSigma) + n tr(\bSigma^2)\right]\ ,
\eeqn 
since $\|\rho\|_2^2=1$ and $\sum \bGamma_{k,k}^2= tr(\bSigma^2)$.	
For $\cJ=\{\{1\},\{2\}\}$, $\|\bB^{'}\|_{\cJ}$ is the spectral norm of $\bB'$ when considered as $2$-dimensional $np\times np$ matrix. Since $\bB'$ can be seen as a block diagonal matrix, we obtain
\[\|\bB'\|_{\{1\},\{2\}}\leq C n   \lambda_{\max}(\bSigma)\ . \]
We arrive at 
\beq\label{eq:control_derivative_2}
\min_{\cJ\in P_2}\left(\frac{ t}{\|\E[\bGamma^2 f(\bZ)]\|_{\cJ}}\right)^{\frac{2}{\#\cJ}}\geq C \frac{t^2}{n^3 \lambda^{2}_{\max}(\bSigma)+ n tr(\bSigma^2) }\wedge \frac{t}{n\lambda_{\max}(\bSigma)}\  .
\eeq
Proposition \ref{prop:deviation_chaos} together with \eqref{eq:exprectation_V}, \eqref{eq:control_derivative_4}, and \eqref{eq:control_derivative_2} allows us to conclude.

\subsection{Proof of Proposition \ref{thm:eta_gauss_Lasso}}

This proposition is a consequence of the analysis of the square-root Lasso in \cite{2012_Sun}. We start with the decomposition
\beq\label{eq:decomposition_general}
\widehat{\eta}^{SL}-\eta_*= \widetilde{\sigma}^2_{SL}\left(\frac{n}{\| Y\|_2^2} - \frac{1}{\var{y_{1}}}\right)+  \frac{\widetilde{\sigma}^2_{SL}- \|\epsilon\|_2^2/n}{\var{y_{1}}}+\frac{\|\epsilon\|_2^2/n - \sigma^2}{\var{y_{1}}}  \ .
\eeq
By definition of the Lasso estimator, we have $n \widetilde{\sigma}^2_{SL}= \|Y-\bW \widehat{\beta}_{SL}\|_2^2\leq \|Y\|_2^2$. As a consequence, the first term in the above equation is smaller in absolute value than  $|1- \|Y\|_2^2/(n\var{y_{1}}) | $. Since $\|Y\|_2^2/\var{y_{1}}$ and $\|\epsilon\|_2^2/\sigma^2$ each follow a $\chi^2$ distribution with $n$ degrees of freedom, we have 
\[
\mathbb{E}\Big[\Big|\frac{\|Y\|_2^2}{n\var{y_{1}}} - 1\Big|^2\Big] = \frac{2}{n}\ \text{ and }\mathbb{E}\Big[\Big(\frac{\|\epsilon\|_2^2/n - \sigma^2}{\var{y_{1}}}\Big)^2\Big]= \frac{2\sigma^4}{n
(\var{y_{1}})^2}
 \leq \frac{2}{n}\ ,
\]
where we used  $\var{y_{1}} = \|\bSigma^{1/2}\beta^*\|_2^2 + \sigma^{2}$.
Let $\mathcal{A}$ be an event of large probability to be defined below. Since $|\widehat{\eta}^{SL}-\eta_*|\leq 1$, we deduce from \eqref{eq:decomposition_general} that 
\begin{equation}\label{eq:upper_bound_general}
 \mathbb{E}\Big[|\widehat{\eta}^{SL}-\eta_*|^2\Big]\leq  \P(\cA^c)+\frac{12}{n} + \frac{3\E\left[\big(\widetilde{\sigma}^2_{SL}- \|\epsilon\|_2^2/n\big)^2\mathbf{1}_{\cA}\right]}{\mathrm{Var}^2(y_1)}\ , 
\end{equation}
so that we only have to focus on $\P(\cA^c)$ and the difference $\widetilde{\sigma}^2_{SL}- \|\epsilon\|_2^2/n$. We need a few more notation. In the sequel,  $J_*$ denotes the support of $\beta^*$, that is the set of indices $i$ such that $\beta^{*}_{i}\neq 0$. 
For $T\subset [p]$ and $\xi>0$, the compatibility constant $\kappa[\xi,T;\bW]$ is defined by
\begin{equation*}
 \kappa[\xi,T;\bW]= \min_{u\in \mathcal{C}(\xi,T)} \left\{\frac{|T|^{1/2}\|\bW u\|_2}{\|u_T\|_1}\right\},\ \ \textrm{where}\ 
  \mathcal{C}(\xi,T)=\{u:\ \|u_{T^c}\|_1< \xi \|u_T\|_1\}\ . 
\end{equation*}
The compatibility constant, which quantifies how the design acts on the cone $\cC(\xi,T)$, arises in state  of the art results for the Lasso estimator~\cite{bickeltsy08, Kolt11,geer_condition}. 
We now define $\mathcal{A}$ as the event on which the following conditions are satisfied:
\begin{eqnarray}\label{eq:cond1}
 \|\bW^T\epsilon\|_{\infty}&\leq & 2 \sigma \sqrt{\log(p)}\ ,  \\ \label{eq:cond2}
 \frac{8}{9}n\sigma^2\leq \|\epsilon\|_2^2 &\leq & \frac{6}{5}n\sigma^2\ ,  \\ 
\label{eq:cond4} \kappa[5,J_*,\bW] &\geq & 16^{-1}\lambda^{1/2}_{\min}(\bSigma)/\lambda^{1/2}_{\max}(\bSigma) \ .
\end{eqnarray}

The first lemma provides a deterministic prediction error for the square-root estimator. It is a simplified version of Theorem 2 in \cite{2012_Sun} (the notation and normalizations are slightly different).

\begin{lem}[\cite{2012_Sun}]\label{lem:square_root_Lasso}
On the event $\mathcal{A}$, 
the design $\bW$ and the noise $\epsilon$ are such that 
\beq\label{eq:condition_1}
12 \lambda_0^2 |\beta^*|_0 \log(p) \leq   n \kappa^2[5,J_*;\bW]\quad \text{ and }\quad \|\bW^T\epsilon\|_{\infty}\leq \frac{\lambda_0\|\epsilon\|_2}{4\sqrt{n}}\ ,
\eeq
and the square root Lasso estimator satisfies
 \beq\label{eq:control:tilde_sigma}
 \max\Big[1 - \frac{\sqrt{n}\widetilde{\sigma}_{SL}}{\|\epsilon\|_2}, 1 - \frac{\|\epsilon\|_2}{\sqrt{n}\widetilde{\sigma}_{SL}}\Big]  \leq  3\lambda_0^2  \frac{  |\beta^*|_{0}}{n\kappa^2[5,J_*;\bW] }\leq 1/2\ .
 \eeq
 
\end{lem}

\begin{proof}[Proof.]
First, the second part of \eqref{eq:condition_1} is  enforced by Conditions \eqref{eq:cond1} and \eqref{eq:cond2} together with the definition of $\lambda_0$. The first part in \eqref{eq:condition_1} is a consequence of \eqref{eq:cond4} and hypothesis \eqref{eq:hypo_max_dimension}.
Then,  we apply  Theorem 2 of \cite{2012_Sun} to the estimator $\widetilde{\sigma}_{SL}$. 
Notice that the choice of $\lambda_0$ and \eqref{eq:condition_1} in the above lemma differs by a factor $\sqrt{n}$ from Theorem 2 in \cite{2012_Sun} because the design is normalized differently.
Using the notation of \cite{2012_Sun}, we fix $\xi=2$ so that Condition \eqref{eq:condition_1} implies that $\tau_*^2 \leq 1/4$ (we fix $\nu=1/2$ in \cite[][Eq.(16)]{2012_Sun}). Then,  the condition on $z^*$ in \cite[][Th.2]{2012_Sun} is a consequence of  the second part of \eqref{eq:condition_1}. The result follows. 
\end{proof}
It follows from  \eqref{eq:control:tilde_sigma} that, under $\cA$, 
\beqn 
\big(\widetilde{\sigma}^2_{SL}- \|\epsilon\|_2^2/n\big)^2&\leq& 9\frac{\|\epsilon\|_2^2}{n}\big(\widetilde{\sigma}_{SL}- \|\epsilon\|_2/\sqrt{n}\big)^2\\
&\leq &  C \frac{\|\epsilon\|_2^4}{n^2}\lambda_0^4  \frac{  |\beta^*|^2_{0}}{n^2\kappa^4[5,J_*;\bW] }\\
&\leq & C' \sigma^4  \frac{|\beta^*|^2_0\log^2(p)}{n^2}\frac{\lambda^2_{\max}(\bSigma)}{\lambda^2_{\min}(\bSigma)}\ ,
\eeqn 
where we used the conditions \eqref{eq:cond2} and \eqref{eq:cond4} in the last line. In view of \eqref{eq:upper_bound_general}, Proposition \ref{thm:eta_gauss_Lasso}  follows finally from  the following lemma. 

\begin{lem}\label{lem:control_probability}
Under Assumption \eqref{eq:hypo_max_dimension}, we have for some positive constants $C$, $C'$, and $C"$ that 
\[
\P[\cA^c]\leq  C [pe^{-C'(n\wedge p)} + p^{-1}]\leq \frac{C"}{n} \ .
\]
\end{lem}

\begin{proof}[Proof of Lemma \ref{lem:control_probability}]
We control the probability of each event defined by  (\ref{eq:cond1}), (\ref{eq:cond2}), and (\ref{eq:cond4}).
Conditionally to $\bW$, $\|\bW^T\epsilon\|_{\infty}/\sigma$ is distributed as a supremum of $p$ independent standard Gaussian variables. Applying an union bound over all variables $(\bW^T\epsilon)_i$, we derive that  
\[\P\left[\|\bW^T\epsilon\|_{\infty}\leq  \sigma \sqrt{4\log(p)}\right]\geq 1 - p^{-1}\ .\] 

Turning to \eqref{eq:cond2}, we see that $\|\epsilon\|_2^2/\sigma^2$ follows a $\chi^2$ distribution with $n$ degrees of freedom. By Lemma \ref{lem:chi_2}, we obtain 
 \[\P\left[\frac{8}{9}n\sigma^2\leq \|\epsilon\|_2^2 \leq  \frac{6}{5}n\sigma^2\right]\geq 1- 2 e^{-Cn}\ , \]
for some positive constant $C>0$.

Finally, we need to control the compatibility constant $\kappa\big[5,J_*;\bW \big]$. As the compatibility constant is larger than restricted eigenvalues, we can readily apply the results of \cite{10:RWG_restricted}. In particular, their Corollary 1 entails that, with probability larger than $1-c_1\exp[-c_2n]$ for some $c_{1}>0$ and $c_{2}>0$, 
\begin{eqnarray*}
\kappa\big[5,J_*;\bX/(\sqrt{n} \lambda_{\min}^{1/2}(\bSigma) )\big]\geq 1/8 \ ,
\end{eqnarray*}
as long as $|J_*| \log(p)< c_3   n$. The latter condition is satisfied by hypothesis \eqref{eq:hypo_max_dimension}. 
Coming back to the definition of $\bW$ and of the compatibility constant, we have
\[
\kappa\big[5,J_*;\bW\big] \geq \kappa\big[5,J_*;\bX/(\sqrt{n} \lambda_{\min}^{1/2}(\bSigma) )\big] \frac{\sqrt{n}\lambda_{\min}^{1/2}(\bSigma)}{\max\|\bX_{\bullet i}\|_2} \ . 
\]
Since, for all $i$, $\bSigma_{ii}$ is larger than 
$\lambda_{\min}(\bSigma)$, we can apply Lemma \ref{lem:chi_2} to get
\beq
\P\left[\min_{i=1,\ldots,p}\|\bX_{\bullet i}\|_2^2\geq n\lambda_{\min}(\bSigma)/2\right]\geq 1 - pe^{-Cn}
\eeq
for some positive constant $C>0$.  Finally, Assumption \eqref{eq:hypo_max_dimension} enforces that $\log(p)$ is small compared to $n$ so that $pe^{-Cn}$ is smaller than $C'/n$ for some positive constant $C'$.
\end{proof}

\subsection{Proof of Theorem \ref{prp:adaptation}}\label{sec:proof_prpe_adaptation}

Notice first that we always have
$
\left| \widehat{\eta}_T^D (\bOmega) - \eta \right| \leq \left| \widehat{\eta}^D (\bOmega) - \eta \right|
$
and $\left| \widehat{\eta}_T^D (\bOmega) - \eta \right| \leq 1$.
\\

\noindent
We first consider the case where $(\beta^*, \sigma)$ is arbitrary. The difference $\widehat{\eta}^A - \eta$ decomposes as
$$
\widehat{\eta}^A - \eta=\left(\widehat{\eta}^{SL} - \widehat{\eta}^D (\bOmega)\right)\mathbf{1}_{\widehat{\eta}^A= \widehat{\eta}^{SL}}+\left(\widehat{\eta}^D (\bOmega) - \eta\right)\mathbf{1}_{\widehat{\eta}^A= \widehat{\eta}^{SL}}+\left(\widehat{\eta}_T^D (\bOmega) - \eta\right)\mathbf{1}_{\widehat{\eta}^A= \widehat{\eta}_T^D (\bOmega)}\ .
$$
The difference $\widehat{\eta}^D -\eta$ is controlled thanks to Corollary \ref{cor:risk_dense}, whereas the difference $\widehat{\eta}^{SL} - \widehat{\eta}^D (\bOmega)$ is small when $\widehat{\eta}^A= \widehat{\eta}^{SL}$ by definition of $\widehat{\eta}^A$.
\beqn 
\E\left[\left(\widehat{\eta}^A - \eta\right)^2\right]&\leq& 3\E\left[\left(\widehat{\eta}^D (\bOmega) - \eta\right)^2\right]+3\E\left[\left(\widehat{\eta}_T^D (\bOmega) - \eta\right)^2\right]  +  3\E\left[\left(\widehat{\eta}^{SL} - \widehat{\eta}^D (\bOmega)\right)^2\mathbf{1}_{\widehat{\eta}^A= \widehat{\eta}^{SL}}\right]\\
&\leq &6\E\left[\left(\widehat{\eta}^D (\bOmega) - \eta\right)^2\right] + 3\E\left[\left(\widehat{\eta}^{SL} - \widehat{\eta}^D (\bOmega)\right)^2\mathbf{1}_{\widehat{\eta}^A= \widehat{\eta}^{SL}}\right]\\
& \leq &  C \frac{p}{n^2}+2c_{0}^{2} \frac{p\log(p)}{n^2}\ ,
\eeqn 
where we used the definition of $\widehat{\eta}^A= \widehat{\eta}^{SL}$ in the last line. 
Thus, considering that the risk of the estimator is bounded by $1$, it is possible to choose the numerical constant $C$ such that Theorem \ref{prp:adaptation} holds true if
$\beta^*$ is $k$-sparse with $k$ such that (\ref{eq:hypo_max_dimension}) does not hold.

Assume now that $\beta^*$ is $k$-sparse with $k$ such that (\ref{eq:hypo_max_dimension}) holds.
We start from the decomposition 
\[\E\left[\left(\widehat{\eta}^A - \eta\right)^2\right]= \E\left[\left(\widehat{\eta}^{SL} - \eta\right)^2\mathbf{1}_{\widehat{\eta}^A= \widehat{\eta}^{SL}}\right] + \E\left[\left(\widehat{\eta}_T^{D} (\bOmega)- \eta\right)^2\mathbf{1}_{\widehat{\eta}^A= \widehat{\eta}_T^{D} (\bOmega)}\right] \ . 
\]
In this sparse setting, the risk of $\widehat{\eta}^{SL}$ is minimax optimal but the risk of $\widehat{\eta}_T^{D} (\bOmega)$ is possibly quite large. We have to work around the event $\widehat{\eta}^A= \widehat{\eta}_T^{D} (\bOmega)$. This event can only be achieved if either we have $|\widehat{\eta}_T^{D} (\bOmega) - \eta|\geq c_0\sqrt{p\log(p)}/(2n)$ or if we have simultaneously $|\widehat{\eta}^{SL} - \eta|\geq c_0\sqrt{p\log(p)}/(2n)$ and 
$|\widehat{\eta}_T^{D} (\bOmega)- \eta|\leq c_0\sqrt{p\log(p)}/(2n)$. Under this last possibily, observe that 
$|\widehat{\eta}_T^D (\bOmega)-\eta|\leq |\widehat{\eta}^{SL}-\eta|$. Thus, we obtain
\beqn 
\E\left[\left(\widehat{\eta}^A - \eta\right)^2\right]
& \leq  & \E\left[\left(\widehat{\eta}^{SL} - \eta\right)^2\mathbf{1}_{\widehat{\eta}^A= \widehat{\eta}^{SL}}\right] 
+ \E\left[\left(\widehat{\eta}_T^{D} (\bOmega)- \eta\right)^2\mathbf{1}_{|\widehat{\eta}_T^{D} (\bOmega)- \eta|\geq c_0\sqrt{p\log(p)}/(2n)}\right] \\
&&+ \E\left[\left(\widehat{\eta}_T^{D} (\bOmega)- \eta\right)^2\mathbf{1}_{|\widehat{\eta}^{SL} - \eta|\geq c_0\sqrt{p\log(p)}/(2n)}
\mathbf{1}_{|\widehat{\eta}_T^{D} (\bOmega)- \eta|\leq c_0\sqrt{p\log(p)}/(2n)}\right]\\
&\leq& 2 \E\left[\left(\widehat{\eta}^{SL} - \eta\right)^2\right] + \P\left[|\widehat{\eta}_T^{D} (\bOmega)- \eta|\geq c_0\sqrt{p\log(p)}/(2n)\right]\ .
\eeqn 
The risk $\E[\left(\widehat{\eta}^{SL} - \eta\right)^2]$ is bounded thanks to Proposition \ref{thm:eta_gauss_Lasso} whereas the deviation inequality\[\P\left[|\widehat{\eta}^D (\bOmega)- \eta|\geq c_0\sqrt{p\log(p)}/(2n)\right]\] is smaller than $C_{2}/p$ by Theorem \ref{prp:risk_dense}. Together with the fact that $p\geq n$,  we have proved that when  $\beta^*$ is $k$-sparse with $k$ such that \ref{eq:hypo_max_dimension} holds,
\[
\E\left[\left(\widehat{\eta}^A - \eta\right)^2\right]
\leq  C \left[\frac{1}{n}+ \frac{k^2\log^2(p)}{n^2} \frac{\lambda^2_{\max}(\bSigma)}{\lambda^2_{\min}(\bSigma)}\right]\  .\]
Theorem \ref{prp:adaptation} follows.

\subsection{Analysis of the plug-in method}

\begin{proof}[Proof of Proposition \ref{prp:upper_risk_unknown_Sigma}]

We first note that the estimator $\widehat{\eta}^D(\widehat{\bOmega})$ is built using the following linear regression model
\[
 Y= \Big[\bX\widehat{\bOmega}^{1/2}\Big] \Big[ \widehat{\bOmega}^{-1/2}\beta^* \Big]+ \epsilon \ .
\]
It then follows from Theorem \ref{prp:risk_dense}  in Section \ref{sec:dense_regime} that 
 $\widehat{\eta}^D ( \widehat{\bOmega})$ is an 
 estimator of $\|\widehat{\bOmega}^{1/2}\bSigma \beta^*\|_2^2/\var{y_1}$. More precisely, we have 
 \[
\P\left[\big|\widehat{\eta}^D(\widehat{\bOmega})- \frac{\|\widehat{\bOmega}^{1/2}\bSigma\beta^*\|_2^2}{\var{y_1}}\big|\geq C_1 \|\bSigma\|_{op} \|\widehat{\bOmega}\|_{op}\frac{\sqrt{pt}}{n} \Big|\bX^{(2)}\right]\leq C_2 e^{-t}
\]
for all $t<n^{1/3}$. Decomposing the difference $\widehat{\eta}^D (\widehat{\bOmega})-\eta$ into
\beqn 
\widehat{\eta}^D(\widehat{\bOmega})- \eta = \widehat{\eta}^D(\widehat{\bOmega}) - \frac{\|\widehat{\bOmega}^{1/2}\bSigma\beta^*\|_2^2 }{\var{y_1}}+ \frac{\|\widehat{\bOmega}^{1/2}\bSigma\beta^*\|_2^2  }{\var{y_1}} - \eta \ ,
\eeqn
we only have to consider the second term
\beqn 
\Big|\frac{\|\widehat{\bOmega}^{1/2}\bSigma\beta^*\|_2^2 }{\var{y_1}} - \eta\Big|&=& 
\Big| \frac{\beta^{*T} \bSigma^{1/2}(\bSigma^{1/2}\widehat{\bOmega}\bSigma^{1/2} -\bI_p )\bSigma^{1/2}\beta^*}{\var{y_1}}\Big|\\
&\leq & \|\bSigma^{1/2}\widehat{\bOmega}\bSigma^{1/2} -\bI_p\|_{op}\\&\leq& \|\bSigma\|_{op} \|\widehat{\bOmega}-\bOmega\|_{op}
\eeqn
where we used in the second line that $\beta^{*T} \bSigma \beta^* / \var{y_1} \leq 1$.
\end{proof}

\begin{proof}[Proof of Corollary \ref{cor:dense-plug_in}]
 Define the event $\cB$ such that inequality \eqref{eq:control_tilde_omega} is true. Assume first that 
 \beq\label{eq:condition_consistensy}
 CM^2M_{1}^{ 2}\frac{\sqrt{p}}{n}\leq (2M_1)^{-1} \ ,
 \eeq
 where $C$ is the numerical constant in \eqref{eq:control_tilde_omega}.
 \beqn 
 \E\left[\left\{\widehat{\eta}_T^D(\widehat{\bOmega}_{CL})-\eta\right\}^2\right]& =&  \E\left[\left\{\widehat{\eta}_T^D(\widehat{\bOmega}_{CL})-\eta\right\}^2\mathbf{1}_{\cB}\right]+ \E\left[\left\{\widehat{\eta}_T^D(\widehat{\bOmega}_{CL})-\eta\right\}^2\mathbf{1}_{\cB^c}\right]\\
 &\leq &\E\left[\left\{\widehat{\eta}_T(\widehat{\bOmega}_{CL})-\eta\right\}^2\mathbf{1}_{\cB}\right] +\P[\cB^c]\ , 
 \eeqn 
 where we used that $\widehat{\eta}_T^D(\widehat{\bOmega}_{CL})$ belongs to $[0,1]$.
 By Lemma \ref{lem:estimation_omega}, $\P[\cB^c]\leq 4/p$. Under event $\cB$, $\lambda_{\min}(\widehat{\bOmega}_{CL})\geq \lambda_{\min}(\bOmega)- \|\widehat{\bOmega}_{CL}- \bOmega\|_{op}\geq (2M_1)^{-1}$ and $\widehat{\bOmega}_{CL}$ is therefore non-singular. Plugging \eqref{eq:control_tilde_omega} in Proposition \ref{prp:upper_risk_unknown_Sigma} and integrating the deviation bound with respect to $t>0$, we get that for some numerical constant $C$,
 \[
 \E\left[\left\{\widehat{\eta}_T^D(\widehat{\bOmega}_{CL})-\eta\right\}^2\right]\leq  C M^{4}M_{1}^{ 6}\frac{p}{n^2}\ .
 \]
If now,  \eqref{eq:condition_consistensy} is not satisfied, we just use that 
since the thresholded estimator is $\widehat{\eta}_T^D(\widehat{\bOmega}_{CL})$ belongs to $[0,1]$, the risk is always smaller than $1$, which is smaller than $CM^2M_{1}^{3}\frac{\sqrt{p}}{n}$ . 

 \end{proof}
 \begin{proof}[Proof of Corollary \ref{cor:plug_in}]
 
 In order to show \eqref{eq:risk_A_CL}, we follow the same steps as for proving Proposition \ref{prp:adaptation}, the only difference being that we need to prove that $\P[\big|\widetilde{\eta}_T^D(\widehat{\bOmega}_{CL}) -\eta\big|\geq c_0(M,M_1)\sqrt{p\log p}/(2n)]$ is larger than $C/p$ for some $C>0$. As above, we consider two cases whether \eqref{eq:condition_consistensy} is satisfied or not. If Condition \eqref{eq:condition_consistensy} is satisfied, we use Proposition \ref{prp:upper_risk_unknown_Sigma} with $t=\log(p)$ and the event $\cB$ to prove that 
 \[
 \P\left[\big|\widehat{\eta}_T^D(\widehat{\bOmega})- \eta \big|\geq  2C_1M_1^2 \frac{\sqrt{p\log(p)}}{n}+ C_{3}  M^2 M_1^{3} \frac{\sqrt{p}}{n} \right]\leq \frac{C_2 +4 }{p} \ .
 \]
 If Condition \eqref{eq:condition_consistensy} is not satisfied, we again use that
 \[|\widehat{\eta}_T^D(\widehat{\bOmega}) -\eta|\leq 1 \leq 2CM^2M_1^{ 3}\frac{\sqrt{p}}{n} . \]

\end{proof}

\section{Proofs of the minimax lower bounds}\label{sec:proof_minimax_lower}

\subsection{Proof of Proposition \ref{prp:lower_minimax}}
\label{sec:lower}

\subsubsection{Proof of the parametric rate $R^*(k)\geq R^*(1)\geq C n^{-1}$}

First, we prove that $\eta$ cannot be estimated faster than the parametric rate $n^{-1/2}$. 
Fix $\sigma=1$, $\beta^*_1= (1,0,\ldots, 0)^T$ and $\beta^*_2= (1+ n^{-1/2},0,\ldots, 0)^T$. Then $\eta_1=\eta(\beta_1^*,\sigma)= 1/2$ and 
$\eta_2=\eta(\beta_2^*,\sigma)\geq 1/2+n^{-1/2}/4$. Denoting $\bbK(\P_{\beta_1^*,\sigma}; \P_{\beta_2^*,\sigma})$ the Kullback-Leibler divergence between $\P_{\beta_1^*,\sigma}$ and $\P_{\beta_2^*,\sigma}$, we have 
\[\bbK(\P_{\beta_1^*,\sigma}; \P_{\beta_2^*,\sigma})= \E\left[\frac{\|\bX(\beta_1^*-\beta_2^*)\|_2^2}{2}\right]= \frac{1}{2}\ . \]
Using Pinsker's inequality, we provide a lower bound of $R^*(1)$ in terms of $\bbK(\P_{\beta_1^*,\sigma}; \P_{\beta_2^*,\sigma})$ and $(\eta_1-\eta_2)^2$ as follows:
\beqn 
R^*(1)&\geq& \inf_{\widehat{\eta}}\E_{\beta_1^*,\sigma}\left[(\widehat{\eta} -\eta_1)^2\right]\bigvee \E_{\beta_2^*,\sigma}\left[(\widehat{\eta} -\eta_2)^2\right]\\
&\geq & \frac{(\eta_2-\eta_1)^2}{4}  \inf_{\widehat{\eta}} \P_{\beta_1^*,\sigma}\left[\widehat{\eta}\geq (\eta_1+\eta_2)/2   \right]\bigvee \P_{\beta_2^*,\sigma}\left[\widehat{\eta}\leq (\eta_1+\eta_2)/2\right]\\
&\geq & \frac{(\eta_2-\eta_1)^2}{8}  \inf_{\cA} \P_{\beta_1^*,\sigma}(\cA) + \P_{\beta_2^*,\sigma}(\cA^c)\ ,\quad \quad \text{where $\cA$ is any measurable event}\\
&\geq & \frac{(\eta_2-\eta_1)^2}{8} \left[ 1- \|\P_{\beta_1^*,\sigma} - \P_{\beta_2^*,\sigma}\|_{TV}\right]\\
&\geq & \frac{(\eta_2-\eta_1)^2}{8} \left[ 1- 2^{-1/2}\bbK^{1/2}(\P_{\beta_1^*,\sigma}; \P_{\beta_2^*,\sigma})\right]\ , \quad \quad \text{ by Pinsker's inequality}\\
&\geq & \frac{(\eta_2-\eta_1)^2}{16}\geq \frac{1}{16^2n}\ ,
\eeqn
which concludes the proof.

\subsubsection{Proof of $R^*(k)\geq C \left\{\left[\frac{k}{n}\log\left(1+ \frac{p}{k^2}\vee \sqrt{\frac{p}{k^2}}\right)\right]^{2}\wedge 1 \right\}$}

In this proof, we follow the standard strategy of reducing the heritability estimation problem to a detection problem, thereby taking advantage on available bounds of \cite{2010_AS_Verzelen}. We could simply derive Proposition \ref{prp:lower_minimax} from Theorem 4.3 in \cite{2010_AS_Verzelen}, but we prefer to detail the arguments as a first step towards the minimax lower bounds for adaptation problems. 

\medskip 
\noindent
Denote $\P_0$ the distribution of ($Y,\bX$) when $\beta^*=0$ and $\sigma=1$. Let $\rho>0$ be a positive quantity that will be fixed later. Also, denote $\cB$ the collection of all vectors $\beta\in \mathbb{R}^p$  with exactly $k$ non-zero components that are either equal to $\tfrac{\rho}{[(1+\rho^2)k]^{1/2}}$ or $-\tfrac{\rho}{[(1+\rho^2)k]^{1/2}}$. Defining $\sigma_{\rho}^{2}:= (1+\rho^2)^{-1}$, we obtain, for all $\beta\in \cB$, $\eta(\beta,\sigma_{\rho})= \rho^2/(1+\rho^2)$. Following the beaten path of Le Cam's approach, we consider  $\mu$ the uniform measure on $\cB$ and denote $\mathbf{P}_{\mu}$ the mixture probability measure 
\beq \label{eq:definition_pmu}
\mathbf{P}_{\mu}= \int_{\cB}\mathbb{P}_{\beta,\sigma_{\rho}}\, \mu(d\beta)
\eeq

\bigskip 

Let $\widehat{\eta}$ be any estimator of $\eta$. The minimax risk $R^*(k)$ is obviously lower bounded as follows:
\beqn 
R^*(k)&\geq&  \E_0\left[\widehat{\eta}^2\right]\bigvee \vee_{\beta\in \cB}\E_{\beta,\sigma_\rho}\left[\left(\widehat{\eta}- \frac{\rho^2}{1+\rho^2}\right)^2\right]\\
&\geq &\frac{1}{2}\left[\E_0\left[\widehat{\eta}^2\right]+ \mathbf{E}_{\mu}\left[\left(\widehat{\eta}- \frac{\rho^2}{1+\rho^2}\right)^2\right]\right]\\
&\geq & \frac{\rho^4}{8(1+\rho^2)^2}\left[\P_0\left[\widehat{\eta}>\frac{\rho^2}{2(1+\rho^2)}\right]+ \mathbf{P}_\mu\left[\widehat{\eta}\leq \frac{\rho^2}{2(1+\rho^2)}\right]\right]\ .
\eeqn 
Defining the test statistic $\widehat{T}:=\mathbf{1}\{\widehat{\eta}>\rho^2/[2(1+\rho^2)]\}$, one recognizes in the bound above the sum of type I and type II errors of the test $\P_0$ versus $\mathbf{P}_{\mu}$. We arrive at 
\begin{eqnarray}
  \nonumber 
R^*(k)&\geq& \frac{\rho^4}{8(1+\rho^2)^2}\left[\P_0[\widehat{T}=1] + \mathbf{P}_{\mu}[\widehat{T}=0]\right]\\ \nonumber 
&\geq & \frac{\rho^4}{8(1+\rho^2)^2}\left[1 - |\P_0(\widehat{T}=0) - \mathbf{P}_{\mu}(\widehat{T}=0)|\right]\\ \nonumber 
&\geq & \frac{\rho^4}{8(1+\rho^2)^2}\left[1 - \E_0|L_{\mu}-1|\right]\quad\quad \text{where $\mathbb{L}_{\mu}=\frac{d\mathbf{P}_{\mu}}{d\P_0}$ } \\ \label{eq:lower_proba_minimax}
&\geq & \frac{\rho^4}{8(1+\rho^2)^2}\left[1 -  \left(\chi^2(\mathbf{P}_{\mu},\P_0)\right)^{1/2} \right]\ , \quad \quad \text{(by Cauchy-Schwarz inequality)}
\end{eqnarray}
where $\chi^2(\mathbf{P}_{\mu},\P_0)=\E_0[(L_{\mu}-1)^2]$ stands for the $\chi^2$ distance between probability distributions. As a consequence, we only need to bound the $\chi^2$ distance between $\mathbf{P}_{\mu}$ and $\P_0$. Fortunately, this distance has been controlled in \cite{2010_AS_Verzelen} (take  $v=1$, $\var{y}=1$ in \cite[p.741, line 14]{2010_AS_Verzelen} and note that $k\lambda^2=\rho^2/(1+\rho^2)$).

\begin{lem}[\cite{2010_AS_Verzelen}] \label{lem:chi2_upper}
We have
\beq\label{eq:uper_bound_chi2}
 \chi^2(\mathbf{P}_{\mu},\P_0)\leq  \exp\left[k\log\left(1+ \frac{k}{p}\left(\cosh\left(\frac{n\rho^2}{k}\right) - 1\right)\right)	\right] - \frac{1}{2}\ . 
 \eeq
 \end{lem}

Let us fix $\rho^2$ in such a way that
\beq\label{eq:definition_rho}
\frac{n\rho^2}{k} = \log\left[1+ \frac{p}{k^2}\log(5/4)+ \sqrt{(1+\frac{p}{k^2}\log(5/4))^2-1}\right]\ .
\eeq
Using the classical equality $\cosh(\log(1+x+\sqrt{x^2+2x}))= 1+ x$ for $x\geq0$, we arrive at 
\[\chi^2(\mathbf{P}_{\mu},\P_0)\leq \exp\left[k\log(1+ \log(5/4)/k)\right]-1/2\leq 3/4\ ,\]
which, together with \eqref{eq:lower_proba_minimax}, implies
\[ R^*(k)\geq \frac{\rho^4}{8(1+\rho^2)^2}(1-(3/4)^{1/2})\ . \]
Since $\log(1+ux)\geq u\log(1+x)$ for any $u\in (0,1)$ and $x>0$, we derive from \eqref{eq:definition_rho} that 
\[\rho^2\geq \log\left(\frac54\right) \left[\frac{k}{n}\log\left(1+ \frac{p}{k^2}\vee \sqrt{\frac{p}{k^2}}\right)\right]^{2} \ .\]
But 
$$
\frac{\rho^2}{1+\rho^2} \geq \frac{\rho^2}{2} \wedge 1\;,
$$
which concludes the proof.

\subsection{Proof of Proposition \ref{prp:lower_adaptation}}

Define the quantity  $\rho>0$ by 
\beq\label{eq:def_rho}
\rho^2:=  \frac{a\sqrt{p\log p}}{4n}\ .
\eeq
We consider $\mu$, $\mathbf{P}_{\mu}$, $\mathbf{E}_{\mu}$ as introduced in the proof of Proposition \ref{prp:lower_minimax}.  
\\
Let  $\widehat{\eta}$ be a given estimator. Define
$$
R:=n\sqrt{\frac{n}{p}}\mathbf{E}_{0}[\widehat{\eta}^{2}] + \frac{n^{2}}{p\log p}\mathbf{E}_{\mu}\left[\left(\widehat{\eta}-\frac{\rho^{2}}{1+\rho^{2}}\right)^{2}\right].
$$
Then,
$$
\frac{R(\widehat{\eta},1)}{\frac{1}{n}\sqrt{\frac{p}{n}}} + \frac{R(\widehat{\eta},k)}{\frac{p\log p}{n^{2}}}\geq R.
$$
Now define the event $\cA(\widehat{\eta}):= \{\widehat{\eta}\geq \rho^2/[2(1+\rho^2)]\}$. Then, one has
\[
n\sqrt{\frac{n}{p}}\mathbf{E}_{0}[\widehat{\eta}^{2}] \geq n\sqrt{\frac{n}{p}}\P_0[\cA(\widehat{\eta})]\frac{\rho^4}{4(1+\rho^2)^2}   \geq \frac{ a^{2}}{4^{4}}  \sqrt{\frac{p}{n}}\log p \P_0[\cA(\widehat{\eta})]\ .                                                                                                                                                                                                                                                      \]
 Similarly, $\mathbf{E}_{\mu}(\widehat{\eta}- \rho^2/(1+\rho^2))^2\geq \mathbf{P}_{\mu}(\cA^c(\widehat{\eta})) \rho^4/4(1+\rho^2)^2 \geq \frac{ a^{2}}{4^{4}}  \frac{p\log p}{n^{2}}\mathbf{P}_{\mu}(\cA^c(\widehat{\eta}))$ so that 
 \[
  R\geq \frac{ a^{2}}{4^{4}}  \inf_{\cA} \left\{\P_0\left[\mathcal{A}\right]\sqrt{\frac{p}{n}}\log p  + \mathbf{P}_{\mu}\left[\mathcal{A}^c\right]\right\} \ \ , 
 \]
were the infimum is taken over all measurable events $\mathcal{A}$. Restricting the events $\cA$ to have  small probability, we arrive at 
\beq \label{eq:lower_R}
  R\geq  \frac{ a^{2}}{4^{4}} \left[1 \bigwedge \inf_{\cA,\ \P_0[\cA]\leq  \sqrt{n}/(\sqrt{p}\log p) }\mathbf{P}_{\mu}\left[\mathcal{A}^c\right]\right\}\ , 
\eeq 
so that it suffices to obtain a uniform lower bound $\mathbf{P}_{\mu}\left[\mathcal{A}^c\right]$ over events $\cA$ of small $\P_0$-probability. 
\begin{eqnarray}\nonumber
 \mathbf{P}_{\mu}(A^c)&\geq& 1- \P_0(\mathcal{A}) - |\mathbf{P}_{\mu}(\mathcal{A})-\P_0(\mathcal{A})|\\ \nonumber
 &\geq & 1- \P_0(\mathcal{A}) - |\mathbf{E}_{0} \left[(\mathbb{L}_{\mu}-1)\mathbf{1}_{\mathcal{A}}\right]|\quad\quad \text{where $\mathbb{L}_{\mu}=\frac{d\mathbf{P}_{\mu}}{d\P_0}$ } \\ \label{eq:lower_proba_ac}
 &\geq & 1- \P_0(\mathcal{A}) - \left(\P_0[\mathcal{A}]\chi^2(\mathbf{P}_{\mu},\P_0)\right)^{1/2}\ . \quad \quad \text{(by Cauchy-Schwarz inequality)}
\end{eqnarray}
Define $x= \frac{ap}{2k^2}\log(p)$. 
Since $\sqrt{p\log p}\leq k$, 
and since $\log(1+ux)\geq u\log(1+x)$ for any $u\in (0,1)$ and $x>0$, 
we have 
 \[\frac{n\rho^2}{k}\leq \log\left[1+ x \vee \sqrt{x} \right] \leq  \log\left[ 1+x+ \sqrt{2x+x^2}\right]\ . \]
 Together with Lemma \ref{lem:chi2_upper} and the classical identity $\cosh[\log(1+u+ \sqrt{2u+u^2})]= 1+ u$ for all $u>0$, we arrive at
\beq
\frac{\sqrt{n}}{\sqrt{p}\log p} \chi^2(\mathbf{P}_{\mu},\P_0)\leq \frac{\sqrt{n}}{\sqrt{p}\log p }\exp\left[\frac{k^2}{p} x\right] \leq \frac{\sqrt{n}}{\sqrt{p}\log p }p^{a/2}= \sqrt{\frac{n}{p^{1-a}}}\frac{1}{\log p}\ .
\eeq
Coming back to the lower bound \eqref{eq:lower_proba_ac}, we conclude that, for any event $\cA$ satisfying $\P_0(\mathcal{A})\leq   \sqrt{n}/(\sqrt{p}\log p)$,  we have 
$$\mathbf{P}_{\mu}(A^c)\geq 1-\sqrt{\frac{n}{p}}\frac{1}{\log p}-\left(\sqrt{\frac{n}{p^{1-a}}}\frac{1}{\log p}\right)^{1/2}.$$ 
Plugging this result in \eqref{eq:lower_R} and using the fact that $p^{1-a}(\log p)^{2}\geq 16 n$ leads to the desired result.

\subsection{Proof of Theorem \ref{thrm:lower_minimax_unknown_variance}}

\subsubsection{General arguments}
Suppose that Condition \eqref{eq:assumption_minimax_lower_unknown_variance} is satisfied for some $\varsigma>0$. Define $r$ be the smallest integer such that $\varsigma \geq 1/(2r)$ so that we can assume henceforth that $n^{1+1/(2r)}/p\to 0$.

In this proof, we follow the same general approach as in the other minimax lower bounds, that is we define two mixture distributions $\mathbf{P}_0$ and $\mathbf{P}_1$
\[
\mathbf{P}_0:= \int \P_{\beta,\sigma_0,\bSigma}\mu_0(d\beta, d\bSigma)\ ,\quad 
\mathbf{P}_1:= \int \P_{\beta,\sigma_1,\bSigma}\mu_1(d\beta, d\bSigma)\ , 
\]
in such a way that $\mathbf{P}_0$ and $\mathbf{P}_1$ are almost indistinguishable and at the same time the function $\eta(\beta,\sigma)$ takes different values for parameters in the support of the prior distribution $\mu_0$ and parameters in the support of the prior distribution $\mu_1$. The main difference with previous proofs lies in the fact that $\mu_0$ and $\mu_1$ are now prior probabilities on both the regression coefficient $\beta$ and the covariance matrix $\bSigma$.

\medskip 

Let $\alpha_{0}=(\alpha_{i,0})$, $\gamma_0=(\gamma_{i,0})$, $i=1,\ldots r$  and $\alpha_1= (\alpha_{i,1})$, $\gamma_1=(\gamma_{i,1})$, $i=0,\ldots r$ be positive parameters whose exact values will be fixed later. We emphasize that the values of these parameters will only depend on $r$ and not on $n$ and $p$. 
Given a positive integer $q$ and $\alpha=(\alpha_1,\ldots,\alpha_q)$ whose coordinates $\alpha_j$ are positive, define the probability distribution $\pi_{\alpha}$ on vectors of $\mathbb{R}^{ q \times p}$ whose density is proportional to $(|\bI_p+ \sum_{i=1}^q \alpha_i x_ix_i^T|)^{-n/2}e^{-\sum_{i=1}^q p\|x_i\|_2^2/2}$ for $x_i\in \mathbb{R}^p$, $i=1,\ldots, q$.

\medskip 
The distribution $\mu_0$ is defined as follows. Let $(v_{i,0})$, $i=1,\ldots, r$ be independently sampled  according to the distribution $\pi_{\alpha_{0}}$. Then, conditionally  to $(v_{1,0},\ldots, v_{r,0})$, $\beta$ and $\bSigma$ are fixed to the following values 
\beq\label{eq:definition_sigma_0}
\beta=\sum_{i=1}^r \gamma_{i,0} v_{i,0}\, ;\quad  \quad \bSigma^{-1}= \bI_p + \sum_{i=1}^{r} \alpha_{i,0}v_{i,0}v_{i,0}^T.
\eeq
Similarly, under $\mu_1$, 
\beq
\beta=\sum_{i=0}^r \gamma_{i,1} v_{i,1}\, ;\quad  \quad \bSigma^{-1}= \bI_p + \sum_{i=0}^{r} \alpha_{i,1}v_{i,1}v_{i,1}^T\label{eq:definition_sigma_1}
\eeq
where the vectors  $(v_{i,1})$, $i=1,\ldots, r$ are independently sampled  according to the distribution $\pi_{\alpha_{1}}$. 
%
Finally, the noise variances are fixed to the following values.  
\begin{equation}\label{eq:choice_sigma}
\sigma^2_0=3/2\ , \quad \quad \sigma_1^2=1/2.
\end{equation}
\noindent

To prove that $\mathbf{P}_0$ and $\mathbf{P}_1$ are almost indistinguishable we will consider separately the marginal distribution of $\mathbf{X}$ and the conditional distribution of $Y$ given $\mathbf{X}$. We will see that the centered Gaussian distribution of $\bX$ under both $\mathbf{P}_0$ and $\mathbf{P}_1$ are indistinguishable from the standard normal distribution when $n=o(p)$, see Lemma \ref{lem:upper_total_variation_sigma} below.

Let us now choose the parameters  $\gamma_{i,j}$ and $\alpha_{i,j}$ in such a way that the conditional distribution of $Y$ given $\mathbf{X}$ under  $\mathbf{P}_0$ is  indistinguishable from that under $\mathbf{P}_1$ when
$n^{1+1/(2r)}=o(p)$.
We first consider a truncated moment problem. 

\begin{lem}\label{lem:truncated_moment}
There exist two discrete positive measures   $\rho_0= \sum_{i=1}^r \xi_{i,0} \delta_{\tau_{i,0}}$ on $\rho_1= \sum_{i=0}^r \xi_{i,1}\delta_{\tau_{i,1}}$ supported on $(0,1)$ such that 
\begin{enumerate}
 \item The atoms $\tau_{i,j}$ for $j=0,1$ and $i=1,\ldots, r$ lie in $[1/5,4/5]$, whereas the first atom $\tau_{0,1}$ of $\rho_1$ is allowed to be smaller.
 \item The total mass of $\rho_0$ equals $1/2$, whereas the total mass of $\rho_1$ is $3/2$.
 \item For all $q=1,\ldots, 2r-1$, the $q$-th moment of $\rho_0$ and $\rho_1$ coincide
 \[\int x^q d\rho_0 = \int x^q d\rho_1= \int_{1/4}^{3/4} x^{q}dx:= m_q\ .\]
\end{enumerate} 
\end{lem}

For $j=0,1$, we set the values  $\gamma_{i,j}= [\xi_{i,j}/\tau_{i,j}]^{1/2}$ and $\alpha_{i,j}=\tau_{i,j}^{-1}-1$. 

Let us give a hint why such a choice leads to what we need.
As a consequence of our parameter choices, the following identities are satisfied
\begin{eqnarray}
\sum_{i=1}^r \frac{\gamma_{i,0}^2}{1+\alpha_{i,0}} +\sigma_0^2&=& \sum_{i=0}^r \frac{\gamma_{i,1}^2} {1+\alpha_{i,1}}+\sigma_1^2= 2 \label{eq:relation_1}\\
 \sum_{i=1}^r \frac{\gamma_{i,0}^2}{(1+\alpha_{i,0})^q}&=& \sum_{i=0}^r \frac{\gamma_{i,1}^2}{(1+\alpha_{i,1})^q}=\frac{3^q-1}{q4^q}= m_{q-1},\quad \quad \forall q=2,\ldots, 2r  \label{eq:relation_2}
\end{eqnarray}
\medskip 
Had the random vectors $(v_{i,j})$ introduced in $\mu_j$ formed an orthonormal family, then we would have had $\beta^T\bSigma^q\beta= \sum_i \frac{\gamma_{i,j}^2}{(1+\alpha_{i,j})^q}$ for any positive integer $q$. 
We shall prove later that, under the distribution $\mu_j$, the vectors $v_{i,j}$ have a norm close to one and are almost  orthogonal with large probability.
Hence, identities  \eqref{eq:relation_2} imply that the moments $\beta^T\bSigma^q\beta$ concentrate around the same value under $\mu_0$ and $\mu_1$, this for all $q=2,\ldots, 2r$. This will lead to the fact that the conditional distribution of $Y$ given $\mathbf{X}$ under  $\mathbf{P}_0$ is  indistinguishable from that under $\mathbf{P}_1$ when
$n^{1+1/(2r)}=o(p)$ as proved in Lemma \ref{lem:upper_total_variation_sigma_Y} below.
In the same way, 
\eqref{eq:relation_1} will imply that $\beta^T\bSigma\beta +\sigma_j^2$ concentrate around $2$ under $\mu_j$ for $j=0,1$ so that $\eta$ will concentrate around different values under $\mathbf{P}_0$ and $\mathbf{P}_1$ since $\sigma_0^2 \neq \sigma_1^2$. This is stated in Lemma \ref{lem:concentration_eta} below.

\medskip
\noindent 
{\bf Remark}. As, with large probability,  the random vectors $v_{i,j}$ will be proved to be almost orthonormal, the spectrum of $\bSigma$ almost lies in $(1/5,1)$ with high probability under $\mu_0$. Under $\mu_1$ all the eigenvalues of $\bSigma$, except the smallest one, almost lie in $(1/5,1)$ with high probability, whereas the smallest eigenvalue of $\bSigma$, which is of order $1/(1+\alpha_{0,1})$, will be closer to zero. 
If we had wanted to restrict ourselves to covariance matrices with uniformly bounded eigenvalues (in say $[M,1/M]$) as suggested in the discussion below Theorem \ref{thrm:lower_minimax_unknown_variance}, we would have defined the parameters thanks to discrete measures $\rho_0$ and $\rho_1$ with support in $[1/M,1]$. However, to constrain the $q$-th moment of $\rho_0$ and $\rho_1$ to coincide for $q=1,\ldots, 2r-1$, the difference in total mass between $\rho_0$ and $\rho_1$ would now depend on $r$. The remainder of the proof would be unchanged except that the quantities $\eta_0$ and $\eta_1$ in \eqref{eq:definition_eta} would depend on $r$ and the ultimate conclusion would be that $\underline{\lim}\overline{R}^*[p,M] \geq C(r)$.

\medskip

 Let us now define the quantities
\begin{equation}\label{eq:definition_eta}
 \eta_0:=1 - \frac{\sigma_0^2}{\sum_{i=1}^r \frac{\gamma_{i,0}^2}{1+\alpha_{i,0}} +\sigma_0^2}= 1/4\ ,\quad \quad \eta_1:=1 - \frac{\sigma_0^2}{\sum_{i=0}^r \frac{\gamma_{i,1}^2}{1+\alpha_{i,1}} +\sigma_1^2}= 3/4\ .
\end{equation}

The next lemma states that, for $j=0,1$, $\eta(\beta, \sigma)$ is close to $\eta_j$ under $\mu_j$.

\begin{lem}\label{lem:concentration_eta}
There exists three positive constants $C_1(r)$, $C_2(r)$ and $C_3(r)$ such that the following holds. If $p$ is larger than $n C_1(r)$, then
\beq\label{eq:concentration_eta1}
 \mu_j\left[|\eta(\beta, \sigma) - \eta_j|\geq C_2(r)\big(p^{-1/4}+ (\frac{n}{p})^{1/2}\big)\right]\leq e^{-C_3(r)p^{1/2}}\ , 
 \eeq
for $j=0,1$. Also, the spectrum of $\bSigma$ is bounded away from zero with large probability, that is for $j=0,1$,
\beq\label{eq:concentration_eta_lambda_Sigma}
 \mu_j\left[\lambda_{\min}(\bSigma)\geq \frac{1}{2}\min_{i}\frac{1}{1+\alpha_{i,j}} \right]\leq e^{-C_3(r)p^{1/2}}\ .
 \eeq 

\end{lem}

Define $M(r):= [2\max_{i,j}(1+\alpha_{i,j})^{-1}]$. By definition of $\mu_0$ and $\mu_1$, the largest eigenvalue of $\bSigma$ is always equal to one. By Lemma \ref{lem:concentration_eta}, with $\mu_0$ and $\mu_1$ probability going to one, the spectrum of $\bSigma$ lies in $[1/M(r),M(r)]$.

Let us now bound the minimax risk $\overline{R}^*[p,M(r)]$.
Contrary to the prior distributions chosen in the proof of Proposition \ref{prp:lower_adaptation}, the proportion of explained variation $\eta(\beta,\sigma)$ is not constant either on $\mu_0$ or on $\mu_1$, so that we cannot directly relate the minimax estimation rate to the total variation distance as done before. Nevertheless, these proportions of explained variation concentrate around  $\eta_0$ and $\eta_1$  so that it will be possible to work around this difficulty. This slight refinement of Le Cam's method has already been applied for other functional estimation problems (see e.g.\ \cite{MR2816346}).  Also to circumvent the issue that some eigenvalues of $\bSigma$ are smaller than $M[r]$ with positive (but very small) probability, we consider a thresholded version of the risk  $\mathbf{E}^*_1\left[.\right]:= \mathbf{E}_1\left[.\mathbf{1}_{\lambda_{\min}(\bSigma)\geq M^{-1}(r)}\right]$ and $\mathbf{E}^*_0\left[.\right]:= \mathbf{E}_0\left[.\mathbf{1}_{\lambda_{\min}(\bSigma)\geq M^{-1}(r)}\right]$.

 Without loss of generality, we may assume that all the estimators $\widehat{\eta}$ below only take values in $[0,1]$.
\beqn 
\overline{R}^*[p,M(r)]&\geq& \inf_{\widehat{\eta}}\mathbf{E}^*_0\left[\left\{\widehat{\eta}-\eta(\beta,\sigma)\right\}^2\right] \bigvee  \mathbf{E}^*_1\left[\left\{\widehat{\eta}-\eta(\beta,\sigma)\right\}^2\right] \\
&\geq &\inf_{\widehat{\eta}}\mathbf{E}_0\left[\left\{\widehat{\eta}-\eta(\beta,\sigma)\right\}^2\right] \bigvee  \mathbf{E}_1\left[\left\{\widehat{\eta}-\eta(\beta,\sigma)\right\}^2\right] -\bigvee_{i=0,1}\mu_i\left[ \lambda_{\min}(\bSigma) \leq  M^{-1}(r) \right] \\ 
&\geq & \inf_{\widehat{\eta}}\frac{1}{2}\bigvee_{i=1,2}\mathbf{E}_i\left[\left\{\widehat{\eta}-\eta_i\right\}^2\right]  - \bigvee_{i=1,2}\mathbf{E}_i\left[\left\{\eta(\beta,\sigma)-\eta_i\right\}^2\right] -\bigvee_{i=0,1}\mu_i\left[ \lambda_{\min}(\bSigma) \leq  M^{-1}(r) \right] \ , 
\eeqn 
where we used $(x-y)^2 \geq (x-z)^2/2 - (y-z)^2$. From  \eqref{eq:concentration_eta1} and the fact that $\eta(\beta,\sigma)$ belongs to $[0,1]$, we derive that 
\[\bigvee_{i=1,2}\mathbf{E}_i[\{\eta(\beta,\sigma)-\eta_i\}^2] \leq C(r)(p^{-1/2}+n/p)\ ,\] 
when $p$ is large enough. Besides, the probabilities $\mu_i[ \lambda_{\min}(\bSigma) \leq  M^{-1}(r) ]$ are smaller than $e^{-Cp^{1/2}}$ by \eqref{eq:concentration_eta_lambda_Sigma}.
Then, we control the maximum $\bigvee_{i=1,2}\mathbf{E}_i\left[\left\{(\widehat{\eta}-\eta_i\right\}^2\right]$ using the total variation distance between $\mathbf{P}_0$ and $\mathbf{P}_1$ as we did in the proof of Proposition \ref{prp:lower_minimax}. More precisely,
\beqn 
\overline{R}^*[p,M(r)]+  C(r) [p^{-1/2}+ (n/p)]&\geq & \frac{(\eta_1-\eta_0)^2}{8}\inf_{\widehat{\eta}}\mathbf{P}_0\big(\widehat{\eta}\geq \tfrac{\eta_1+\eta_0}{2}\big)\bigvee \mathbf{P}_1\big(\widehat{\eta}\leq \tfrac{\eta_1+\eta_0}{2}\big)  \\
 &\geq &\frac{(\eta_1-\eta_0)^2}{16} \inf_{\cA} \mathbf{P}_0(\cA)+ \mathbf{P}_1(\cA^c) \\
 &\geq & \frac{(\eta_1-\eta_0)^2}{16} \left[1 - \|\mathbf{P}_1-\mathbf{P}_0\|_{TV}\right] \ , 
\eeqn 
so that we only have to focus on $\|\mathbf{P}_1-\mathbf{P}_0\|_{TV}$. Let us decompose the total variation distance between $\mathbf{P}_{0}$ and $\mathbf{P}_{1}$ in a way enabling to consider separately the marginal distribution of $\bX$ and the conditional distributions of $Y$ given $\bX$.
Since the total variation distance is, up to a multiplicative constant, the $l_1$ distance between the density functions, we obtain
\begin{eqnarray}
2\|\mathbf{P}_1 - \mathbf{P}_0\|_{TV} &= &\int |f_0(y,\bx) - f_1(y,\bx)|dyd\bx   \nonumber\\
& = & \int |f_0(y|\bx)f_0(\bx) - f_1(y|\bx)f_1(\bx)|dyd\bx  \nonumber \\
&\leq & \int  f_1(y|\bx) |f_0(\bx) -f_1(\bx)|dyd\bx + \int f_0(x)|f_0(y|\bx) - f_1(y|\bx)|dyd\bx  \nonumber \\
&\leq & \int   |f_0(\bx) -f_1(\bx)|d\bx+\int f_0(x)|f_0(y|\bx) - f_1(y|\bx)|dyd\bx \nonumber \\
&\leq & 2 \|\mathbf{P}^{\bX}_{0} - \mathbf{P}^{\bX}_{1}\|_{TV}+ 2 \mathbf{E}^{\bX}_0\left[ \|\mathbf{P}^{Y|\bX}_{0} - \mathbf{P}^{Y|\bX}_{1}\|_{TV}\right]\ , \label{eq:upper_tv}
\end{eqnarray}
where, for $i=0,1$,  $\mathbf{P}^{\bX}_i$ (resp. $f_i$) denotes the marginal probability distribution (resp. density) of $\bX$ under $\mathbf{P}_i$, $\mathbf{P}^{Y|\bX}_{i}$ (resp. $f_{i}(\cdot|\bx)$) is the conditional distribution (resp. density) of $Y$ given $\bX$ and $\mathbf{E}_0^{\bX}$ stands for the expectation with respect to $\mathbf{P}^{\bX}_0$. The main difficulty in the proof lies in controlling  these two total deviation distances $\|\mathbf{P}^{\bX}_{0} - \mathbf{P}^{\bX}_{1}\|_{TV}$ and $\mathbf{E}^{\bX}_0[ \|\mathbf{P}^{Y|\bX}_{0} - \mathbf{P}^{Y|\bX}_{1}\|_{TV}]$.

\medskip 

The marginal distribution of $\bX$ under $\mathbf{P}_{0}$ and $\mathbf{P}_{1}$ is that of a $n$ sample of $p$-dimensional normal distribution whose precision matrix is a rank $r$ perturbation of the identity matrix and whose $r$ principal directions are sampled nearly uniformly. In a high-dimensional setting, such perturbations are indistinguishable from the standard normal distribution as shown in the next lemma.

\begin{lem}\label{lem:upper_total_variation_sigma}
There exist two positive constants $C(r)$ and $C'(r)$ only depending on $r$ such that the following holds. 
If $p\geq C(r) n$, then 
\[
 \|\mathbf{P}^{\bX}_{0} - \mathbf{P}^{\bX}_{1}\|_{TV}\leq C'(r)\sqrt{\frac{n}{p}}\ .
\]
\end{lem}

The intricate construction of $\mu_0$ and $\mu_1$ (and especially the choices of the parameters $\alpha_{i,j}$ and $\gamma_{i,j}$) has been made to 
force the conditional $\mathbf{P}^{Y|\bX}_{0}$ and $\mathbf{P}^{Y|\bX}_{1}$ to be close to each other. Informally, the fact that the quantities  $\beta^T\bSigma^q\beta$ almost coincide under $\mu_0$ and $\mu_1$, this for all $q=2,\ldots, 2r$, will translate into the  total distance $\|\mathbf{P}^{Y|\bX}_{0} - \mathbf{P}^{Y|\bX}_{1}\|_{TV}$ as illustrated by the next lemma.

\begin{lem}\label{lem:upper_total_variation_sigma_Y}
There exist two positive constants $C(r)$ and $C'(r)$ only depending on $r$ such that the following holds. 
If $p\geq C(r) n$, then 
\[
 \mathbf{E}^{\bX}_0\left[ \|\mathbf{P}^{Y|\bX}_{0} - \mathbf{P}^{Y|\bX}_{1}\|_{TV}\right]\leq C'(r)\left(\frac{n^{1+ 1/(2r)}}{p} \right)^{r}\ .
\]
\end{lem}

Under assumption \eqref{eq:assumption_minimax_lower_unknown_variance}, the distance $\|\mathbf{P}_1-\mathbf{P}_0\|_{TV}$ goes to $0$, and the minimax risk $\overline{R}^*[p,M(r)]$ is therefore bounded away from zero:
\[\underline{\lim}\overline{R}^*[p,M(r)] \geq \frac{(\eta_1-\eta_0)^2}{32}\ .\]

\subsubsection{Proof of the truncated moment problem (Lemma \ref{lem:truncated_moment})}

Define $\overline{\rho}$ the uniform measure over the interval $[1/4,3/4]$. First, we want to construct $\rho_0$ an $r$-atomic measure whose support is in $[1/4,3/4]$ and whose moments up to order $2r-1$ coincide with those of $\overline{\rho}$. This truncated moment problem has received a lot of attention in the literature. For instance, Theorem 4.1 (equivalence between (i) and (ii)) in~\cite{curto_truncated} ensures the existence of $\rho_0$.  Define the Hankel matrix $\bA$ of order $r-1$ and the matrix $\bB$ by 
\[\bA_{i,j}:= m_{i+j},\quad \, \bB_{i,j}:= m_{i+j+1} ,\quad  \quad \forall 0\leq i,j\leq r-1\ .\]
(Here, $m_0= \int_{1/4}^{3/4}dx=1/2$).
The same theorem 4.1 in~\cite{curto_truncated} then ensures that the symmetric Hankel matrix  $\bA$ is positive semidefinite, $\bA\geq 0$, and $3/4 \bA \geq \bB\geq 1/4\bA$ where $\bB\geq 1/4\bA$ implies that $\bB-1/4\bA$ is positive semidefinite. Since the representation of the truncated moment problem $(m_0,\ldots, m_{2r-1})$ is not unique ($\overline{\rho}$ and $\rho_0$ are admissible) Theorem 3.8 in~\cite{curto_truncated} ensures that $\bA$ is non-singular. Hence, up to modifying the constants, we can obtain strict inequalities in the bounds
\beq\label{eq:condition}
\bA>0 \ ,\quad \quad  \frac{4}{5} \bA > \bB> \frac{1}{5}\bA\ .
\eeq
\medskip 
Given $\varepsilon>0$, define the modified  matrices $\bA^{\varepsilon}$ and $\bB^{\varepsilon}$ by
\[\bA^{\varepsilon}_{i,j}:=\left\{\begin{array}{ccc}
                                     m_{0} &\text{ if }& i=j=0\\
                                     m_{i+j-1}-\varepsilon^{i+j} &\text{else}
                                    \end{array}\right. \ , \quad \quad \bA^{\varepsilon}_{i,j}=m_{i+j}-\varepsilon^{i+j}\ .
\]
Since the set of positive matrices is open, there exists some $\varepsilon_0>0$ such that the Hankel matrix $\bA^{\varepsilon_0}$ is positive and $\tfrac{4}{5} \bA^{\varepsilon_0} > \bB^{\varepsilon_0}> \tfrac{1}{5}\bA^{\varepsilon_0}$.
As a consequence of Theorem 4.1 in~\cite{curto_truncated}, there exists an $r$-atomic measure $\rho^{\varepsilon_0}$ with support in $[1/5,4/5]$ whose $q$-th moment is $m_q-\varepsilon_0^{q}$ for $q=1,\ldots, 2r-1$ and $m_0$ for $q=0$.

Finally, the measure $\rho_1:=  \delta_{\varepsilon_0}+ \rho^{\varepsilon_0}$ satisfies all the desired moment conditions, that is $\int d\rho_1=m_0+1$ and $\int x^q d\rho_1=m_q$ for $q=1,\ldots, 2r-1$.

\subsubsection{Additional lemma}

\begin{lem}\label{lem:concentration_v}
There exists three positive constants $C_4(r)$, $C_5(r)$ and $C_6(r)$ such that the following holds. Assuming that $p\geq n C_4(r)$, we have, for both $j=0$ and $j=1$, 
\beq\label{eq:concentration_v}
\pi_{\alpha_j}\Big[\max_i|\|v_i\|_2^2-1| \geq C_5(r) \big(\frac{n}{p}\big)^{1/2} + p^{-1/4}\Big]\leq e^{-C_6 (r)p^{1/2}}\ .
\eeq
\end{lem}

\begin{proof}
We only prove the lemma for $j=0$, since for $j=1$ the proof is similar.
To ease the notation, we simply write $v$ and $\alpha$ for $(v_{i,0})_{i}$ and $(\alpha_{i,0})_{i}$ and $\mu$ for $\mu_0$. Recall that the density of $v=(v_1,\ldots, v_r)$ is proportional to $e^{-p\sum_i\|v_i\|_2^2/2}|\bI_p+\sum_i\alpha_i v_iv_i^T|^{-n/2}$. Denote $\omega$ the uniform probability measure on the $p$-dimensional sphere. Let us first change the coordinate system. The density $g$ of $t=(\|v_i\|_2^2)$ and $w=(v_i/\|v_i\|_2 )$ satisfies $g(t,w)= \phi(t,w)[\smallint \phi(x,w')\prod dx_i d\omega(dw'_i)]^{-1}$, where
\[\phi(t,w):=  (\Pi_i t_i^{-1})e^{-\frac{p}{2}\sum_{i=1}^r(t_i-1-\log(t_i))} |\bI_p+ \sum_{i}\alpha_it_i w_iw_i^T|^{-n/2}\ .\]

In order to control the density $g$, we first provide a lower bound on the normalizing constant. 
For $t\in (1-(rp)^{-1/2},1)^r$, $|\bI_p+ \sum_{i}\alpha_it_i w_iw_i^T|\leq \|\bI_p+ \sum_{i}\alpha_it_i w_iw_i^T\|^r_{op}\leq (1+\|\alpha\|_{\infty} r )^r$.
As a consequence, 
\begin{eqnarray*}
\int \phi(x,w')\Pi_i dx_id\omega(w'_i)&\geq &\left[\int_{0}^{(rp)^{-1/2}}\frac{1}{1-u}^{-\frac{p}{2}(-u-\log (1-u))}du\right]^{r} (1+\|\alpha\|_{\infty}r)^{-nr/2}\\
&\geq&(rp)^{-r/2}e^{-r^{2}/4} (1+\|\alpha\|_{\infty}r)^{-nr/2}\ .
\end{eqnarray*}

Since the determinant $|\bI_p+ \sum_{i}\alpha_it_i w_iw_i^T|$ is always larger than one, we obtain for some constant $C(r)$ depending only on $r$ and some universal constant $C$
\begin{eqnarray}\nonumber
g(t,w) &\leq&\  (rp)^{r/2} e^{-r^{2}/4}(1+\|\alpha\|_{\infty}r)^{nr/2} (\Pi_i t_i^{-1})e^{-\frac{p}{2}\sum_{i=1}^r(t_i-1-\log(t_i))} \\
&\leq & C(r) p^{r/2} (1+\|\alpha\|_{\infty}r)^{nr/2}\exp\Big[- C p \sum_{i}\Big\{(t_i-1)^2\mathbf{1}_{t_i\in (1/2,3/2)}+ \mathbf{1}_{t_i\leq 1/2} +t\mathbf{1}_{t_i\geq 3/2} \Big\} \Big]. \nonumber\\ \label{eq:lower_density_g}
\end{eqnarray}
 If $\|t-1\|_{\infty}$ belongs to  $((\frac{nr}{2Cp}\log(1+\|\alpha\|_{\infty}r))^{1/2}+ p^{-1/4},1/2)$, then  for some constant $C'(r)$ depending only on $r$ 
\[g(t,w)\leq C'(r) p^{r/2} \exp(-C p^{1/2})\ .  \]
 If $\|t-1\|_{\infty}\geq 1/2$, then $g(t,w)\leq C'(r) p^{r/2} \exp(-C' p \|t-1\|_{\infty})$ for some universal constant $C'$. Integrating these bounds with respect to $w$ and $t$, we conclude that
 for some constant $C_{5}(r)$ depending only on $r$ and some universal constant $C''$
\[\pi\Big[\max_i|\|v_i\|_2^2-1| \geq \big(C_{5}(r)\frac{n}{p}\big)^{1/2} + p^{-1/4}\Big]\leq C'(r) p^{r/2}e^{-C''p^{1/2}}\ ,\]
which is smaller than $e^{-C_{6} (r) p^{1/2}}$ for some constant $C_{6}(r)$ for $p$ is large compared to $r$. 

\end{proof}

\subsubsection{Control of  $\|\mathbf{P}^{\bX}_1-\mathbf{P}^{\bX}_0\|_{TV}$ (Proof of Lemma \ref{lem:upper_total_variation_sigma})}

 Define the probability distribution $\overline{\mathbf{P}}^{\bX}$,  such that, under $\overline{\mathbf{P}}^{\bX}$, the entries of $\bX$ follow independent standard normal distributions. By the triangular inequality, we have
\beq\label{eq:decomposition_tv_PX}
\|\mathbf{P}^{\bX}_{0} - \mathbf{P}^{\bX}_{1}\|_{TV}\leq \|\overline{\mathbf{P}}^{\bX} - \mathbf{P}^{\bX}_{0}\|_{TV}+\|\overline{\mathbf{P}}^{\bX} - \mathbf{P}^{\bX}_{1}\|_{TV}\ ,
\eeq
We will prove that $\|\overline{\mathbf{P}}^{\bX} - \mathbf{P}^{\bX}_{0}\|_{TV}$ is small, the distance $\|\overline{\mathbf{P}}^{\bX} - \mathbf{P}^{\bX}_{1}\|_{TV}$ being handled similarly. In order to simplify the notation in the remainder of this proof, we drop the subscript $0$ in the vector $\alpha_{0}$ and $v_0$.
%
Let $\omega$ denote the Haar measure on dimension $p$ orthogonal matrices. We work out 
the marginal density $f_0(\bX)$  as follows
\beqn
f_0(\bX)&= &\int (2\pi)^{-(np)/2}\big|\bI_p + \sum_{i=1}^r \alpha_{i}v_{i}v_{i}^T\big|^{n/2} \exp\Big[-tr(\bX\bX^T)/2-\sum_{i=1}^r\alpha_{i} \|\bX v_{i}\|_2^2/2\Big] \pi_{\alpha}(dv) \\
&= & \int_{\mathbb{R}^+}    h_0(v,\bX)  \pi_{\alpha}(dv)\ ,\\ &&\quad \text{where} \quad h_0(v,\bX):=\big|\bI_p + \sum_{i=1}^r \alpha_{i}v_{i}v_{i}^T\big|^{n/2} e^{-tr(\bX\bX^T)/2}\int e^{- \sum_{i=1}^r\alpha_{i}  \|\bX \bO v_i \|_2^2/2} \omega(d\bO) \ , 
\eeqn 
since the values of $\big|\bI_p + \sum_{i=1}^r \alpha_{i}v_{i}v_{i}^T\big|$ and $\sum_{i=1}^{r}\|v_{i}\|^{2}$ are rotation invariant.

For fixed $v=(v_1,\ldots,v_r) $, $h_0(v,\bX)$ stands for the density of $\bX$ when the corresponding precision matrix of the rows of  $\bX$ is a rank $r$ perturbation of the identity matrix whose  directions  are  sampled uniformly on the unit sphere. We shall prove below that it is impossible to distinguish this distribution from $\overline{\mathbf{P}}^{\bX}$ (i.e.\ no perturbation) when $p$ is large compared to $n$. 
Denote $\overline{f}(\bX)$ the density of $\bX$ under $\overline{\mathbf{P}}^{\bX}$.
In the following equations, $ \|f_0(\bX) - \overline{f}(\bX)\|_{1}$ denotes the $l_1$ distance (in $\mathbb{R}^{n\times p}$)  between the densities. Using Fubini's Theorem, we obtain
\begin{eqnarray} 
2 \| \mathbf{P}^{\bX}_{0}- \overline{\mathbf{P}}^{\bX}\|_{TV} &= & \|f_0(\bX) - \overline{f}(\bX)\|_{1}  \nonumber \\
&\leq&  \int \  \|h_0(\bX,v) - \overline{f}(\bX)\|_1  \pi_{\alpha}(dv)\ , \label{eq:decompositiontv_PX}
 \end{eqnarray}
so that we will bound the $l_1$ distance $\|h_0(\bX,v) - f(\bX)\|_1 $ for all $v=(v_1,\ldots, v_r)$.  Denote $L_v(\bX)$ the likelihood ratio $h_0(v,\bX)/\overline{f}(\bX)$. 
\begin{eqnarray}
  \|h(\bX,v) - \overline{f}(\bX)\|_1&=& \overline{\mathbf{E}}^{\bX}\left[|L_v(\bX)-1|\right] \nonumber \\
 &\leq &\sqrt{\overline{\mathbf{E}}^{\bX}\left[(L_v(\bX)-1)^2\right]}= \sqrt{\overline{\mathbf{E}}^{\bX}[L^2_v(\bX)] -1} \ ,\label{eq:upper_TV_chi2}
\end{eqnarray}
so that we have to compute the second moment of the likelihood $L_v(\bX)$. As the proof of the following lemma is a bit tedious, it is postponed to the end of the subsection.

\begin{lem}\label{lem:control_second_moment_marginal_X}
There exist three positive constants $C_7(r)$, $C_8(r)$, $C_9(r)$, only depending on $r$ such that the following holds. Assuming $p\geq n C_7(r)$, we have 
\[
\overline{\mathbf{E}}^{\bX}\left[L^2_v(\bX)\right]\leq  1+  C_8(r) \frac{n}{p}+ 4r \exp(-C_9(r)p) \ ,
\]
simultaneously for all $v=(v_1,\ldots, v_r)$ satisfying $\max_i\|v_i\|_2 \leq 2$.
\end{lem}

This lemma, together with \eqref{eq:upper_TV_chi2}, gives us a uniform  bound of $\|h(\bX,v) - \overline{f}(\bX)\|_1$ over all $v$ satisfying the above condition
\[
\|h(\bX,v) - \overline{f}(\bX)\|_1^2 \leq  C_8(r) \frac{n }{p} + 4r \exp(-C_9(r)p)\leq  C_{10}(r)\frac{n}{p}\ .
\]
 Denote $\cV$ the collection of  $v$ such that $v=(v_1,\ldots, v_r)$ satisfying $\max_i\|v_i\|_2 \leq 2$. Coming back to the decomposition \eqref{eq:decompositiontv_PX}, we conclude that 
\begin{eqnarray}
 \|\mathbf{P}^{\bX}_{0} - \overline{\mathbf{P}}^{\bX}\|_{TV} &\leq & \int_{\cV} \|h(\bX,v) - \overline{f}(\bX)\|_1 \pi_{\alpha}(dv) + \int_{\cV^c}   \|h(\bX,v) - \overline{f}(\bX)\|_1 \pi_{\alpha}(dv) \nonumber \\
 &\leq &   C(r) \sqrt{\frac{n}{p}} + \mu_0\left[\max_{i=1}^r\|v_{i}\|_2 \geq 2 \max_{i}|\alpha_i|\right] 
 \leq  C'(r)\sqrt{\frac{n}{p}} \nonumber \ , 
\end{eqnarray}
 when $p$ is large compared to $n$ using Lemma \ref{lem:concentration_v}. Handling analogously the difference $\|\mathbf{P}^{\bX}_{1} - \overline{\mathbf{P}}^{\bX}\|_{TV}$, we conclude that 
\beq
\|\mathbf{P}^{\bX}_{1} - \overline{\mathbf{P}}^{\bX}\|_{TV}\leq  C'(r)\sqrt{\frac{n}{p}} \ . \label{eq:upper_tv_X}
\eeq

\begin{proof}[Proof of Lemma \ref{lem:control_second_moment_marginal_X}]

Relying on Fubini identity and the fact that $\bX$ follows a normal distribution, we have 
\begin{eqnarray}
\overline{\mathbf{E}}^{\bX}\left[L^2_v(\bX)\right]&=& \overline{\mathbf{E}}^{\bX}\left[\big|\bI_p + \sum_{i=1}^r \alpha_{i}v_{i}v_{i}^T\big|^{n} \int e^{-\sum_{i=1}^r\alpha_{i} [\|\bX \bO v_i \|_2^2+ \|\bX \bO' v_i \|_2^2]/2}  \omega(d\bO)\omega(d\bO')\right]\nonumber\\
& = &\int \left[\frac{\big|\bI_p + \sum_{i=1}^r \alpha_{i}v_{i}v_{i}^T\big|^{2}}{|\bI_p + \sum_{i=1}^r \alpha_{i}(\bO v_{i})(\bO v_{i})^T + \alpha_{i}(\bO' v_{i})(\bO' v_{i})^T\big|^{2}| }\right]^{n/2}\omega(d\bO)\omega(d\bO') \nonumber \\ \nonumber
& = &\int \left[\frac{\big|\bI_p + \sum_{i=1}^r \alpha_{i}v_{i}v_{i}^T\big|^{2}}{|\bI_p + \sum_{i=1}^r \alpha_{i} v_{i}v_{i}^T + \alpha_{i}(\bO v_{i})(\bO v_{i})^T\big|^{2}| }\right]^{n/2}\omega(d\bO)\ .
\end{eqnarray}
Diagonalizing the matrix $\bI_p + \sum_{i=1}^r \alpha_{i}v_{i}v_{i}^T$, we find $\tilde{\alpha}_1,\ldots, \tilde{\alpha_r}>0$ and an orthonormal family $w_1,\ldots, w_r$ such that $\sum_{i=1}^r \alpha_{i}v_{i}v_{i}^T= \sum_{i=1}^r \tilde{\alpha}_i w_iw_i^T$. Note that $\|\tilde{\alpha}\|_{\infty}\leq \|\alpha\|_{\infty}\sum_{i=1}^r\|v_i\|_2^2\leq 2r \|\alpha\|_{\infty}$.
We arrive at the following representation
\beq\label{eq:decomposition_ELt_X}
 \overline{\mathbf{E}}^{\bX}\left[L^2_v(\bX)\right] = \int \left[\frac{\big|\bI_p + \sum_{i=1}^r \tilde{\alpha}_{i}w_{i}w_{i}^T\big|^{2}}{|\bI_p + \sum_{i=1}^r \tilde{\alpha}_{i} w_{i}w_{i}^T + \tilde{\alpha}_{i}(\bO w_{i})(\bO w_{i})^T\big|^{2}| }\right]^{n/2}\omega(d\bO)\ .
\eeq
We see that $\overline{\mathbf{E}}^{\bX}[L^2_t(\bX)]$ expresses as the $n/2$-moment of a ratio of determinants. 
In order to ease the notation, we extend the vector $\tilde{\alpha}$ in a $2r$-dimensional vector by concatenating it with itself. Define the size diagonal matrix $\bD$ by $\bD_{i,i}= \tilde{\alpha_{i}}$ for $1\leq i\leq 2r$. Extend the orthonormal family $(w_1,\ldots, w_r)$ into $(w_1,\ldots, w_{2r})$ by a Gram-Schmidt orthonormalization of $(w_1,\ldots, w_r,\bO w_1,\ldots, \bO w_r)$.
Since the determinant of $\bI_p+ \sum_{i=1}^r  \tilde{\alpha}_iw_iw_i^T+ \sum_{i=1}^r  \tilde{\alpha}_i(\bO w_i)(\bO w_i)^T $ is only determined by its restriction of the basis $(w_1,\ldots, w_{2r})$, we introduce the matrix $\bA$ of this linear application into the space spanned by $(w_1,\ldots, w_{2r})$. We arrive at
\beq\label{eq:decomposition_ELt_X_2}
\overline{\mathbf{E}}^{\bX}\left[L^2_v(\bX)\right]= \int \left(\frac{|\bI_{2r}+\bD|}{|\bA|}\right)^{n/2}\omega(d\bO)\ .
\eeq

In order to prove that this quantity is close to one, we shall show that the matrix close $\bA$ is close (in entry-wise supremum norm) to $\bI_{2r}+\bD$. The difference matrix $\bV:= \bA - \bI_{2r} - \bD$ writes as
\[
\bV_{l,m}= \sum_{i=1}^{r}  \tilde{\alpha}_{i} \langle \bO w_i, w_l\rangle  \langle \bO w_i, w_m\rangle - \tilde{\alpha}_{l}\mathbf{1}_{l=m}\mathbf{1}_{l>r}
\]
Given $i=1,\ldots, r$, define the space $S_i= \mathrm{Vect}\{w_1,\ldots, w_r, \bO w_1,\ldots, \bO w_{i-1}\}$. 
By definition of $w_l$, observe that $\langle w_i, \nu_l\rangle=0$  for all $i<l$. As a consequence, for any $l<m$, 
 \[|\bV_{l,m}|\leq \sum_{i=m}^r \tilde{\alpha}_{i}  |\langle \bO w_i, w_l\rangle|\leq 
 \sum_{i=1}^{r}  \tilde{\alpha}_{i} \|\Pi_{S_i} \bO w_i\|_2\ , \] where $\Pi_S$ denote the orthogonal projection onto the vector space $S$. The diagonal terms of $\bV$ satisfy
\[
|\bV_{l,l}|\leq    \sum_{i=1}^{r}  \tilde{\alpha}_{i}  \|\Pi_{S_i}\bO w_i\|_2^2  \ .
\]
For $i=1,\ldots, r$, denote $W_i=  \|\Pi_{S_i}\bO w_i\|_2$. 
 Define the matrix  $\bV_2:=(\bI_{2r}+\bD)^{-1/2}\bV(\bI_{2r}+\bD)^{-1/2} $. As a consequence of the  previous inequalities, we obtain
\[
 |tr(\bV_2)|\leq 2r^2\|\tilde{\alpha}\|_{\infty}\max_{i}W_i^2\ , \quad\quad \| \bV_2\|_{\infty}\leq r\|\tilde{\alpha}\|_{\infty} \max_{i} W_i\ .
\]
Assume that $2r\|\bV_2\|_{\infty}\leq 1/2$ so that $\|\bV_{2}\|_{op}\leq 1/2$. Also denote $\lambda_{i}(\bV_2)$ the ordered eigenvalues of $\bV_2$.
\beqn 
\log\left[\frac{|\bA|}{|\bI_{2r}+\bD|}\right]& =& \log\left[|\bI_{2r}+ \bV_2|\right]= \sum_{i=1}^{2r}\log(1+\lambda_{i}(\bV_2))\\
&\geq & tr(\bV_2) - \sum_{i}\lambda^2_{i}(\bV_2) \quad \quad \text{(since $\log(1+x)\geq x-x^2$ for $x\in [-1/2,1/2]$  )}\\
&\geq & - 6r^4(\|\tilde{\alpha}\|^2_{\infty}\vee 1)\max_{i}W_i^2\ .
\eeqn 
Define the event $\cA:=\{4r^2\|\tilde{\alpha}\|_{\infty}\max_{i}W_i<1\}$, so that, under $\cA$, we have $2r\|\bV_2\|_{\infty}\leq 1/2$. Under $\cA$, we bound $\log[|\bA|]$ as above, whereas, under $|\cA|^c$, we simply use that $|\bA|>1$. We also write $\mathbb{E}^\omega$ for the expectation with respect to the Haar measure $\omega$.
\begin{eqnarray}
\overline{\mathbf{E}}^{\bX}\left[L^2_v(\bX)\right]&\leq& \mathbb{E}^\omega\left[\exp\left(4nr^4(\|\tilde{\alpha}\|^2_{\infty}\vee 1) \max_{i} W^2_i\right) \right]+ \omega[\cA^c]|\bI_{2r}+\bD|^{n}  \nonumber \\ 
&\leq& \mathbb{E}^\omega\left[\exp\left( \frac{n C(r)}{p}  (p \max_{i}W_i^2)\right)\right]+  \omega[\cA^c] \exp\left(n C'(r)\right)\ ,
\label{eq:upper_Lt_X}
\end{eqnarray}
where $C(r)$ and $C'(r)$ only depend on $r$. We used in the last line that $\|\tilde{\alpha}\|_{\infty}\leq  2r \|\alpha\|_{\infty}$ and that the choice of $\alpha$ only depends on $r$.

In order to work out this quantity, we need to control the deviations of $p\max_{i}W^2_i$. Recall that $(w_1,\ldots, w_r)$ form an orthonormal family. Hence,  conditionally to $(\bO w_1,\ldots, \bO w_{i-1})$,  $\bO w_i$ follows a uniform distribution on the unit sphere intersected with the orthogonal space of $\mathrm{Vect}(\bO w_1,\ldots, \bO w_{i-1})$.
As a consequence, $W^2_i$ follows the same distribution as $\sum_{i=1}^{r}Z_i^2/\|Z\|_2^2$ where $Z=(Z_1,\ldots, Z_{p-i+1})\sim \cN(0,\bI_{p-i+1})$ (since the Gaussian distribution is isotropic). Noting that $W_i^2$ is always smaller than one $1$, we consider for any $t\in (0,p)$
\beqn 
\omega\left[pW_i^2\geq t+2r\right]&= &\P\left[\frac{p\sum_{i=1}^{r}Z_i^2}{\|Z\|_2^2}\geq t+2r\right]\leq \P\left[\|Z\|_2^2\leq p/2\right]+ \P\left[\sum_{i=1}^{r}(Z_i^2-1) \geq \frac{t}{2}\right]\\
&\leq & e^{-Cp}+ e^{-\tfrac{t}{8}\wedge \tfrac{t^2}{64r} }\leq 2e^{-C(r)t} \ , 
\eeqn 
where we used Lemma \ref{lem:chi_2} and $r\leq p/8$ in the last line. Taking an union bound and integrating this deviation bound, we derive that 
\[
 \mathbb{E}^{\omega}\left[\exp\left( \frac{n C(r)}{p}  (p \max_{i}W_i^2)\right)\right]\leq 1+  C'(r) \frac{n}{p} \ ,
\]
where we used that $p$ is large compared to $n$. We also derive from the above deviation inequality that $\omega[\cA^c]\leq 4re^{-C(r)p}$. Together with \eqref{eq:upper_Lt_X}, this concludes the proof.
\end{proof}

\subsubsection{Control of $ \mathbf{E}^{\bX}_0[ \|\mathbf{P}^{Y|\bX}_{0} - \mathbf{P}^{Y|\bX}_{1}\|_{TV}]$ (Proof of Lemma \ref{lem:upper_total_variation_sigma_Y})}

Let us first characterize the conditional distributions  $\mathbf{P}^{Y|\bX}_{0}$ and $\mathbf{P}^{Y|\bX}_{1}$.

\begin{lem}[Distribution of  $Y$ conditionally to  $\bX$ under $\mathbf{P}_0$ and $\mathbf{P}_1$]\label{lem:distribution_P0_P1}
 
Define the matrices  
\beq \label{eq:def_bB}
\bB_0= \sum_{i=1}^r \gamma^2_{i,0}(p\bI_p + \alpha_{i,0}\bX^T\bX)^{-1} \ ,\quad \quad 
\bB_1= \sum_{i=0}^r \gamma^2_{i,1}(p\bI_p + \alpha_{i,1}\bX^T\bX)^{-1}\ .
\eeq 
and for $j=0,1$,
\beq
\bGamma_j:= \frac{1}{\sigma_j^2}\left[\bI_n- \frac{1}{\sigma_j^2}\bX (\bB_j^{-1}+ \bX^T\bX/\sigma_j^2)^{-1}\bX^T \right]
\eeq
 Under $\mathbf{P}_0$ (resp. $\mathbf{P}_1$),   $Y$ follows, conditionally to $\bX$, a centered normal distribution with precision matrix $\bGamma_0$ (resp. $\bGamma_1$).
\end{lem}

The precision matrices $\bGamma_0$ and $\bGamma_1$ are both diagonalizable in the same basis that diagonalizes $\bX\bX^T$. Denoting $\lambda_i$, $i=1,\ldots,n$, the ordered eigenvalues of $\bX\bX^T/p$, we define 
\beq\label{eq:definition_h_j}
h_j(\lambda_i) = \left(\sigma_j^2 + \sum_{l} \frac{\gamma^2_{l,j}\lambda_i}{1+\alpha_{l,j}\lambda_i}\right)^{-1},\;j=0,1
\eeq
the corresponding eigenvalues of $\bGamma_0$ and $\bGamma_1$. 

Suppose that  $\lambda_i$ lies in $(1/2,3/2)$ (this occurs with high probability). Since, by \eqref{eq:relation_1}, $\sigma_j^2+\sum_{l} \frac{\gamma^2_{l,j}}{1+\alpha_{l,j}}=2$, we have 
\beq\label{eq:control_hj}
1/4\leq h_j(\lambda_i) \leq 4
\eeq

Let us develop the Taylor's expansion of $1/h_j(\lambda_i)$ with respect to $(\lambda_i -1)$:
\beq \label{eq:decomposition_eta_j}
\frac{1}{h_j(\lambda_i) }= \sigma_j^2 +  \sum_{l} \frac{\gamma^2_{l,j}}{1+\alpha_{l,j}}+ \sum_{m=1}^{+\infty}(\lambda_i-1)^m(-1)^{m-1}\sum_l \frac{\gamma^2_{l,j}\alpha^{m-1}_{l,j}}{(1+\alpha_{l,j})^{m+1}}\ . 
\eeq
By the definition \eqref{eq:choice_sigma} of $\sigma_j^2$ and the property \eqref{eq:relation_1}, the constant term $\sigma_j^2 +  \sum_{l} \frac{\gamma^2_{l,j}}{1+\alpha_{l,j}}$ equals $2$. Now consider any $m\in \{1,\ldots ,2r-1\}$.
The rational function $x^{m-1}/(1+x)^{m+1}$ decomposes as a linear combination of $\tfrac{1}{(1+x)^{2}},\ldots \tfrac{1}{(1+x)^{m+1}}$. As a consequence , the term or order $(\lambda_i-1)^m$ in \eqref{eq:decomposition_eta_j} is a linear combination of the moments $m_q$ (defined in \eqref{eq:relation_2}) for $q=1,\ldots, m-1$. Thus, the choice of the parameters $\gamma_{l,j}$ and $\alpha_{l,j}$ makes all the terms of the development of $1/h_0(\lambda_i)$ and $1/h_1(\lambda_i)$ coincide up to order $2r-1$.
\beqn 
\Big| \frac{1}{h_1(\lambda_i)}-\frac{1}{h_0(\lambda_i)}\Big|&=& \Big|\sum_{m=2r}^{\infty}(\lambda_i-1)^m(-1)^{m} \left[\sum_l \frac{\gamma^2_{l,0}\alpha^{m-1}_{l,0}}{(1+\alpha_{l,0})^{m+1}} - \frac{\gamma^2_{l,1}\alpha^{m-1}_{l,1}}{(1+\alpha_{l,1})^{m+1}}\right]\Big|\\
&\leq & \sum_{m=2r}^{\infty} \big|\lambda_{i}-1\big|^{m}\left[\sum_l \frac{\gamma^2_{l,0}}{(1+\alpha_{l,0})^{2}} + \sum_{i}\frac{\gamma^2_{l,1}}{(1+\alpha_{l,1})^{2}}\right]\\
&\leq & \big|\lambda_{i}-1\big|^{2r}\ ,
\eeqn 
where we used in the last line that $|\lambda_i-1|\leq 1/2$ and $m_1=1/4$   (with all $m_1$ defined in \eqref{eq:relation_2}).
Suppose that all the eigenvalues $\lambda_i$ lie in $(1/2,3/2)$ so that  $\tfrac{h_0(\lambda_i)}{h_1(\lambda_i)}$ belongs to $(1/16,16)$ by \eqref{eq:control_hj}.
\beqn 
2 \|\mathbf{P}^{Y|\bX}_{0} - \mathbf{P}^{Y|\bX}_{1}\|^2_{TV}&\leq& \bbK[\mathbf{P}^{Y|\bX}_{0};\mathbf{P}^{Y|\bX}_{1}] \quad \quad \text{(by Pinsker's inequality)}\\
&  & \quad    = \frac{1}{2} \sum_{i=1}^n \left( \frac{h_1(\lambda_i)}{h_0(\lambda_i)}-1 -  \log\left[\frac{h_1(\lambda_i)}{h_0(\lambda_i)}\right]\right)\\
&\leq & C\sum_{i=1}^n h^2_1(\lambda_1)\left(\frac{1}{h_0(\lambda_i)} - \frac{1}{h_1(\lambda_i)}\right)^{2}\\
&\leq&  C'\sum_{i=1}^n(\lambda_i-1)^{4r} \leq C' n \|\bX\bX^T/p -\bI_n\|_{op}^{4r} 
\eeqn 
since, for some fixed $C>0$,  $x-1-\log(x)\leq C(x-1)^2$ for $x\in (1/16,16) $. The total variation distance is always smaller than one, so that
\beqn 
\mathbf{E}^{\bX}_0\left[ \|\mathbf{P}^{Y|\bX}_{0} - \mathbf{P}^{Y|\bX}_{1}\|_{TV}\right]&\leq& \mathbf{P}_0^{\bX}\left[\|\bX \bX^T/p-\bI_n\|_{op}\geq 1/2\right]+2^{-1/2} \mathbf{E}^{\bX}_0\left[\sqrt{\bbK[\mathbf{P}^{Y|\bX}_{0};\mathbf{P}^{Y|\bX}_{1}]}\mathbf{1}_{\|\bX \bX^T/p-\bI_n\|_{op}\leq 1/2}\right]\\
&\leq& \mathbf{P}_0^{\bX}\left[\|\bX \bX^T/p-\bI_n\|_{op}\geq 1/2\right]+C n^{1/2} \mathbf{E}^{\bX}_0 \left[\|\bX\bX^T/p -\bI_n\|_{op}^{2r} \right] 
\eeqn 
Conditionally to $v=(v_1,\ldots, v_r)$, $\bX$ follows a Gaussian distribution with inverse covariance $\bSigma^{-1}(v)= \bI_p+ \sum_{i=1}^r \alpha_{i,0} v_iv_i^T$ with $\alpha_{i,0}>0$. As a consequence, $\|\bSigma(v)\|_{op}= 1$ (since $p>r$) and $tr(\bSigma(v))$ belongs to $[p-r,p]$. As a consequence, we may apply the deviation inequality for Wishart matrices with non-identity covariances (Lemma \ref{lemma_concentration_spectre_A}) to $\bX$ and reintegrate with respect to $v$.
\beq \label{eq:concentration_X_P0X}
\mathbf{P}_0^{\bX}\left[\|\bX\bX^T -tr(\bSigma(v))\bI_n\|_{op}\leq       2\sqrt{p}\left(\sqrt{n}  + 10+\sqrt{2t}\right)+ 3 \left[n + 100+ 2t\right]\right]\geq 1 - 2e^{-t}\ ,
\eeq
for any $t>0$. Since $n\leq p$, this simplifies in
\[
 \mathbf{P}_0^{\bX}\left[\|\bX\bX^T/p -\bI_n\|_{op}\leq \frac{r}{p}+ C\left\{\sqrt{\frac{n}{p}}+ \sqrt{\frac{t}{p}}+ \frac{t}{p}\right\}\right]\geq 1 - 2e^{-t} \ .
\]
Integrating this deviation bound, we obtain  $\mathbf{E}^{\bX}_0 \left[\|\bX\bX^T/p -\bI_n\|^{2r} \right]\leq C(r) (n/p)^{r}$ where the constant $C(r)$ only depends on the integer $r$. Also, from \eqref{eq:concentration_X_P0X}, we derive that the probability that $\|\bX \bX^T/p-\bI_n\|\geq 1/2$ is smaller than $2 e^{-C'p}$ for some $C'>0$. We have proved that 
\[
\mathbf{E}^{\bX}_0\left[ \|\mathbf{P}^{Y|\bX}_{0} - \mathbf{P}^{Y|\bX}_{1}\|_{TV}\right]\leq 2e^{-C'p}+ C(r) \left(\frac{n^{1+ 1/(2r)}}{p} \right)^{r}\ ,
\]
where the constant $C(r)$ only depends on $r>0$.

\begin{proof}[Proof of Lemma \ref{lem:distribution_P0_P1}]
 We only prove the result for $\mathbf{P}_1$, the result for $\mathbf{P}_0$ being handled similarly. Since  $\mathbf{P}_1$ is a mixture distribution, we introduce  $f_1(Y,\bX,v_0,\ldots, v_r)$  the total density of $(Y,\bX,v)$ where $v=(v_0,\ldots,v_r)$.
\[
f_1(Y,\bX,v) \propto \exp\left[-\frac{\|Y- \bX\sum_{i=0}^r\gamma_{i,1}  v_{i}\|^2}{2\sigma_1^2} - \frac{tr(\bX \bX^T)}{2} - \sum_{i=0}^r \frac{\alpha_{i,1}}{2} \|\bX v_i\|_2^2 - \sum_{i=0}^r \frac{p}{2}\|v_i\|_2^2\right]
\]
Denote $z=\sum_{i=0}^r\gamma_{i,1} v_{i}$. If the density of $v$ is proportional to $\prod_i  \exp\left[-\sum_{i=0}^r \frac{\alpha_{i,1}}{2} \|\bX v_i\|_2^2 - \sum_{i=0}^r \frac{p}{2}\|v_i\|_2^2\right]$, then $z$ follows a centered normal distribution with covariance matrix  $\bB_1= \sum_{i=0}^r \gamma^2_{i,1}(p\bI_p + \alpha_{i,1}\bX^T\bX)^{-1}\ .$ As a consequence, integrating $f_1(Y,\bX,v)$ with respect to $v$ leads to the marginal density
\beqn 
 f_1(Y,\bX)&\propto& \int\tfrac{1}{(2\pi)^{n/2}|\bB_1|^{1/2}} e^{-\|Y- \bX z\|^2/(2\sigma_1^2) - z^T \bB_1^{-1}z/2}e^{- tr(\bX \bX^T)/2}dz\\
 &\propto& e^{-\frac{1}{2\sigma_1^2}Y^T \left(\bI_n- \frac{1}{\sigma_1^2}\bX (B_1^{-1}+ \bX^T\bX/\sigma_1^2)^{-1}\bX^T\right)Y}e^{- tr(\bX \bX^T)/2}\big[|\bB_1||\bB_1^{-1}+\tfrac{\bX^T\bX}{\sigma_1^2}|\big]^{-1/2}
\eeqn 
Hence, conditionally to $\bX$, $Y$ follows a centered normal distribution with precision matrix $\bGamma_1$.

\end{proof}

\subsubsection{Proof of Lemma \ref{lem:concentration_eta}}

We prove the result for $\mu_0$, the proof for $\mu_1$  being handled similarly. In order to ease the notation we respectively write $\alpha_{i}$, $\gamma_i$, $v_i$ and $\sigma$ for  $\alpha_{i,0}$, $\gamma_{i,0}$, $v_{i,0}$ and $\sigma_0$. From \eqref{eq:relation_1} and the definition \eqref{eq:definition_eta}, we have the following decomposition
\[\eta_0=\frac{\sum_{i=1}^{r}\tfrac{\gamma_{i}^2}{1+\alpha_{i}}  }{\sigma^2 + \sum_{i=1}^{r}\tfrac{\gamma_{i}^2}{1+\alpha_{i}} }.\]
Define $s_0 = \sum_{i=1}^{r}\gamma_{i}^2/(1+ \alpha_{i})$ so that $\eta_0= s_0/(\sigma^2+s_0)$. Since $\sigma^2$ is fixed to $3/2$, it suffices to prove that $ \beta^T \bSigma \beta$ is concentrated around 
$s_0$ to obtain a concentration bound for $\eta(\beta,\sigma)$ around $\eta_0$.

\medskip 
We use a similar approach to that of the proof of Lemma \ref{lem:control_second_moment_marginal_X}.
Conditionally to $(v_{1},\ldots, v_{r})$, $\bSigma^{-1}= \bI_p+ \sum_{i=1}^r \alpha_{i} v_{i} v_{i}^T$. Define $w_i:=v_i/\|v_i\|_2$ the standardized version of $v_i$ and $t_i:= \|v_i\|_2^2$. Also define the Gram-Schmidt orthonormalization basis $(\nu_1,\ldots,\nu_r)$ obtained from $(w_1,\ldots, w_r)$.  Define the $r\times r$ matrix $\widetilde{\bSigma}$ which represents the restriction of $\bSigma$ in the orthonormal basis $(\nu_1,\ldots,\nu_r)$. Similarly, define the $r$-dimensional vector $\tilde{\beta}$ so that $s_0= \tilde{\beta}^T\widetilde{\bSigma}\tilde{\beta}$. Define the diagonal  matrix $\overline{\bSigma}$ by $\overline{\bSigma}_{l,l}= 1/(1+\alpha_lt_l)$ for all $l=1,\ldots,r$. and the vector $\overline{\beta}$ by $\overline{\beta}_l=\gamma_l$ for $l=1,\ldots, r$. Then, $\beta^T\bSigma\beta - s_0$ decomposes as
\begin{eqnarray}
\beta^T\bSigma\beta - s_0&=& \left(\overline{\beta}^T \overline{\bSigma}\ \overline{\beta} - s_0\right)+ \overline{\beta}^T\left[ \overline{\bSigma}- \widetilde{\bSigma}\right] \overline{\beta} +  \left( tr\left[\widetilde{\bSigma}  \left\{\widetilde{\beta}\widetilde{\beta}^T- \overline{\beta}\ \overline{\beta}^T\right\}\right]\right) \nonumber\\
&=& (I) + (II) + (III)\ .\label{eq:decomposition_theta_0}
\end{eqnarray}
We shall prove that  each of these three terms is small in absolute value. Recall that the  $\alpha_i$ and $\gamma_i$ are positive constants only depending on $r$.
\beqn 
|(I)|&= &\Big|\sum_{i=1}^r \frac{\gamma_i^2}{1+\alpha_it_i}- \frac{\gamma_i^2}{1+\alpha_i}\Big|\leq r\|\gamma\|_\infty^2\|\alpha\|_{\infty} \max_i |t_i -1|\\
|(II)|&\leq &\|\overline{\beta}\|_2^2\|\overline{\bSigma}- \widetilde{\bSigma}\|_{op}\leq \|\gamma\|_2^2\|\overline{\bSigma}- \widetilde{\bSigma}\|_{op} \\
|(III)|&\leq & \|\widetilde{\bSigma}\|_F \|\widetilde{\beta}\widetilde{\beta}^T- \overline{\beta}\ \overline{\beta}^T\|_F\leq \|\widetilde{\bSigma}\|_F(\|\widetilde{\beta}\|_2+ \|\overline{\beta}\|_2)  \|\widetilde{\beta}- \overline{\beta}\|_2\leq C(r) \|\widetilde{\beta}- \overline{\beta}\|_2\ ,
\eeqn 
where we used in the last line that the eigenvalues of $\widetilde{\bSigma}$ are all smaller than one. Coming back to \eqref{eq:decomposition_theta_0}, we have proved that 
\beq\label{eq:decomposition_theta_02}
|\beta^T\bSigma\beta- s_0|\leq C(r) \left[\max_i |t_i -1| + \|\overline{\bSigma}- \widetilde{\bSigma}\|_{op} + \|\widetilde{\beta}- \overline{\beta}\|_2\right]\ .
\eeq
Let us bound $\|\overline{\bSigma}- \widetilde{\bSigma}\|_{op}$ and  $\|\widetilde{\beta}- \overline{\beta}\|_2$ in terms of $t$ and $w$. By definition of the basis $(\nu_1,\ldots ,\nu_r)$,
\beqn 
\big|[ \widetilde{\bSigma}^{-1}-\overline{\bSigma}^{-1}]_{l,m}\big|&=& \Big|\sum_{i=1}^r  \alpha_it_i  \langle  w_i, \nu_l\rangle \langle  w_i, \nu_m\rangle - \alpha_l t_l \mathbf{1}_{l=m}\Big|\\
&\leq & r\|\alpha\|_{\infty}\|t\|_{\infty} \max_{i=1,\ldots,r}\max_{l\neq i}|\langle w_i, \nu_l\rangle| \\
&\leq & r\|\alpha\|_{\infty}\|t\|_{\infty} \max_{i=1,\ldots,r}W_i\ ,
\eeqn 
where $W_i:= \|\Pi_{\mathrm{Vect}(w_1,\ldots,w_{i-1},w_{i+1},\ldots w_r)}w_i\|_2$ and $\Pi_S$ the orthogonal projection onto the space $S$. Here we used that for all $i,l$, $|\langle w_i, \nu_l\rangle|\leq 1$ and that $1-\langle w_l, \nu_l\rangle^{2}=\sum_{i\neq l} |\langle w_i, \nu_l\rangle|^{2}$.\\
Assuming that $2r\| \widetilde{\bSigma}^{-1}-\overline{\bSigma}^{-1}\|_{\infty}\leq 1$, we have $\| \widetilde{\bSigma}^{-1}-\overline{\bSigma}^{-1}\|_{op}\leq 1/2$,  so that also 
$\|\overline{\bSigma}^{1/2}(\widetilde{\bSigma}^{-1} - \overline{\bSigma}^{-1})\overline{\bSigma}^{1/2}\|_{op}\leq 1/2$ since all the eigenvalues of $\overline{\bSigma}$ are smaller than one. 
Then
\beqn 
\|\widetilde{\bSigma} - \overline{\bSigma}\|_{op}& = &\big\|\overline{\bSigma}^{1/2}\big[\big(\bI_{r}+ \overline{\bSigma}^{1/2}(\widetilde{\bSigma}^{-1} - \overline{\bSigma}^{-1})\overline{\bSigma}^{1/2}\big)^{-1} - \bI_{r}\big]\overline{\bSigma}^{1/2}\big\|_{op}\\
&\leq & \big\|\big(\bI_{r}+ \overline{\bSigma}^{1/2}(\widetilde{\bSigma}^{-1} - \overline{\bSigma}^{-1})\overline{\bSigma}^{1/2}\big)^{-1} - \bI_{r}\big\|_{op}\\
&\leq & \frac{\|\overline{\bSigma}^{1/2}(\widetilde{\bSigma}^{-1} - \overline{\bSigma}^{-1})\overline{\bSigma}^{1/2}\|_{op}}{1-\|\overline{\bSigma}^{1/2}(\widetilde{\bSigma}^{-1} - \overline{\bSigma}^{-1})\overline{\bSigma}^{1/2}\|_{op}}\\
&\leq & 2 \| \widetilde{\bSigma}^{-1}-\overline{\bSigma}^{-1}\|_{op}\\
&\leq & 2r \| \widetilde{\bSigma}^{-1}-\overline{\bSigma}^{-1}\|_{\infty}\ , 
\eeqn
where we used that all the eigenvalues of $\overline{\bSigma}$ are smaller than one. Turning to the difference $\widetilde{\beta}-\overline{\beta}$, we have, for any $l=1,\ldots, r$,
\[
 \big|(\widetilde{\beta}-\overline{\beta})_{l}\big|= \big|\sum_{i=1}^r \gamma_i \langle w_i,\nu_l\rangle - \gamma_l\big|\leq r \|\gamma\|_{\infty}\max_i W_i\ .
\]
Thus, we obtain $\|\widetilde{\beta}-\overline{\beta}\|_2\leq C(r) \max_i W_i$. Together with \eqref{eq:decomposition_theta_02}, this gives us 
\beq\label{eq:decomposition_theta_03}
|\beta^T\bSigma\beta- s_0|\leq C(r) \left[\max_i |t_i -1| + \|t\|_{\infty}\max_{i}W_i\right]\ , 
\eeq
as soon as $\|t\|_{\infty}\max_{i}W_i\leq 1/2r^{2}\|\alpha\|_{\infty}$. 
The deviations of $\max_i |t_i -1|$ are given by Lemma  \ref{lem:concentration_v} so that it only remains to control the deviations of $W_i$.
Let $\cA$ be any event on $w=(w_1,\ldots,w_r)$.
From \eqref{eq:lower_density_g}, we derive that 
\beqn 
 \mu_{0}(\cA)=\int \mathbf{1}_{\cA}g(t,w)\Pi_idt_i d\omega(w_i)\leq C'(r) p^{r/2}   (1+\|\alpha\|_{\infty}r)^{nr/2}\int\mathbf{1}_{\cA}\Pi_i d\omega(w_i)\ . 
\eeqn 
As a consequence, the probability $\mu_{0}(\cA)$ is always smaller than $C'(r) p^{r/2}   (1+\|\alpha\|_{\infty}r)^{nr/2}$ than the probability of $\cA$, when $w_1,\ldots, w_r$ are independently and uniformly distributed on the unit sphere.
When $w_1,\ldots, w_r$ are independently and uniformly distributed on the unit sphere,
 $W^2_i$ follows the same distribution as $\sum_{i=1}^{r-1}Z_i^2/\|Z\|_2^2$ where $Z=(Z_1,\ldots, Z_p)\sim \cN(0,\bI_p)$ (since the Gaussian distribution is isotropic). Arguing as in the proof of Lemma \ref{lem:control_second_moment_marginal_X}, we derive that  for any $t\in (0,p)$
\beqn 
\omega\left[pW_i^2\geq t+2r\right]&= &\P\left[\frac{p\sum_{i=1}^{r-1}Z_i^2}{\|Z\|_2^2}\geq t+2r\right]\leq \P\left[\|Z\|_2^2\leq p/2\right]+ \P\left[\sum_{i=1}^{r-1}(Z_i^2-1) \geq \frac{t}{2}\right]\\
&\leq & e^{-p/16}+ e^{-\tfrac{t}{8}\wedge \tfrac{t^2}{64r} }\leq 2e^{-C(r)t} \ , 
\eeqn 
where we used Lemma \ref{lem:chi_2} in the last line. Taking an union bound,  we derive that 
\[
\mu_{0}\left[p\max_{i}W_i^2\geq t+2r\right] \leq C'(r) p^{1+r/2}   (1+\|\alpha\|_{\infty}r)^{nr/2}e^{-C(r)t}\ .
\]
Thus, 
\begin{equation}
\label{eq:WWi}
\mu_{0}\left[\|t\|_{\infty}\max_{i}W_i \geq p^{-1/4}\right] \leq \mu_{0}\left[\|t\|_{\infty} \geq 2 \right] + C'(r) p^{1+r/2}   (1+\|\alpha\|_{\infty}r)^{nr/2}e^{-C''(r)p^{1/2}}
\end{equation}
since $p$ is large compared to $r$. Together with \eqref{eq:decomposition_theta_03} and Lemma  \ref{lem:concentration_v}, this gives us 
\[
 \mu_0\left[|\beta^T\bSigma\beta- s_0|\geq \tilde{C}(r)\big(p^{-1/4}+ (\frac{n}{p})^{1/2}\big)\right]\leq e^{-\tilde{C'}(r)p^{1/2}}\ . 
\]
for $p$ large enough compared to $n$  and  $r$. This last deviation inequality easily transfers to that of $\eta_0$.
\\
Let us now turn to the spectrum of $\bSigma$. By definition of $\widetilde{\bSigma}$,
\[
\lambda_{\min}(\bSigma)=   \lambda_{\min}(\widetilde{\bSigma})\geq \lambda_{\min}(\overline{\bSigma}) - \|\widetilde{\bSigma} - \overline{\bSigma}\|_{op}
\geq \min \frac{1}{1+\alpha_i} -  \|\widetilde{\bSigma} - \overline{\bSigma}\|_{op}
\]
so that $\lambda_{\min}(\bSigma) \geq \frac{1}{2}\min \frac{1}{1+\alpha_i}$ as soon as $2r^{2}\|\alpha\|_{\infty}\|t\|_{\infty} \max_{i=1,\ldots,r}W_i \leq \frac{1}{2}\min \frac{1}{1+\alpha_i}$ and the end of the lemma follows easily from (\ref{eq:WWi})  and Lemma  \ref{lem:concentration_v}.

\subsection{Proof of Proposition \ref{prp:lower_fixed_design}}

As in the previous minimax lower bounds, we use Le Cam's approach and build two mixture measures. 
Denote $\underline{\P}_0$ the distribution of $Y$ when $\beta^*=0$ and $\sigma=1$. Under $\underline{\P}_0$, $Y$ follows a standard normal distribution and $\eta[0,1,\bX]=0$. Given $\mu$ a continuous prior measure on $\mathbb
{R}^p$, we take 
\[\underline{\mathbf{P}}_{\mu}= \int_{\cB}\mathbb{P}_{\beta,0}\, \mu(d\beta)\]
Note that $\mu$-almost surely, $\eta[\beta,0,\bX]=1$. Recall that $\lambda_i^{1/2}$, $i=1,\ldots, n$ denote the singular values of $\bX$ and $v_i$, $i=1,\ldots, n$ its right eigenvectors. Let us choose $\mu$ such that, under $\mu$, $(\beta^Tv_i )\lambda_i^{1/2}$ follow independent standard normal distributions. Obviously, under $\underline{\mathbf{P}}_{\mu}$, $Y$ also follows a standard normal distribution, that is $\underline{\mathbf{P}}_{\mu}=\underline{\P}_0$.

\medskip
Consider any estimator $\widehat{\eta}$. Then, 
\beqn 
\sup_{\beta \in \mathbb{R}^p,\, \sigma\geq 0 }\underline{\E}_{\beta^*,\sigma}[(\widehat{\eta}-\eta[\beta^*,\sigma, \bX])^2]
&\geq &
 \underline{\E}_{0,1}\left[\widehat{\eta}^2\right]\bigvee \vee_{\beta\in \mathbb{R}^p}\underline{\E}_{\beta,0}\left[\left(\widehat{\eta}- 1\right]\right]\\
 &\geq & \underline{\E}_{0}\left[\widehat{\eta}^2\right]\bigvee \underline{\mathbf{E}}_{\mu}\left[(\widehat{\eta}-1)^2\right]\\
 &\geq& \underline{\E}_{0}\left[\widehat{\eta}^2\right]\bigvee \underline{\E}_{0}\left[\widehat{\eta}^2-1\right]\quad\quad \text{(since $\underline{\P}_0= \underline{\mathbf{P}}_{\mu}$)}\\
 &\geq& \frac{1}{2}\left[\underline{\E}_{0}\left[\widehat{\eta}^2\right]+ \underline{\E}_{0}\left[(\widehat{\eta}-1)^2\right]\right]\\
 &\geq& \frac{1}{2}+ \underline{\E}_{0}\left[\widehat{\eta}\right]^2 -   \underline{\E}_{0}\left[\widehat{\eta}\right] \quad \quad \text{(by Cauchy-Schwarz inequality)}\\
 &\geq & 1/4.
 \eeqn

 \appendix
 
\section{Auxiliary lemmas}

\begin{lem}[$\chi^2$ distributions~\cite{Laurent00}]\label{lem:chi_2}
Let $Z$ stands for a standard Gaussian vector of size $k$ and let $\bA$ be a symmetric matrix of size $k$. For any $t>0$, 
\[
 \P\left[Z^T \bA Z \geq tr(\bA)+ 2\|\bA\|_F\sqrt{t}+2 \|\bA\|_{op}t\right]\leq e^{-t}\ .
\]
When $\bA$ is the identity matrix, the above bound simplifies as 
\[
 \P\left[\chi^2(k)\geq k+2\sqrt{kt}+2t\right]\leq e^{-t}\ ,
\]
where $\chi^2(k)$ stand for a $\chi^2$-distributed random variable with $k$ degrees of freedom. We also have 
\[
 \P\left[\chi^2(k)\leq k-2\sqrt{kt}\right]\leq e^{-t}\ ,
\]
for any $t>0$.
\end{lem}

Laurent and Massart~\cite{Laurent00} have only stated a specific version of Lemma \ref{lem:chi_2} for positive matrices $\bA$, but their argument straightforwardly extend to general symmetric matrices $\bA$.

\begin{lem}[Wishart distributions~\cite{Davidson2001}]\label{lem:concentration_vp_wishart}
Let $\bZ$ be a $n\times d$  matrix whose entries follow independent standard normal distributions. For any positive number $x$, 
\beqn 
\mathbb{P}\left[\lambda_{\max}\left[\bZ^T\bZ\right] \geq n\left(1+\sqrt{d/n}+\sqrt{2x/n}\right)^2 \right] &\leq& \exp(-x)\\
\mathbb{P}\left[\lambda_{\min}\left[\bZ^T\bZ\right] \leq n\left(1-\sqrt{d/n}-\sqrt{2x/n}\right)^2 \right] &\leq& \exp(-x)
\eeqn
\end{lem}

\section{Proof of Lemma \ref{lemma_concentration_spectre_A}}
\label{appendixB}

Recall that $\bX\bX^T$ a weighted sum of Wishart matrices with parameters $(1,n)$
\[
\bX\bX^T= \bZ\bGamma\bZ^T=  \sum_{j=1}^p \bGamma_{jj} \left(\bZ_{\bullet j}\bZ_{\bullet j}^T\right)
\]
Define the matrix $\bU$ by 
\beq\label{eq:definition_U}
\bU := \bGamma^{1/2}\bZ^T
\eeq
The singular values of $\bU$ are the same as those of $\bX^T$. Denote $s_1(\bU)\geq s_2(\bU)\geq  \ldotsÂ \geq  s_n(\bU)$ the ordered singular values of $\bU$. From the previous remark, the following decomposition holds 
\[\lambda_{\max}(\bA)= s^2_1(\bU)-tr(\bSigma)\quad \text{ and }\quad \lambda_{\min}(\bA)= s^2_n(\bU)-tr(\bSigma)\ .\]
Hence,  it will suffice to derive deviation inequalities for both $s_1(\bU)$ and $s_n(\bU)$ to get the result \eqref{eq:concentration_spectre_A}.

\bigskip 
Denote $\mathbb{S}^{p-1}$ the $p$ dimensional unit sphere. Since $s_1(\bU)=\sup_{x\in \mathbb{S}^{p-1 }}\|\bU x\|_2$  and 
$s_n(\bU)=\inf_{x\in \mathbb{S}^{p-1 }}\|\bU x\|_2$, both $s_1(\bU)$ and $s_n(\bU)$ are Lipschitz (with respect to the Frobenius norm) functions  with constant 1 of the entries of $U$. 
As a consequence, $s_1(\bU)$ and $s_n(\bU)$ are Lipschitz functions with constant $\max_i(\bGamma^{1/2}_{i,i})=\|\bSigma\|_{op}^{1/2}$ of the entries of $\bZ$. Applying the Gaussian isoperimetric inequality \cite{book_concentration}, it follows that 
\begin{eqnarray}\label{eq:deviation_s1_sn}
\P\left[s_1(\bU)\geq \mathbb{E}\left[s_1(\bU)\right]+ \|\bSigma\|_{op}^{1/2}\sqrt{2t} \right]&\leq& \exp\left[-t\right]\\
\P\left[s_n(\bU)\leq \mathbb{E}\left[s_n(\bU)\right]- \|\bSigma\|_{op}^{1/2}\sqrt{2t} \right]&\leq& \exp\left[-t\right] \ . \nonumber
\end{eqnarray}

In order to control $\mathbb{E}\left[s_1(\bU)\right]$ and $\mathbb{E}\left[s_n(\bU)\right]$, we apply Gordon-Slepian lemma following the approach of Davidson and Szarek  \cite[Appendix IIc]{Davidson2001}. First, recall Gordon's extension of Slepian lemma.

\begin{lem}[Gordon-Slepian lemma]\label{lem:gordon} Let $(X_t)_{t\in T}$ and $(Y_t)_{t\in T}$ be two finite families of jointly Gaussian mean zero random variables such that $\var{X_t-X_t'}\leq \var{Y_t-Y_t'}$ for all $t,t'\in T$. Then $\E[\max_{t\in T}X_t]\leq \E[\max_{t\in T}Y_t]$. Similarly, if $T= \cup_{s\in S}T_s$ and 
\begin{eqnarray}
 \var{X_t-X_t'}\leq \var{Y_t-Y_t'} \text{ if }t\in T_s,\ t'\in T_{s'}\text{ with }s\neq s'\label{eq:gordon_cond1}\\
 \var{X_t-X_t'}\geq \var{Y_t-Y_t'} \text{ if }t,t'\in T_{s}\text{ for some }s
\label{eq:gordon_cond2}
\end{eqnarray}
one has $\E[\max_{s\in S}\min_{t\in T_s}X_t]\leq \E[\max_{s\in S}\min_{t\in T_s}Y_t]$ 
\end{lem}

Define the Gaussian process $P_{(u, v)}$ indexed by $(u,v)\in \mathbb{S}^{p-1}\times \mathbb{S}^{n-1}$,
\[
P_{(u, v)}:= \langle  u, \bU v\rangle = tr(\bZ(\bGamma^{1/2} u v^T))
\]
For any $(u,v)$ and $(u',v')$, this process satisfies 
\beqn 
 \var{P_{(u, v)}- P_{(u', v')}}&=& \|(\bGamma^{1/2} u)^T v- (\bGamma^{1/2} u')^T v' \|_F^2\\
 &= &\|\bGamma^{1/2}(u-u')\|_2^2 + u^T \bGamma u' \|v-v'\|_2^2 \\
 &\leq& \|\bGamma^{1/2}(u-u')\|_2^2 + \|\bSigma\|_{op} \|v-v'\|_2^2.
\eeqn 
Let $Z_1$ and $Z_2$ be two independent standard Gaussian vectors of respective size $p$ and $n$. For any $u\in \mathbb{R}^{p}$ and any $v\in \mathbb{R}^n$, define 
\[
Q_{(u,v)}:= (\bGamma^{1/2} u)^T Z_1 + \|\bSigma\|_{op}^{1/2} v^T Z_2
\]
Hence,
\beqn 
\var{Q_{(u,v)}- Q_{(u',v')}}&=&\|\bGamma^{1/2}(u-u')\|_2^2 + \|\bSigma\|_{op} \|v-v'\|_2^2\ . 
\eeqn 
We are therefore in position to apply Slepian lemma to the processes $P_{u,v}$ and $Q_{u,v}$ (although the set $\mathbb{S}^{p-1}\times \mathbb{S}^{n-1}$ is not finite, the result is still true). Observe that 
$\max_{(u,v)\in \mathbb{S}^{p-1}\times \mathbb{S}^{n-1}} P_{u,v}= s_1(\bU)$. It follows that
\begin{eqnarray}
\mathbb{E}\left[s_1(\bU)\right]&\leq& \mathbb{E}\left[\max_{(u,v)\in \mathbb{S}^{p-1}\times \mathbb{S}^{n-1}} Q_{(u,v)} \right]= \mathbb{E}\left[\|\bGamma^{1/2} Z_1\|_{2} \right]+ \|\bSigma\|_{op}^{1/2}\mathbb{E}\left[\|Z_2\|_{2}\right] \nonumber \\
&\leq & \sqrt{tr(\bGamma)}+ \|\bSigma\|_{op}^{1/2}\sqrt{n} \ ,  \label{eq:upper_E_s1}
\end{eqnarray}
by Cauchy-Schwarz inequality.
\medskip

For any $v\in \mathbb{S}^{n-1}$, define $T_{v}:= \{(u,v),\  u\in \mathbb{S}^{p-1}\}$. Hypothesis \eqref{eq:gordon_cond1}  is still satisfied for $P_{(u,v)}$ and $Q_{(u,v)}$. For $(u,v)$ and $(u',v)\in T_{v}$, 
\[ \var{P_{(u, v)}(\bZ)- P_{(u', v)}(\bZ)}= \|(\bGamma^{1/2}(u-u')\|_2^2=  \var{Q_{(u, v)}(\bZ)- Q_{(u', v)}(\bZ))}\ ,\]
and Hypothesis  \eqref{eq:gordon_cond2} is also satisfied. Applying Gordon-Slepian lemma, we obtain 
\begin{eqnarray}
- \mathbb{E}\left[s_n(\bU)\right]&=&\mathbb{E}\left[\max_{v\in \mathbb{S}^{n-1}}\min_{u\in T_v} P_{(u, v)} \right]  \nonumber\\
&\leq & \mathbb{E}\left[\max_{v\in \mathbb{S}^{n-1}}\min_{u\in T_v} Q_{(u, v)} \right]\nonumber \\
&\leq & \mathbb{E}\left[\|\bSigma\|_{op}^{1/2}\|Z_2\|_{2}-\|\bGamma^{1/2}Z_1\|_{2}\right] \label{eq:upper_E_sn}
\end{eqnarray}
By Cauchy-Schwarz inequality $\mathbb{E}\left[\|Z_2\|_{2}\right]\leq \sqrt{n}$. It remains to lower bound $\E[\|\bGamma^{1/2}Z_1\|_{2}]$. Denote $V= \|\bGamma^{1/2}Z_1\|_{2}$. By isoperimetric Gaussian inequality $\P\left[V\geq\E[V]+ \|\bSigma\|_{op}^{1/2}\sqrt{2t}\right]\leq e^{-t}$ for any $t>0$. Squaring the above inequality, it follows that for any $t>0$,
\[\P\left[\tfrac{V^2-\E^2[V]}{2\|\bSigma\|_{op}}\geq t\right]\leq \exp\left[- \frac{\|\bSigma\|_{op}t^2}{2\E^2[V]}\wedge t\right]
.\]
Integrating this bound with respect to $t>0$, we obtain 
\[\frac{\E[V^2]- \E^2[V]}{2\|\bSigma\|_{op}}\leq \int_{0}^{\infty}e^{-\tfrac{\|\bSigma\|_{op}t^2}{4\E^2[V]}}dt+ \int_{0}^{\infty}e^{-t}dt\leq   \sqrt{\frac{\pi\E^2[V]}{2\|\bSigma\|_{op}}}+  1
\leq \sqrt{\frac{\pi\E[V^2]}{2\|\bSigma\|_{op}}}+  1 \]
by Cauchy-Schwarz inequality, which implies 
\beqn 
\E[V]&\geq& \sqrt{\left(\E[V^2]- \sqrt{8\pi\|\bSigma\|_{op}\E[V^2]} - 4\|\bSigma\|_{op}\right)_+ }\\
&\geq & \sqrt{\E[V^2]}\left[ 1-\sqrt{\tfrac{8\pi \|\bSigma\|_{op}}{\E[V^2]}}- 4\tfrac{\|\bSigma\|_{op}}{\E[V^2]}\right]\\
&\geq & \sqrt{\E[V^2]} -\sqrt{8\pi \|\bSigma\|_{op}} -4 \|\bSigma\|_{op}/ \sqrt{\E[V^2]}\ , 
\eeqn 
where we used $\sqrt{1-x}\geq 1-x$ for all $x\in (0,1)$ in the second line. Since $\E[V^2]= tr(\bSigma)$, we conclude that 
\[
 \E[V]\geq \sqrt{tr(\bSigma)}- \|\bSigma\|_{op}^{1/2}\pi^{1/2}8^{1/2}- 4\frac{\|\bSigma\|_{op}}{\sqrt{tr(\bSigma)}}\geq\sqrt{tr(\bSigma)}- \|\bSigma\|_{op}^{1/2}(\sqrt{8\pi}+4)  \ . 
\]
Gathering this bound together with \eqref{eq:deviation_s1_sn}, \eqref{eq:upper_E_s1}, and \eqref{eq:upper_E_sn}, we obtain
\beqn 
\P\left[s_1(\bU)\geq \sqrt{tr(\bSigma)}+\|\bSigma\|_{op}^{1/2}\sqrt{n}+ \|\bSigma\|_{op}^{1/2}\sqrt{2t}\right]&\leq &e^{-t}\\
\P\left[s_n(\bU)\leq \sqrt{tr(\bSigma)}- \|\bSigma\|_{op}^{1/2}\sqrt{n}-  \|\bSigma\|_{op}^{1/2}(10+\sqrt{2t})\right]&\leq& e^{-t}
\eeqn 
Recalling that $s_1^2(\bU)=\lambda_{\max}(\bX \bX^T)$ and $s_n^2(\bU)=\lambda_{\min}(\bX \bX^T)$ concludes the proof.

\bibliography{biblio}
\bibliographystyle{plain}

\end{document}